\documentclass[1p,number,11pt]{elsarticle} 
\makeatletter
\renewcommand\subsection{\@startsection{subsection}{2}{\z@}%
  {-18\p@ \@plus -4\p@ \@minus -4\p@}%
  {8\p@ \@plus 4\p@ \@minus 4\p@}%
  {\normalfont\normalsize\bfseries}}
\makeatother

\usepackage{amsfonts}
\usepackage{amsmath}
\usepackage{amssymb}
\usepackage{amsthm}
\usepackage[english]{babel}
\usepackage{complexity}
\usepackage{enumerate}
\usepackage{mathrsfs}
\usepackage{graphics}
\usepackage[pagewise]{lineno}
\usepackage{mathtools}
\usepackage{microtype}
\usepackage[many]{tcolorbox}
\usepackage{xcolor}
\usepackage{xspace}
\usepackage[normalem]{ulem}
\usepackage{zref-clever}
\usepackage{slantsc} 
\zcsetup{cap}

\usepackage[store-sets-label,thmtools-compat]{keytheorems}
\usepackage{multicol}

\newlength{\remparindent}
\AtBeginDocument{\setlength{\remparindent}{\parindent}}
\newtcolorbox{remove?}{
  enhanced, breakable,
  colback=red!20,
  frame hidden, boxrule=0pt, sharp corners,
  left=0pt, right=0pt, top=4pt, bottom=4pt, boxsep=0pt,
  grow to left by=4pt, grow to right by=4pt,
  before skip=\parskip, after skip=\parskip,
  parbox=false,
  before upper={\setlength{\parindent}{\remparindent}}
}

\newcommand{\ccon}{\hbox{\sf color\raisebox{0.4pt}{\texttt{-}}conn}}
\newcommand{\con}{\hbox{\sf conn}}
\newclass{\FOcc}{FO\hspace{1pt}\textup{\texttt{+}}\hspace{1pt}\ccon\xspace}
\newclass{\FOconn}{FO\hspace{1pt}\textup{\texttt{+}}\hspace{1pt}\con\xspace}

\newcommand{\scon}{\hbox{\sf subconn}}
\newcommand{\smallsets}{\exists_{|\cdot|\leq b}X}
\newclass{\FOscon}{FO\hspace{1pt}\textup{\texttt{+}}\hspace{1pt}\scon\xspace}
\newclass{\FOsconsets}{FO\hspace{1pt}\textup{\texttt{+}}\hspace{1pt}\smallsets\hspace{1pt}\textup{\texttt{+}}\hspace{1pt}\scon\xspace}

\newcommand{\GFOscon}{\textup{Guarded}\raisebox{0.4pt}{\texttt{-}}\FOscon}

\newcommand{\placeholder}[1]{\textup{\textbf{$\boldsymbol{\langle}$#1$\boldsymbol{\rangle}$}}}

\usetikzlibrary{calc,arrows, automata, fit, shapes, arrows.meta, positioning}

\newkeytheoremstyle{tc-theorem}{
tcolorbox-no-titlebar={
colframe=blue, sharp corners=west, 
    colback=blue!60!cyan!20!white,
	rightrule=0pt,
	toprule=0pt,
	bottomrule=0pt,
	top=3pt,
	right=2pt,
	bottom=3pt},
}
\newkeytheoremstyle{tc-defi}{
tcolorbox-no-titlebar={
    colframe=yellow, sharp corners=west, 
    colback=yellow!60!yellow!20!white,
	rightrule=0pt,
	toprule=0pt,
	bottomrule=0pt,
	top=3pt,
	right=2pt,
	bottom=3pt,
 }
}
\newkeytheoremstyle{tc-prop}{
tcolorbox-no-titlebar={
    colframe=blue, sharp corners=west, 
    colback=blue!50!cyan!20!white,
	rightrule=0pt,
	toprule=0pt,
	bottomrule=0pt,
	top=3pt,
	right=2pt,
	bottom=3pt,
 }
}
\newkeytheoremstyle{tc-cor}{
tcolorbox-no-titlebar={
    colframe=cyan, sharp corners=west, 
    colback=blue!30!cyan!10!white,
	rightrule=0pt,
	toprule=0pt,
	bottomrule=0pt,
	top=3pt,
	right=2pt,
	bottom=3pt,
 }
 }
\newkeytheoremstyle{tc-obs}{
tcolorbox-no-titlebar={
    colframe=cyan, sharp corners=west, 
    colback=blue!30!cyan!10!white,
	rightrule=0pt,
	toprule=0pt,
	bottomrule=0pt,
	top=3pt,
	right=2pt,
	bottom=3pt,
 }
}
\newkeytheoremstyle{tc-lem}{
tcolorbox-no-titlebar={
    colframe=cyan, sharp corners=west,
    colback=blue!40!cyan!20!white,
	rightrule=0pt,
	toprule=0pt,
	bottomrule=0pt,
	top=3pt,
	right=2pt,
	bottom=3pt,
 }
}
\newkeytheoremstyle{tc-res}{
tcolorbox-no-titlebar={
    colframe=yellow, sharp corners=west, 
    colback=yellow!10,
	rightrule=0pt,
	toprule=0pt,
	bottomrule=0pt,
	top=8pt,
    left=20pt,
	right=20pt,
	bottom=10pt,
 }
}

\usepackage{hyperref}

\usepackage{enumitem}
\setlist[itemize]{topsep=4pt,itemsep=3pt,parsep=0pt} 
\setlist[enumerate]{topsep=4pt,itemsep=3pt,parsep=0pt} 

\usepackage{mathrsfs}

\usepackage{amsmath}
\usepackage{amssymb}
\usepackage{amsthm}

\usepackage{float}
\usepackage{wrapfig}
\usepackage[all]{xy}
\usepackage{tikz}
\usepackage{cancel}
\usepackage{mdframed}
\usepackage{xspace}
\usepackage[textsize=normalsize]{todonotes}

\usepackage{caption}
\usepackage{subcaption}

\newcommand{\Oof}{\mathcal{O}}
\newcommand{\N}{\mathbb{N}}
\renewcommand{\phi}{\varphi}
\newcommand{\Gaif}{\operatorname{Gaif}}

\theoremstyle{plain}
\newtheorem{claim}{Claim}
\newenvironment{claimproof}{ \trivlist
	\item[\hskip\labelsep
	\emph{Proof of the claim}.]\ignorespaces
}{\hfill$\vartriangleleft$\vspace{10pt}\par}

\usepackage[noabbrev,capitalise,nameinlink]{cleveref}
\hypersetup{
    colorlinks,
    linkcolor={blue!80!black},
    citecolor={blue!80!black},
    urlcolor={blue!80!black}
}

\zcRefTypeSetup{fact}{
Name-sg = Fact ,
name-sg = fact ,
Name-pl = Facts ,
name-pl = facts ,
}

\crefname{observation}{Observation}{Observation}
\crefalias{enumi}{observation}
\crefname{claim}{Claim}{Claim}

\newcommand{\Cc}{\mathscr{C}}

\newcommand{\Hh}{\mathscr{H}}
\newcommand{\Mm}{\mathscr{M}}
\newcommand{\Dd}{\mathscr{D}}

\newcommand{\conn}{\mathsf{conn}}

\newcommand{\bag}{\mathsf{bag}}

\newcommand{\comp}{\mathsf{comp}}
\newcommand{\mrg}{\mathsf{mrg}}
\newcommand{\adh}{\mathsf{adh}}

\newcommand{\torso}{\mathsf{torso}}
\newcommand{\parent}{\mathsf{parent}}

\newcommand{\MZ}{\mathbf M\!\upharpoonright_Z}

\renewcommand{\FO}{\ensuremath{\mathsf{FO}}}

\newcommand{\MSO}{\ensuremath{\mathsf{MSO}}}

\renewcommand{\FPT}{\ensuremath{\mathsf{FPT}}}

\renewcommand{\leq}{\leqslant}
\renewcommand{\geq}{\geqslant}
\renewcommand{\le}{\leqslant}

\newcommand{\pred}{\mathbin{\prec:}}
\newcommand{\is}{\mathrel{:=}}

\DeclareMathOperator{\dec}{{\rm Dec}}
\DeclareMathOperator{\rew}{{\sf Rewire}}

\newcommand{\forkindep}[1]{%
  \mathrel{
    \mathop{
      \vcenter{
        \hbox{\oalign{\noalign{\kern-.3ex}\hfil$\vert$\hfil\cr
              \noalign{\kern-.7ex}
              $\smile$\cr\noalign{\kern-.3ex}}}
      }
    }\displaylimits_{#1}
  }
}

\newkeytheorem{lemma}[style=tc-lem,parent=section]
\newkeytheorem{theorem}[style=tc-theorem,sibling=lemma]
\newkeytheorem{corollary}[style=tc-cor,sibling=lemma]
\newkeytheorem{fact}[sibling=lemma]
\newkeytheorem{proposition}[style=tc-prop,sibling=lemma]
\newkeytheorem{remark}[style=remark,sibling=lemma]

\newkeytheorem{definition}[sibling=lemma]
\newkeytheorem{importantdefi}[style=tc-defi,name=Definition,sibling=lemma]
\newkeytheorem{observation}[style=tc-obs,sibling=lemma]

\newkeytheorem{problem}[style=tc-res]
\newkeytheorem{conjecture}[style=tc-res,sibling=problem]
\usepackage{color}

\definecolor{MidnightBlack}{rgb}{0.1,0.1,.34}
\definecolor{MidnightBlue}{rgb}{0.1,0.1,0.44}
\definecolor{Black}{rgb}{0,0, 0}
\definecolor{Blue}{rgb}{0, 0 ,1}
\definecolor{Red}{rgb}{1, 0 ,0}
\definecolor{White}{rgb}{1, 1, 1}
\definecolor{Grey}{rgb}{.6, .6, .6}
\definecolor{Mygreen}{rgb}{.0, .5, .0}
\definecolor{Yellow}{rgb}{.55,.55,0}
\definecolor{Mustard}{rgb}{1.0, 0.86, 0.35}
\definecolor{applegreen}{rgb}{0.55, 0.71, 0.0}
\definecolor{darkturquoise}{rgb}{0.0, 0.81, 0.82}
\definecolor{celestialblue}{rgb}{0.29, 0.59, 0.82}
\definecolor{green_yellow}{rgb}{0.68, 1.0, 0.18}
\definecolor{crimsonglory}{rgb}{0.75, 0.0, 0.2}
\definecolor{darkmagenta}{rgb}{0.30, 0.0, 0.30}
\definecolor{internationalorange}{rgb}{1.0, 0.31, 0.0}
\definecolor{internationalorange}{rgb}{1.0, 165/256, 0.0}
\definecolor{darkorange}{rgb}{1.0, 0.55, 0.0}
\definecolor{ao}{rgb}{0.0, 0.5, 0.0}
\definecolor{awesome}{rgb}{1.0, 0.13, 0.32}

\newcommand{\Minor}{\mathrm{Minor}}

\journal{arXiv}

\begin{document}

\begin{frontmatter}

\newcommand{\funding}{
This research is supported by the
French-German Collaboration ANR/DFG Project UTMA (ANR-20-CE92-0027, DFG 446200270) {\em ``Unifying Theories for Multivariate Algorithms''}.
NM received funding from 
the ERC grant {\sc BUKA} (No.~101126229) and
the European Union through an ERA Fellowship with grant agreement No.\ 101334340 -- {\sc LoCoMoDe}.
POdM has received funding from the European Research Council (ERC) under the European Union's Horizon 2020 research and innovation program (grant agreement No 810115 -- {\sc Dynasnet}) for this project.
DT was also supported by the French National Research Agency (ANR) under project GODASse ANR-24-CE48-4377 and under the France 2030 grant reference number ANR-24-RRII-0002 operated by the Inria Quadrant Program.
}

\title{A Logic for Minor-Free Graph Classes:\\ Model Checking, Dependence, and Combinatorial Reconfiguration}
\tnotetext[t1]{\funding
}

\author[1]{Nikolas Mählmann} 
\ead{maehlmann@mimuw.edu.pl}
\author[2]{Patrice Ossona de Mendez} 
\ead{pom@ehess.fr}
\author[3]{Nicole Schirrmacher}
\ead{schirrmacher@uni-bremen.de}
\author[3]{Sebastian Siebertz}
\ead{siebertz@uni-bremen.de}
\author[4]{Giannos Stamoulis}
\ead{stamoulis@irif.fr}
\author[5]{Dimitrios M. Thilikos}
\ead{sedthilk@thilikos.info}
\author[6]{Alexandre Vigny}
\ead{alexandre.vigny@uca.fr}

\affiliation[1]{organization={University of Warsaw}, country=Poland}
\affiliation[2]{organization={CAMS, EHESS/CNRS UMR 8557}, country=France}
\affiliation[3]{organization={University of Bremen}, country=Germany}
\affiliation[4]{organization={Université Paris Cité, CNRS, IRIF, F-75013, Paris}, country=France}
\affiliation[5]{organization={LIRMM, Univ Montpellier, CNRS}, country=France}
\affiliation[6]{organization={LIMOS, Univ Clermont Auvergne, Clermont Auvergne INP}, country=France}

\begin{abstract}
We introduce \emph{sub-connectivity logic}, denoted by $\FOscon$, an extension of first-order logic for graphs with a specified set $Z$ of admissible edges. Its additional atom $\scon(x,y;\bar z)$ asserts that $x$ and $y$ are connected by a path using only edges of $Z$ and avoiding the vertices in $\bar z$.

Our main structural result states that, for every weakly sparse graph class $\Cc$, the class of all expansions of graphs in $\Cc$ by rooted spanning forest orders has bounded twin-width if and only if $\Cc$ excludes a fixed minor. When $\Cc$ excludes a fixed minor, we can also compute an additional linear order that preserves bounded twin-width in polynomial time.

We translate $\FOscon$ into first-order logic over an expansion by a depth-first spanning forest order of the admissible-edge graph. Combining this translation with our structural result, we obtain a model checking algorithm for $\FOscon$ with running time $f(|\phi|,h(G))\cdot |G|^c$, where $f$ is computable, $c$ is an absolute constant, and $h(G)$ is the Hadwiger number of $G$. The translation and model-checking algorithm extend to binary relational structures, with the Hadwiger number measured on the Gaifman graph.

For a graph class $\Cc$ let
$\Cc_Z:=\{(G,Z):G\in\Cc,\ Z\subseteq E(G)\}$. We show that for every weakly sparse graph class~$\Cc$, the class~$\Cc_Z$ is monadically dependent for $\FOscon$ if and only if $\Cc$ excludes a fixed minor. 
For monotone classes that admit efficient minor encodings, a corresponding hardness result makes this frontier computationally tight. 
We give applications to combinatorial reconfiguration and solution discovery problems, and study a guarded extension of $\FOscon$ motivated by database queries.
\end{abstract}
\end{frontmatter}



\newpage 
\setcounter{tocdepth}{2}
\tableofcontents
\thispagestyle{empty}
\newpage
\setcounter{page}{1}

\section{Introduction}
Bridging algorithmic graph theory and finite model theory, a central line of research seeks efficient graph algorithms on well-behaved classes, together with logical and structural characterizations of the tractable problems and the underlying classes. 
The resulting statements are commonly expressed as \emph{algorithmic meta-theorems} of the following form.
\begin{quote}
\emph{Every problem expressible in \placeholder{logic} can be solved in time \placeholder{runtime} on every structure satisfying \placeholder{tameness condition}.}
\end{quote}
Two prototypical examples of this pattern are the following.
\begin{itemize}
    \item Model checking monadic second-order logic (\MSO) can be done in time $f(|\phi|,t)\cdot |G|$ on graphs $G$ of tree-width at most~$t$.
    This is known as Courcelle's theorem.~\cite{arnborg1991,courcelle1990,courcelle1993}

    \item Model checking first-order logic (\FO) can be done in time 
    $f_\Cc(|\phi|,\varepsilon) \cdot |G|^{1+\varepsilon}$ on all graphs $G$ from every fixed nowhere dense graph class $\Cc$.~\cite{grohe2017}
\end{itemize}
In both of these examples, the involved structural tameness condition coincides with the notion of \emph{monadic dependence}, a key notion from model theory that we will define later. For every graph class~$\Cc$ that is monotone (i.e., closed under taking subgraphs), we have that $\Cc$ is monadically dependent~\ldots
\begin{itemize}
    \item \ldots for \MSO~if and only if $\Cc$ has bounded tree-width; \cite{grohe2004learnability}
    \item \ldots for \FO~if and only if $\Cc$ is nowhere dense.~\cite{ADLER2014322}
\end{itemize}
In this paper, we investigate the connection between efficient model checking and monadic dependence for minor-free graph classes. 
As minor-free graph classes are strictly sandwiched between bounded tree-width and nowhere dense classes, this calls for a logic that is less expressive than \MSO~and more expressive than \FO.

We introduce an extension of first-order logic, called \emph{sub-connectivity logic} and denoted by $\FOscon$. 
Formulas are evaluated on a graph $G$ with a specified set $Z\subseteq E(G)$ of \emph{admissible edges}. In addition to the usual first-order atoms, the logic has an atom $\scon(x,y;\bar z)$ asserting that $x$ and $y$ are not among the vertices in $\bar z$ and are connected by a path using only edges of $Z$ and avoiding those vertices. 

We translate $\FOscon$ on $(G,Z)$ into first-order logic on an expansion by a forest order induced by a rooted depth-first spanning forest of the graph $(V(G),Z)$. 
Combining this translation with our structural results, we show that there are a computable function $f$ and an absolute constant~$c$ such that an $\FOscon$-sentence $\phi$ can be model-checked on $(G,Z)$ in time
\[
 f(|\phi|,h(G))\cdot |G|^c,
\]
where $h(G)$ is the \emph{Hadwiger number} of $G$, that is, the largest order of a clique minor of~$G$. 
Thus, $\FOscon$ model checking is fixed-parameter tractable on every graph class excluding a fixed minor.
The algorithm extends to binary relational structures $\mathbf M$, with running time $f(|\phi|,h(\Gaif(\mathbf M)))\cdot\|\mathbf M\|^c$, where $\Gaif(\mathbf M)$ is the Gaifman graph of $\mathbf M$. 

For a graph class $\Cc$, let $\Cc_Z:=\{(G,Z):G\in\Cc,\ Z\subseteq E(G)\}$. We call $\Cc$ \emph{weakly sparse} if it excludes some complete bipartite graph as a subgraph. 
For every weakly sparse class $\Cc$, we prove that~$\Cc_Z$ is monadically dependent for $\FOscon$ if and only if $\Cc$ excludes a fixed minor.

\pagebreak
\subsection{Context}
\paragraph{Model checking}
 
For a logic~$\mathcal L$ and a class~$\Cc$ of finite structures, the associated computational problem in the broader program of algorithmic meta-theorems is the \emph{model checking} problem: given
$\mathbf M\in\Cc$ and a sentence $\varphi\in\mathcal L$, decide whether
$\mathbf M\models\varphi$.
Reformulated from above, algorithmic meta-theorems assert that model checking for a logic is tractable\footnote{When we speak of tractable model checking on a class $\Cc$, we mean fixed-parameter tractability with respect to the formula size, that is, a running time of the form $f(|\phi|)\cdot \|\mathbf M\|^c$ for some computable function $f$ and an absolute constant $c$, on every $\mathbf M\in \Cc$ and formula $\phi$.
For the general sparsity and stability results recalled below, we assume effective bounds on the corresponding class properties, as required for computable parameter dependence.}  on classes satisfying suitable structural tameness conditions, thereby linking logical definability, (graph) structure, and efficient computation in a uniform way.
Classical milestones of this program include tractable
model checking for $\MSO$ on bounded tree-width graphs~\cite{arnborg1991,courcelle1990,courcelle1993} and for $\FO$ on sparse graph classes such as bounded degree, excluded minor, bounded expansion, and nowhere dense classes~\cite{seese1996,frick2001,dvorak2013,grohe2017}.
More recently, \emph{twin-width} has led to model-checking algorithms that also apply to dense graphs~\cite{bonnet2022twwI}. In particular, and particularly important for this paper, $\FO$ model checking is fixed-parameter tractable on classes of linearly ordered binary structures of bounded twin-width~\cite{bonnet2024twin}.
Beyond unifying large families of tractable problems, such results repeatedly lead to new structural viewpoints and algorithmic techniques.
It has therefore become one of the central problems in algorithmic model theory to determine the exact tractability frontier of efficient model checking for various logics.

\paragraph{The limits of tractability}
The limits of tractability for model checking are closely related to \emph{dependence}, one of the main dividing lines in classical model theory. 
A formula $\varphi(\bar x;\bar y)$ is \emph{independent on a class $\Cc$} if, for every $n\in\mathbb N$, there are a structure $\mathbf M\in\Cc$, tuples $\bar a_i\in M^{|\bar x|}$ for $i\in[n]$, and tuples $\bar b_J\in M^{|\bar y|}$ for $J\subseteq[n]$ such that
\[
 \mathbf M\models\varphi(\bar a_i;\bar b_J)
 \quad\Longleftrightarrow\quad i\in J.
\]
Intuitively, this means there are arbitrarily large sets, from which $\phi$ can pick out all the subsets. 
A formula is \emph{dependent on $\Cc$} if it is not independent on $\Cc$, and a class is \emph{dependent} for a logic if every formula of that logic is dependent on it. 
A class is \emph{monadically dependent} if every expansion of the class by finitely many arbitrary unary predicates is dependent. 

For \emph{hereditary classes} (i.e., classes closed under taking induced substructures), the notion of monadic dependence for {\FO} is known to 
imply fixed-parameter tractable $\FO$ model checking for: 
\begin{multicols}{2}
\begin{itemize}
    \item weakly sparse classes 
of graphs~\cite{dvorak2013,grohe2017},

\item edge-stable\footnote{A graph class is \emph{edge-stable} if it excludes a half-graph as a semi-induced subgraph.} classes of graphs~\cite{dreier2023structurally,dreier2024stable},

\item classes of linearly ordered graphs~\cite{bonnet2024twin},

\item tournaments~\cite{geniet2023first}.
\end{itemize}
\end{multicols}
Moreover, model checking is known to be $\mathsf{AW}[\star]$-hard on every hereditary graph class that is not monadically dependent~\cite{dreier2024flip}.
Strikingly, in many cases, monadic dependence corresponds to a graph-theoretic notion: 
\begin{itemize}
    \item nowhere-denseness for weakly sparse classes~\cite{ADLER2014322},
    \item flip-flatness for edge-stable classes \cite{flip-flatness},
    \item bounded twin-width for linearly ordered graphs and tournaments \cite{bonnet2024twin,geniet2023first}.
\end{itemize}
 
The connection between monadic dependence of linearly ordered structures, twin-width, and model checking will be central to our results for graphs equipped with forest orders.
We refer to the recent survey~\cite{pilipczuk2025survey} for a broader overview.


\newcommand{\efficientEncodingFootnote}{
For logics more expressive than first-order logic, independence alone need not rule out fixed-parameter tractable model checking, even on hereditary graph classes. 
For example, sufficiently long subdivisions can make the encoded graph small relative to the input, leaving enough time to evaluate a sentence by brute force on the encoded instance. 
Such phenomena occur for monadic second-order logic~\cite{ganian2010there,kreutzer2010lower} and separator logic~\cite{pillipczuk2021connectivity}.

To obtain hardness, we therefore impose an efficiency condition on the encoding. 
An \emph{efficient encoding of graphs} consists of a polynomial-time construction that maps every finite graph to a monadic expansion of a structure in the class, from which a fixed simple interpretation in the logic recovers the input graph. 
The interpretation then gives an $\mathsf{AW}[\star]$-hardness reduction for model checking on these monadic expansions. 
Our hardness theorem uses a specific efficient minor-encoding hypothesis, stated below.
}

\paragraph{First-order logic with connectivity and the quest for a logic for classes with excluded minors}
A basic limitation of first-order logic is that it cannot express connectivity. Several natural extensions of first-order logic were introduced to address this gap.
A recent example is \emph{separator logic}, also denoted by $\FOconn$, which extends $\FO$ on graphs by allowing the use of the atomic predicate $\mathsf{conn}(x,y;\bar z)$, which evaluates to true if there is a path in the graph from $x$ to $y$ that avoids the vertices of~$\bar z$~\cite{bojanczyk2021separator,schirrmacher2023first}.
The following theorem summarizes the resulting frontier on monotone classes. 

\begin{theorem}[
  store=thm:separator-characterization,
  restate-keys={note=restated}
]
Let $\Cc$ be a monotone class of graphs. Then the following are equivalent:
\begin{enumerate}[label=(\arabic*),nosep]
  \item $\Cc$ excludes some graph as a topological minor;
  \item $\Cc$ is monadically dependent (for \FOconn);
  \item $\Cc$ is  dependent (for \FOconn).
\end{enumerate}
\medskip

Moreover, any of these conditions imply that  {\FOconn} model checking is \FPT\ on $\Cc$, while if $\Cc$ is independent and admits efficient encoding\footnote{\efficientEncodingFootnote} of graphs, then
model checking  on $\Cc$ is $\AW[*]$-hard.
\end{theorem}

The algorithmic statements of the theorem were proved in~\cite{pillipczuk2021connectivity}, and we give a brief sketch of the dependence characterizations in~\ref{app:separator-characterization}.
Further works on extensions of $\FO$ for topological minor-free classes include~\cite{golovach2022dpminor,SauST25para,schirrmacher2024topodp,bojanczyk2025lowrank,SauSSSTV26}

The next natural graph-theoretical dividing line corresponds to classes excluding a minor, hence the following question.

\begin{problem}
\label{pb:1}
\textsl{Is there a natural logic $\mathcal L$ such that a monotone or weakly sparse class of graphs is monadically dependent for $\mathcal L$ if and only if it excludes a
minor?}
\end{problem}

\paragraph{Understanding the structure of monadically dependent classes}
A complementary approach to monadically dependent classes studies \emph{tree-ordered weakly sparse} (TOWS) representations of such graph classes.
Many successful structural representations come with an underlying tree order. 
Exposing this order as a part of the structure makes global structure accessible to first-order logic.
This viewpoint is central in the study of classes of many tameness notions for graph classes~\cite{bonnet2024twin,bonnet2022twwI,lmcs_perm,covers,mergemodels,Colcombet2007,courcelle1992monadic,dreier2025merge,gajarsky2020first,ganian2019shrub,ganian2012trees},
and has very recently led to new characterizations of minor exclusion through tree-ordered weakly sparse graphs~\cite{buffiere2026characterizations}, which is a first step toward a solution of \zcref{pb:1} and the departure point of the present paper. Let us give the background.
\medskip

Let $G$ be a graph. A \emph{spanning forest} of $G$ is an acyclic subgraph $F\subseteq G$ with $V(F)=V(G)$. 
Following~\cite{gross2018graph}, we call $F$ \emph{full} if $F$ and $G$ induce the same partition of $V(G)$ into connected components. 
A \emph{rooted spanning forest} has one distinguished vertex, called a \emph{root}, in each component. 
It defines a forest order $\preceq_F$ on $V(G)$ by $x\preceq_F y$ if $x$ lies on the unique path from $y$ to the root of its component. 
For a graph class $\Cc$, let
\[
\begin{aligned}
\Cc^\Upsilon&:=\{(G,\preceq_F):G\in\Cc,\ F\text{ is a full rooted spanning forest of }G\},\\
\Cc^F&:=\{(G,\preceq_F):G\in\Cc,\ F\text{ is a rooted spanning forest of }G\}.
\end{aligned}
\]
These are the classes of graphs in $\Cc$ expanded by full and arbitrary rooted spanning forest orders, respectively.
Two results established a logical characterization of minor exclusion in this setting.

\begin{theorem}[note={\cite[Theorem~1.4]{buffiere2026characterizations}},store=thm:minor-tows]
Let $\Cc$ be a weakly sparse class of graphs. Then the following are equivalent:
\begin{enumerate}[label=(\arabic*),nosep]
  \item $\Cc$ excludes some graph as a minor;
  \item $\Cc^{\Upsilon}$ is monadically dependent (for \FO);
  \item $\Cc^F$ is monadically dependent (for \FO).
\end{enumerate}
\end{theorem}

The equivalence of the first two conditions is \cite[Theorem~1.4]{buffiere2026characterizations}.
The equivalence of (2) and (3) is an easy observation, see \zcref{rem:forest-tree-reduction} in the preliminaries. 

\begin{theorem}[note={\cite[Theorem~2]{bonnet2022twin}},store=thm:minor-tww]
Let $\Cc$ be a monotone class of graphs. Then the following are equivalent:
\begin{enumerate}[label=(\arabic*),nosep]
  \item $\Cc$ excludes some graph as a minor;
  \item $\Cc$ has bounded \emph{spanning twin-width}, that is: there exists a constant $t$ such that every graph $G\in\Cc$ has a full rooted spanning forest $F$ such that ${\rm tww}(G,\preceq_F)\leq t$.
\end{enumerate}
\end{theorem}

\subsection{Results}

\paragraph{Structural and algorithmic results}
If $\Cc$ excludes a fixed minor, then~\zcref{thm:minor-tows} shows that $\Cc^F$ is monadically dependent.
We strengthen this conclusion to bounded twin-width.

\getkeytheorem{thm-main-structure}

This is significant as, for linearly ordered binary structures, bounded twin-width \textsl{exactly} captures monadic dependence and governs the tractability of first-order model checking.
While the above theorem strengthens the $(1) \Rightarrow (2)$ direction of \zcref{thm:minor-tww}, we also obtain a strengthening of the $(2) \Rightarrow (1)$ direction on monotone classes (\zcref{thm-main-monotone}).

The following is an algorithmic version of  \zcref{thm-main-structure}. The exclusion of a minor is restated using the notion of Hadwiger number $h(G)$ of a graph $G$.

\getkeytheorem{thm-main-algo}

This theorem extends to binary relational structures (\zcref{thm-main-algo-structures}). 
As $\FO$ model checking is known to be fixed-parameter tractable on bounded twin-width classes of linearly ordered binary structures \cite{bonnet2024twin},
it also implies the following model checking result.

\getkeytheorem{cor-model-checking}

In the above, for a structure $\mathbf M$ with universe~$M$, its
\emph{Gaifman graph} $\Gaif(\mathbf M)$ has vertex set~$M$, with
distinct vertices adjacent if they occur together in a tuple of
a relation of~$\mathbf M$.

At first glance, \zcref{cor-model-checking} may seem somewhat technical, since it concerns first-order logic over a more general signature obtained by enriching graphs with an additional structure, namely a forest order.
We next show that this is in fact a very natural viewpoint: first-order logic over this enriched signature admits an exact characterization in terms of a natural extension of first-order logic on graphs by suitable connectivity predicates.

\paragraph{Solving \zcref{pb:1}: A logic for classes that exclude a minor}
We introduce \emph{sub-connectivity logic}, denoted by $\FOscon$. 
We now directly work in the setting of binary structures. Let
$\sigma$ be a finite binary relational signature containing a distinguished
binary relation symbol~$Z$. Formulas of \emph{sub-connectivity logic},
denoted by $\FOscon[\sigma]$, are evaluated on $\sigma$-structures
$\mathbf M$ in which $Z(\mathbf M)$ is symmetric and irreflexive. We identify
the $Z$-reduct~$\MZ$ with the simple undirected graph
on~$M$ having edge relation~$Z(\mathbf M)$. In addition to the usual
first-order atoms over~$\sigma$, the logic has atoms
\[
   \scon(x,y;\bar z),
\]
where $\mathbf M\models\scon(x,y;\bar z)$ if and only if $x$ and $y$ lie
outside the set of vertices occurring in $\bar z$ and belong to the same
connected component of $\mathbf M\!\upharpoonright_Z-\bar z$.

Hence, when $\sigma=\{E,Z\}$ and $\mathbf M$ is a $\sigma$-structure whose
$E$-reduct is a graph and which satisfies $Z(\mathbf M)=E(\mathbf M)$, the atom $\scon$ is exactly the connectivity atom of separator logic~$\FOconn$.
However, allowing the interpretation of~$Z$
to specify a second independent edge layer makes it possible to reason about
connectivity in, for example, a spanning forest or the branch sets of a minor
model. This explains why the frontier moves from topological-minor exclusion
for $\FOconn$ to \emph{minor exclusion} for $\FOscon$.

Our first main result about $\FOscon$ is an exact characterization of its expressive power by relating $\FOscon$ to the spanning forest orders discussed above.

\getkeytheorem{thm:scon=FOFO}

Combining the forward translation with the structural theorem yields monadic
dependence whenever the Gaifman graphs exclude a fixed minor. In the
graph special case, the converse holds on every weakly sparse class, answering
\zcref{pb:1}.

\getkeytheorem{thm:dependence}

Combining the same translation with the algorithmic version of our structural
theorem gives the promised meta-theorem for arbitrary binary structures.

\getkeytheorem{thm:main-mc}

We also prove a hardness counterpart to~\zcref{thm:main-mc}. Let $\Cc$ be monotone and exclude no fixed minor. 
Assume that, given a graph $H$, one can compute in polynomial time a graph $F\in\Cc$ that contains $H$ as a minor and is minimal under taking subgraphs with this property. 
Then $\FOscon$ model checking on $\Cc_Z$ is $\mathsf{AW}[\star]$-hard, parameterized by the sentence size. 
Thus, under this encoding hypothesis and the assumption $\mathsf{FPT}\ne\mathsf{AW}[\star]$, minor exclusion is the tractability frontier in this monotone setting. 
See~\cite[Section~8.6]{Flum2006ParameterizedCT} for the complexity class.

\getkeytheorem{thm:main-hardness-new}

\medskip

\paragraph{An algorithmic application}

As an algorithmic application, we consider combinatorial reconfiguration and solution discovery for vertex-deletion problems whose feasibility conditions are expressible in separator logic~$\FOconn$.
The idea behind combinatorial reconfiguration is to transform one configuration of a system into another while keeping the system running throughout the transformation.
For vertex-subset problems on graphs, we can imagine a token placed on each vertex of a solution.
We change the configuration by moving one token at a time to an unoccupied vertex, either by sliding it along an edge or by allowing it to jump to any unoccupied vertex.
The requirement that the system keeps running means that the vertices occupied by tokens must form a feasible solution after every move.
The algorithmic question is whether two given feasible solutions can be transformed into one another by a sequence of such moves.
This viewpoint has been applied to a wide range of
graph problems. 
We refer for example to the surveys~\cite{journals/cosrev/BousquetMNS24,nishimura18reconfiguration,books/cu/p/Heuvel13}.

A particularly relevant family of examples arises from reconfiguring vertex sets that hit or separate prescribed structures. 
For instance, \textsc{Feedback Vertex Set Reconfiguration} asks to transform one deletion set into another while maintaining that the remaining graph is acyclic.
Parameterized algorithms and kernels for several such vertex-set reconfiguration problems were studied in~\cite{journals/algorithmica/MouawadN0SS17}. 
More recently, \textsc{Feedback Vertex Set Reconfiguration} has been studied in its own right on planar graph classes~\cite{journals/tcs/BousquetHKMS23}. 
\textsc{Vertex Separator Reconfiguration} under the token addition/removal, token jumping, and token sliding rules has been studied in~\cite{journals/rairo-ro/GomesNS25}, while~\cite{journals/jcss/GomesLMMOSZ24} considers the reconfiguration of minimum $s$-$t$ separators.

Since individual reconfiguration problems often require tailored arguments, there has also been substantial interest in algorithmic meta-theorems for reconfiguration. Mouawad, Nishimura, Raman, and Wrochna\ gave an $\MSO$-based meta-theorem parameterized by treewidth and the length of the reconfiguration sequence~\cite{conf/iwpec/MouawadNRW14}.
Gima, Ito, Kobayashi, and  Otachi\ obtained meta-theorems for reconfiguration under other structural parameters, such as neighborhood diversity and treedepth~\cite{conf/esa/GimaIKO22}. 
Further meta-algorithmic and structural results are known for sparse graph classes, for example for independent set and dominating set reconfiguration on sparse graphs~\cite{journals/jcss/LokshtanovMPRS18,journals/combinatorics/Siebertz18}.

Our logic provides a new meta-algorithmic viewpoint. We consider
\emph{modulator reconfiguration} problems in which the current token set
$S$ is a deletion set and feasibility is expressed by a logical condition
on the remainder $G-S$. By using $\FOconn$ as the specifying logic, this captures, among others, \textsc{Vertex Cover
Reconfiguration}, \textsc{Feedback Vertex Set Reconfiguration}, and
 \textsc{$s$-$t$-Cut Reconfiguration}.
Let us stress again that we use separator logic to specify the property of $G-S$, problem, while we use the methods for subconnectivity logic to resolve the algorithmic reconfiguration problem. 

\getkeytheorem{thm:modulatorreconfiguration}

The same method applies to \emph{solution discovery}, where the initial token set need not be feasible and feasibility is required only at the end of the sequence. We refer to~\zcref{sec:reconfiguration} for the precise definitions and the corresponding meta-theorem for \textsc{Modulator Discovery}.

\medskip 

\paragraph{Extending to regular path queries is not possible}

The logic {\FOscon} can be further extended by introducing a recursive scheme while keeping the fixed-parameter tractability property on every class of graphs excluding a minor; see \zcref{subsec-guarded-fragment}.

We investigate how far one can push this line of ideas.
In database theory and graph querying, \emph{connectivity with regular constraints on the witnessing path} is a central notion.
It appears in regular path queries, conjunctive regular path queries, and their extensions with regular relations on tuples of paths~\cite{cruz1987,mendelzon1995,barcelo2012,libkin2012,freydenschweik2013,wood2012,salvati2024}.
One might therefore hope to strengthen $\FOscon$ by requiring the connecting paths to satisfy simple regular constraints.
We show that already with two colors and very simple regular constraints,
the corresponding extension has, on monotone classes, bounded tree-width
as its exact dependence frontier. The hardness direction additionally
requires that the relevant grid-minor encodings, including their minor
models, can be constructed efficiently.

\section{Overview of the structural and algorithmic proofs}
\label{sec:outline_struc}

We now give an overview of the proof of our main structural theorems, \zcref{thm-main-structure,thm-main-algo}. 

\subsection{From minor exclusion to bounded twin-width}

Theorems \zcref{thm-main-structure,thm-main-algo} are derived from the following theorem, \zcref{thm-aux-structure}, which treats connected graphs equipped with rooted spanning trees.
The passage to arbitrary rooted spanning forests uses \zcref{rem:forest-tree-reduction}.

\getkeytheorem{thm-aux-structure}

The proof of this derivation as well as the proof of \zcref{thm-aux-structure} will be the subject of \zcref{sec:proof-main-thm}. We now give an outline of the proof of \zcref{thm-aux-structure}.

\medskip 

Our proof is organized around \emph{transductions}. Informally, a transduction
is a uniform logical transformation between classes of structures. It may
take a bounded number of copies of the input, add finitely many unary colors 
and then define the universe and relations of the output by formulas.
Transductions therefore allow us to encode the desired structures inside a
better-behaved auxiliary class.

Two preservation results are central. First, on classes of linearly ordered
binary structures, monadic dependence is equivalent to bounded twin-width.
Second, bounded twin-width is preserved under fixed transductions, with a
computable change in the bound. We therefore do not prove bounded twin-width
directly for every desired tree-ordered expansion. Instead, we construct an
auxiliary linearly ordered class that is monadically dependent, and hence has
bounded twin-width, and recover the desired structures from it by a fixed
transduction.

Fix $t\in\mathbb N$, and let $\mathscr H_t$ be the class of connected graphs with Hadwiger number at most~$t$.
For every graph $G\in \mathscr H_t$ and every rooted spanning tree $Y$ of~$G$,
we construct an auxiliary graph~$H$, together with a rooted spanning tree
$Y_H$ and a linear order~$<$, and we write the resulting tree- and linearly
ordered structure as $(H,\preceq_{Y_H},<_H)$.
The construction is designed so that the
class
\[
    \mathscr M_t = \{ (H,\preceq_{Y_H},<_H) \mid G\in \mathscr H_t,
    \text{$Y$ is a rooted spanning tree of $G$} \}
\]
is monadically dependent. At the same time, $(H,\preceq_{Y_H},<_H)$ is constructed so
that it contains enough information to recover, by one fixed transduction,
the original graph~$G$ together with the spanning-tree order defined by~$Y$
and a linear order~$<_G$. We denote this expansion by~$(G,\preceq_Y,<_G)$.

It follows qualitatively that $\mathscr M_t$ has bounded twin-width, because it
is a linearly ordered monadically dependent class of binary structures. Since
the class of the desired structures $(G,\preceq_Y,<_G)$ is obtained from $\mathscr M_t$ by
a fixed transduction, preservation of twin-width under transductions then
gives a qualitative bound for the desired class as well. To obtain the
computable bound required by the algorithmic theorem, we additionally use
the effective quantitative content of the tree-ordered weakly sparse
characterization: the auxiliary spanning trees have maximum degree~$3$, and
the graph reducts of all their tree-ordered minors exclude a biclique whose
order is computable from~$t$. This yields a computable twin-width bound for
$\mathscr M_t$. The constructive linearization and recovery transductions then
preserve computability of the final bound.

The proof of \zcref{thm-aux-structure} is accordingly split into two phases (see \zcref{fig:overview}). 
In Phase~I, starting from the pair~$(G,Y)$, we construct the auxiliary linearly ordered tree-ordered structure $(H,\preceq_{Y_H},<_H)$ and prove that the class $\Mm_t$ of obtained structures forms a monadically dependent class. 
In Phase~II, we show that no information needed for recovering the original graph and tree order has been lost: from $(H,\preceq_{Y_H},<_H)$ we transduce back a linearly ordered expansion $(G,\preceq_Y,<_G)$ of the original tree-ordered graph. 
The following paragraphs give a more detailed informal description of these two phases before the full formal construction is presented.

\begin{figure}[ht]
\scalebox{.96}{
\begin{xy}
\xymatrix{
&&\mathcal T=(T,\bag)\ar@{..>}[dr]^{2}\\
\text{Phase I}&(G,Y)\ar@{..>}[ur]^(.4){1}\ar@{..>}[r]^(.45){1}&(G_1,Y)\ar@{..>}[r]^(.45){2}&(G_2,Y_2)\ar@{..>}[r]^(.45){3}&(H,Y_H)\ar@{..>}[r]^{4}&(H,\preceq_{Y_H})\ar@{..>}@/^/[dd]_{5}\ar@/_/[dd]_{\mathsf T_1}\\
\ar@{--}[rrrr]&&&&&\\
\text{Phase II}&(G,\preceq_Y,<_G)&{\mathbf N}\ar[l]_{9}^{\mathsf T_5}&{\mathbf M}\ar[l]_{8}^{\mathsf T_4}&(G_2,\preceq_{Y_2},<_{G})\ar[l]_(.55){7}^(.55){\mathsf T_3}& (H,\preceq_{Y_H},<_H)\ar[l]_(.45){6}^(.45){\mathsf T_2}
}
\end{xy}
}
\caption{Outline of the proof of \zcref{thm-aux-structure}. Dotted arrows are polynomial-time computations and plain arrows are transductions. (Precisely, $G\xrightarrow[\mathsf T]{~1~}H$ means that at step 1 of the proof, we prove $H\in\mathsf T(G)$.)
The main steps of the proof are:
    {\bf 1.}~construction of a neat tree decomposition of~$G$ based on the structure theorem of Robertson and Seymour, as well as an augmentation of $G$ so that every adhesion is a clique;
    {\bf 2.}~construction of a subgraph whose nontrivial blocks have explicit
    no-apex almost-embedding witnesses;
    {\bf 3.}~effective blockwise translation into a subcubic graph of bounded
    Hadwiger number;
    {\bf 4.}~computation of the tree-ordered expansion;
    {\bf 5.}~computation of a linear order that can be obtained by transduction $\mathsf T_1$;
    {\bf 6.}~transduction of the linearly ordered tree-ordered expansion of $G_2$ by $\mathsf T_2$;
    {\bf 7.}~transduction of a master structure that permits the graph to be handled margin-wise by~$\mathsf T_3$;
    {\bf 8.}~recovery of the target tree order on the margins by transduction $\mathsf T_4$;
    {\bf 9.}~full recovery of the tree order by transduction $\mathsf T_5$.}
    \label{fig:overview}
\end{figure}
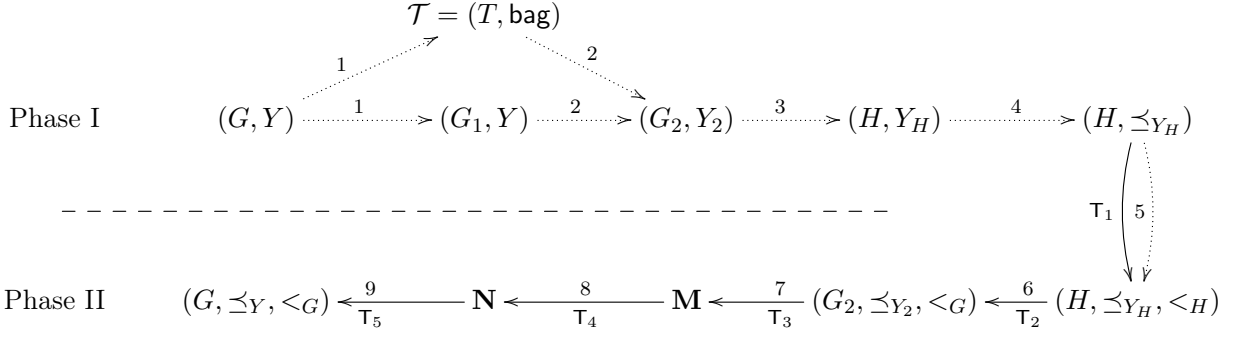

\begin{itemize}
    \item {\bf Phase I.} Given a connected graph $G\in \Hh_t$ with a rooted spanning tree $Y$, we construct in polynomial time a tree-ordered graph $(H,\preceq_{Y_H})$ with a linearly ordered expansion~$(H,\preceq_{Y_H},<_H)$ of~$(H,\preceq_{Y_H})$ such that the class
    \[\mathscr M_t = \{ (H,\preceq_{Y_H},<_H) \mid G\in \mathscr H_t,
    \text{$Y$ is a rooted spanning tree of $G$}\}\]
    is monadically dependent, and hence, as a linearly ordered class, has bounded twin-width.
    
    \item {\bf Phase II.} From $(H,\preceq_{Y_H},<_H)$ we transduce a linearly ordered
  expansion $(G,\preceq_Y,<_G)$ of the original tree-ordered graph. The linear order $<_G$ is
  inherited from the order $<_H$.
\end{itemize}

Consequently, the linear order $<_G$ is computable in polynomial time from $(G,Y)$, and the class of all resulting structures $(G,\preceq_Y,<_G)$ has bounded twin-width because it can be transduced from~$\mathscr M_t$.

\paragraph{Phase I}

Our starting point is \zcref{thm:minor-tows}, stating that a weakly sparse class $\Cc$ of graphs excludes some graph as a minor if and only if the associated class $\Cc^{F}$ of expansions by rooted spanning-forest orders is monadically dependent.

This observation gives monadic dependence, but it does not yet explain how to obtain bounded twin-width. 
The missing step is to pass from forest orders to linear orders. 
Here we use an idea inspired by a result of~\cite{buffiere2026characterizations}: if the underlying spanning tree has bounded degree, then a suitable linearization of this tree order can be obtained by a transduction. 
Of course, graphs in a minor-closed class need not have bounded degree and therefore their rooted spanning trees may also have unbounded degree. 
The main technical point of the proof is to overcome this obstacle. 
We prove the following result, which may be of independent interest.

\begin{theorem}\label{thm:minor-transduction}
There exist computable functions $\xi$ and $f$, a universal constant~$c$, a fixed transduction~$\mathsf L$, effectively computable families of transductions $(\mathsf T_t)_{t\in\mathbb N}$ and $(\mathsf R_t)_{t\in\mathbb N}$, and an algorithm with the following properties, for every $t\in\mathbb N$.
\begin{enumerate}[nosep]
    \item Given a graph~$G$ with $h(G)\leq t$, the algorithm computes in time $f(t)\cdot |G|^c$ a connected subcubic graph~$H$ and a rooted spanning tree~$Y_H$ of~$H$ such that $h(H)\leq\xi(t)$ and
    \[
        G\in\mathsf T_t(H,\preceq_{Y_H}).
    \]
    \item If an arbitrary rooted spanning forest~$F$ of~$G$ is also given, the algorithm computes, within the same time bound, such a pair $(H,Y_H)$ together with a DFS preorder $<_H$ of~$Y_H$ and a linear order $<_G$ on~$V(G)$ such that
    \[
        (H,\preceq_{Y_H},<_H)\in\mathsf L(H,\preceq_{Y_H})
        \qquad\text{and}\qquad
        (G,\preceq_F,{<_G})\in\mathsf R_t(H,\preceq_{Y_H},<_H).
    \]
    In particular, $(G,\preceq_F)$ can be transduced from $(H,\preceq_{Y_H})$ by some transduction effectively computable from~$t$.
\end{enumerate}
\end{theorem}

Thus, for every fixed~$t$, all graphs of Hadwiger number at most~$t$ are
interpretable by a transduction in tree-ordered subcubic graphs
whose Hadwiger number is bounded by a computable function of~$t$. In
particular, every class excluding a fixed minor is a transduction of a
tree-ordered subcubic class that again excludes a fixed minor.

Note that, by contrast, the analogous statement for the minor relation instead of transduction fails already for graphs of Hadwiger number at most~$5$: although every graph is a minor of a subcubic graph, there is no class~$\Dd$ of subcubic graphs that excludes a fixed minor and whose minor closure contains every $K_6$-minor-free graph; this follows e.g.\ from~\cite[Theorem~17]{MilolidakisThesis}.

For every fixed~$t$, let $\mathscr M_t$ be the class of all structures
$(H,\preceq_{Y_H},<_H)$ produced by the second part of
\zcref{thm:minor-transduction} from connected graphs~$G$ with $h(G)\leq t$
and rooted spanning trees~$Y$ of~$G$. The graph reducts of the structures in
$\mathscr M_t$ are subcubic and have Hadwiger number at most~$\xi(t)$, and
each of their tree orders is induced by a rooted spanning tree. Thus, the
corresponding tree-ordered class is a subclass of the class of full rooted
spanning-tree expansions of a minor-excluded class, and is
monadically dependent by~\zcref{thm:minor-tows}. Applying the fixed
linearization transduction~$\mathsf L$ shows that $\mathscr M_t$ is
monadically dependent. Since $\mathscr M_t$ consists of linearly ordered
binary structures, it has bounded twin-width.

\subsection{Sketch of the proof of \zcref{thm:minor-transduction}
}

Let us sketch how we prove \zcref{thm:minor-transduction}. 
We first give an informal introduction to tree decompositions and the structure theorem of Robertson and Seymour for graphs excluding a fixed graph as a minor, on which the proof of the theorem is based.

Very roughly, a tree decomposition represents a graph by arranging
overlapping pieces in the shape of a tree. These pieces are called
\emph{bags}. The bags cover the graph: every vertex occurs in a bag, and every
edge has both endpoints together in some bag. Moreover, the bags containing a
fixed vertex induce a connected subtree of the decomposition tree. Thus, a
vertex may occur in several bags, but those bags form one contiguous region.

If the decomposition tree is rooted, every non-root bag has a parent. For a
bag~$x$, its intersection with the parent bag is the \emph{adhesion} of~$x$.
The vertices in~$x$ but not in its adhesion form the \emph{margin} of~$x$.
Intuitively, the margin is the new part contributed by the bag, while the
adhesion is the part shared with its parent. These overlaps are separators
that control how the local pieces interact.

The Robertson-Seymour graph minor structure theorem states, informally, that
every graph excluding a fixed minor has a tree decomposition of bounded
adhesion whose pieces are almost embeddable in a surface of bounded Euler
genus. Up to deleting a bounded number of exceptional vertices, called
\emph{apices}, and handling a bounded number of controlled local
complications, called \emph{vortices}, each piece behaves like a graph drawn
on a surface.

We begin with a connected graph $G\in\mathscr H_t$. In Step~1, we use the
Robertson-Seymour structure theorem to compute a tree decomposition and
augment $G$ to a graph~$G_1$ in which every adhesion is a clique. We may
further arrange that the decomposition is \emph{neat}, meaning that it
satisfies additional structural properties specified later.

The margins of a neat tree decomposition partition the vertex set. For each
non-root node~$x$ of the decomposition tree, we choose a representative
$\beta(x)$ in the margin of~$x$ and a neighbor $\alpha(x)$ in the adhesion of
$x$ that belongs to the margin of the parent of~$x$; see
\zcref{fig:rob-sey}. A representative is also chosen for the root node, for
which no vertex $\alpha(x)$ is needed.

\begin{figure}[h!t]
\begin{center}
    \includegraphics[width=.85\textwidth]{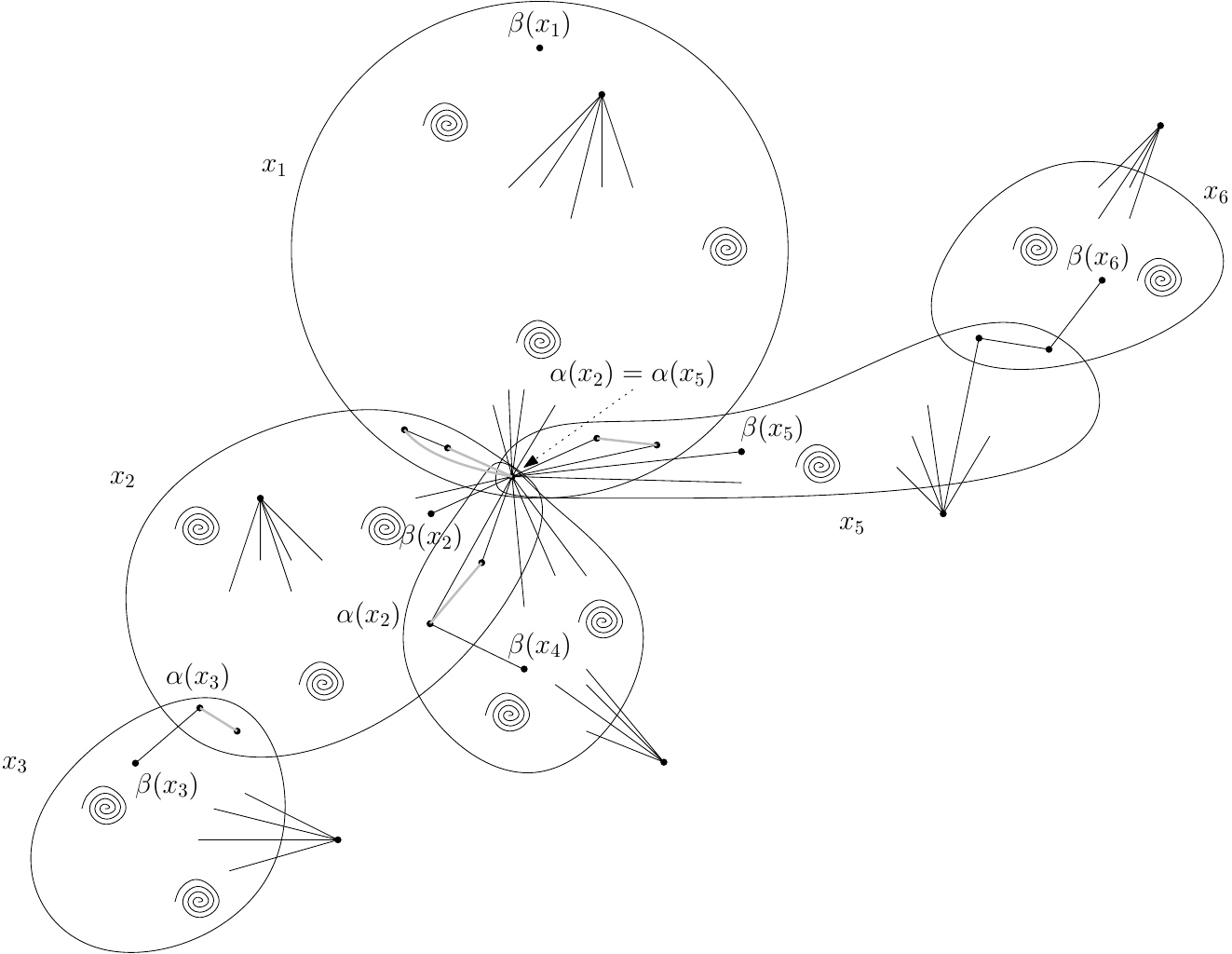}
\end{center}
    \caption{In the first step, we apply the decomposition theorem of Robertson and Seymour and process the decomposition so that it becomes neat. 
    For each node $x$ of the decomposition tree, we select a representative vertex~$\beta(x)$ from its margin $\mrg(x)$. 
    If $x$ is not the root, we also select a neighbor $\alpha(x)\in\adh(x)\cap\mrg(\parent(x))$ of~$\beta(x)$.}
    \label{fig:rob-sey}
\end{figure}

In Step~2, we delete all edges between distinct margins except the selected
edges $\alpha(x)\beta(x)$. We also delete the apex-incident edges and then add
back an inclusion-minimal set that preserves connectivity inside each margin.
After a corresponding local modification of the original spanning tree, we
obtain a spanning tree~$Y_2$ of the new graph~$G_2$; see
\zcref{fig:apex-minor-free}. We prove that $G_2$ excludes a fixed apex graph as
a minor. Recall that an apex graph is a graph that becomes planar after deleting one vertex.

\begin{figure}[h!t]
    {}\hfill\includegraphics[width=.52\textwidth]{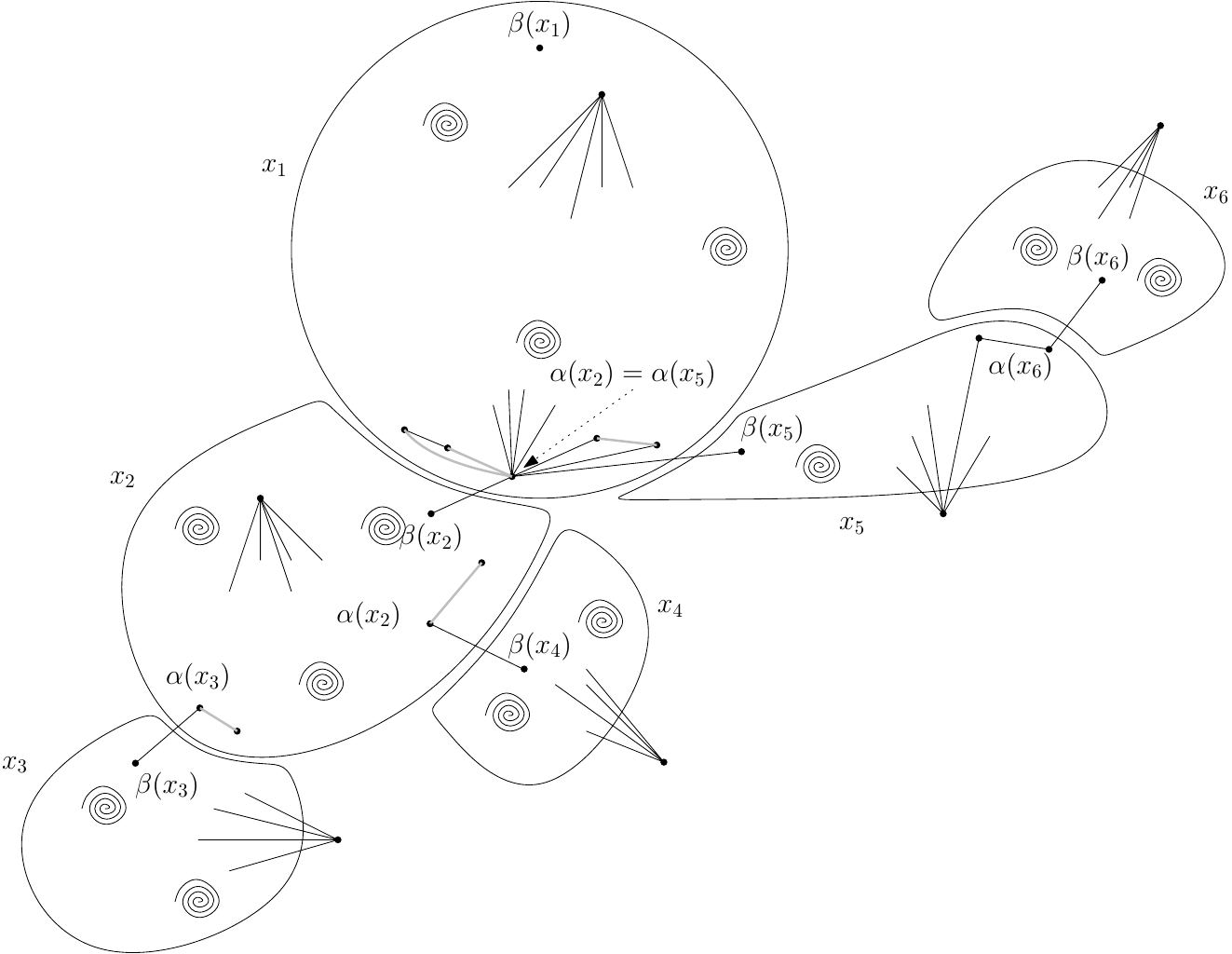}\hfill
    \includegraphics[width=.45\textwidth]{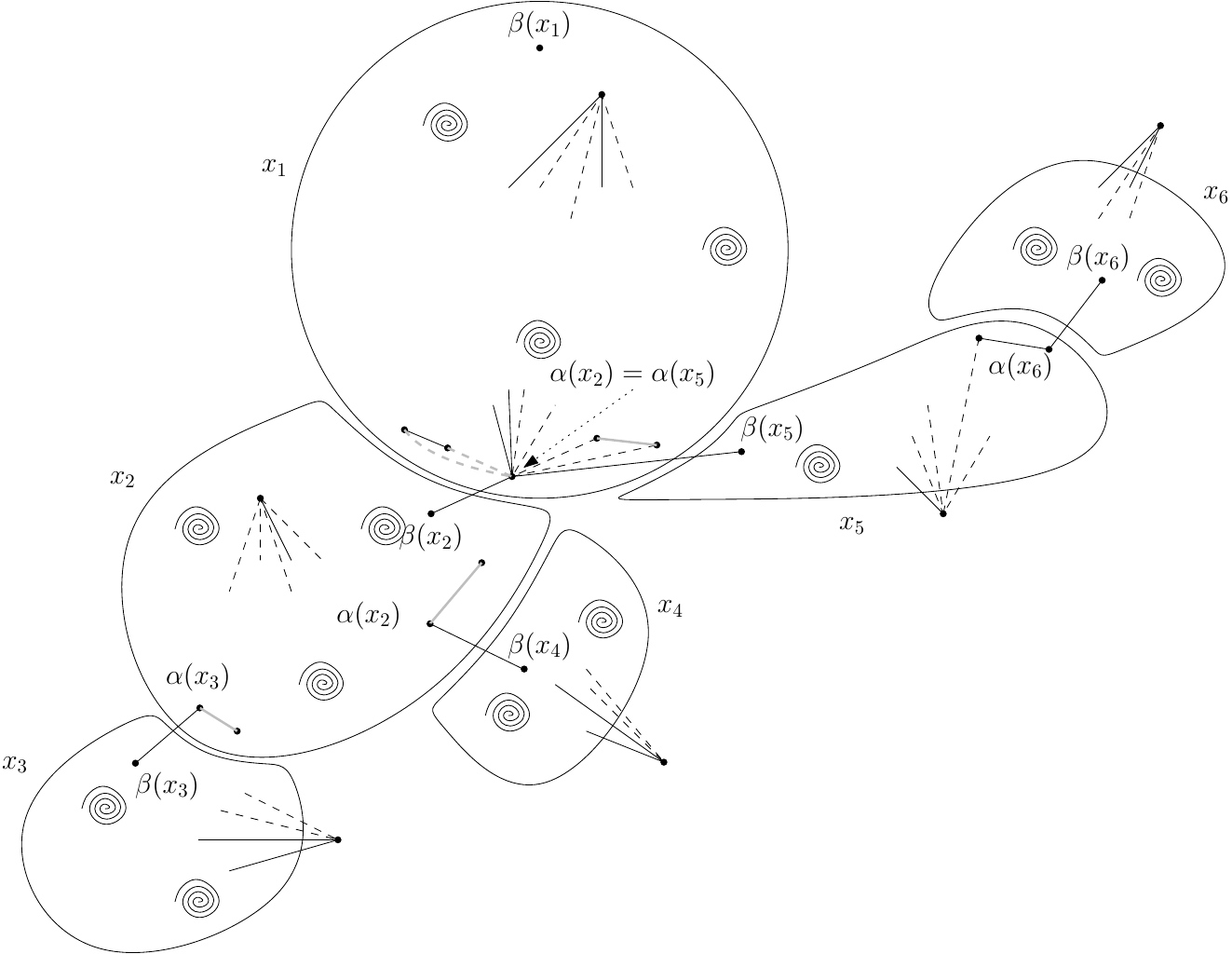}
    \hfill{}
\caption{Illustration of the second step. We delete all edges between margins, except for the single edge $\alpha(x)\beta(x)$ keeping the tree structure of the tree decomposition intact. We also keep only a minimal set of edges incident with apex vertices so as to preserve the connectivity of the margins and such that the resulting blocks fall into an apex-minor-free class.}
\label{fig:apex-minor-free}
\end{figure}

In Step~3, we apply the degree-reduction construction of~\cite{MilolidakisThesis} to~$G_2$, according to
which every graph in an apex-minor-free minor-closed class is a minor of
a subcubic graph in a fixed proper minor-closed superclass. Applied
to~$G_2$, the construction produces a connected subcubic graph
$H\in\mathscr H_{\xi_c(t)}$ such that~$G_2$ is a minor of~$H$, where
$\xi_c$ is computable, and transfers the spanning tree~$Y_2$ to a
spanning tree~$Y_H$ of~$H$. It runs in time
$a(t)\cdot |G_2|^{c_0}$ for a computable function~$a$ and an absolute
constant~$c_0$; see~\zcref{fig:tocubic}.

\begin{figure}[h!t]
\begin{center}
    {}\hfill\includegraphics[width=.5\textwidth]{dec3}\hfill
    \includegraphics[width=.5\textwidth]{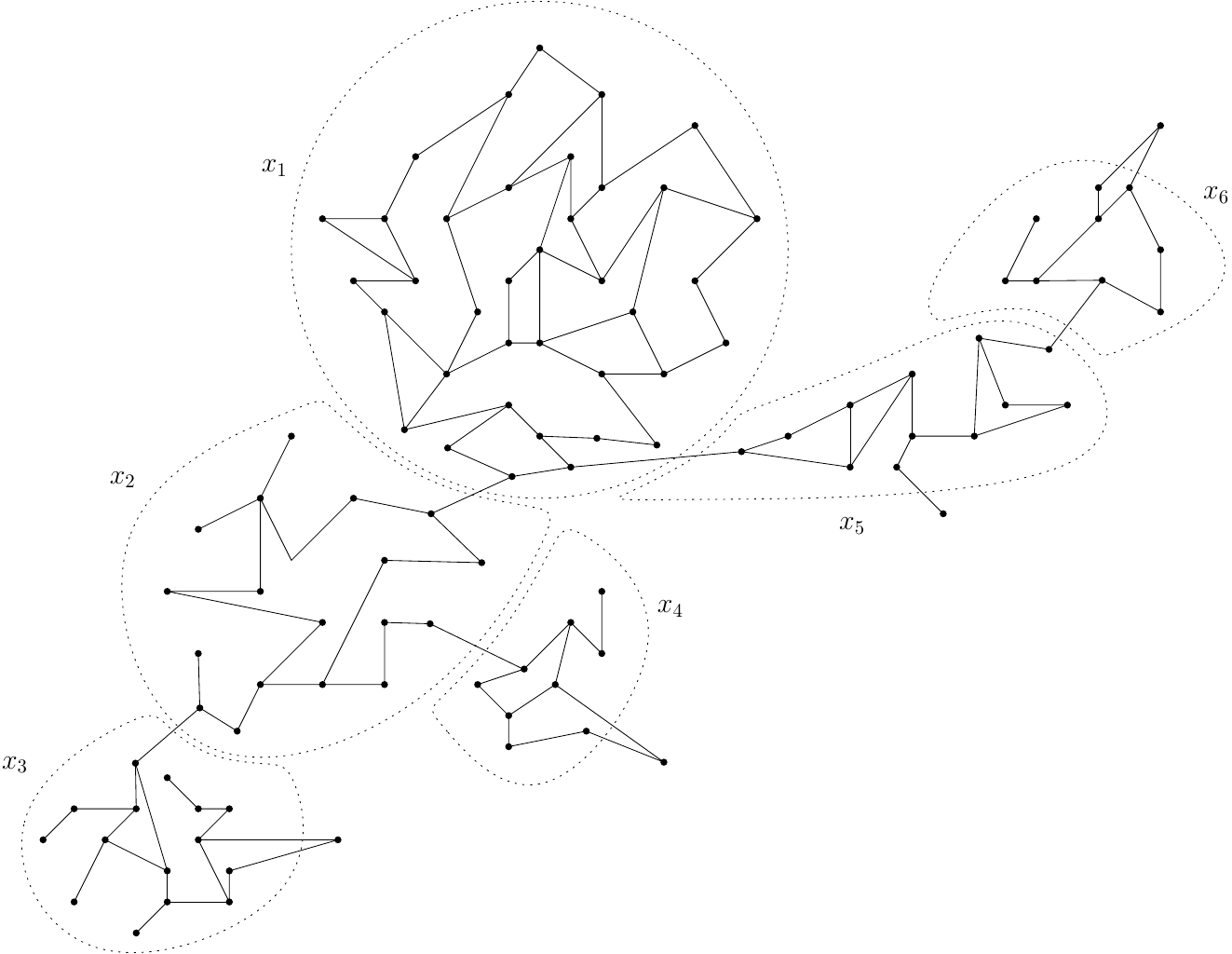}\hfill{}
\end{center}
\caption{Illustration of the third step. We move to a connected subcubic graph $H$ from a class excluding a minor.}
\label{fig:tocubic}
\end{figure}

In Step~4, we form the tree-ordered lift $(H,\preceq_{Y_H})$. The relevant class of
forest-ordered subcubic graphs is monadically dependent by
\cite{buffiere2026characterizations}. In Step~5, we compute a linear order $<_H$ of
$V(H)$ from a depth-first traversal of~$Y_H$, which is also obtainable by
a transduction of $(H,\preceq_{Y_H})$. Hence, the resulting class~$\mathscr M_t$ of
linearly ordered structures $(H,\preceq_{Y_H},<_H)$ is monadically dependent and therefore
has bounded twin-width.
\bigskip

\paragraph{Phase II}

Phase~II begins by transducing $(G_2,\preceq_{Y_2})$ from $(H,\preceq_{Y_H},<_H)$ and transforming
the computed linear order $<_{H}$ to a linear order $<_{G}$ of~$V(G_2)$ (Step~6). We then construct a master
structure~$\mathbf M$ that lets us work one margin at a time while retaining
access to the order of the decomposition tree (Step~7); see
\zcref{fig:masterM}.

\begin{figure}[h!t]
\begin{center}
    \includegraphics[width=.75\textwidth]{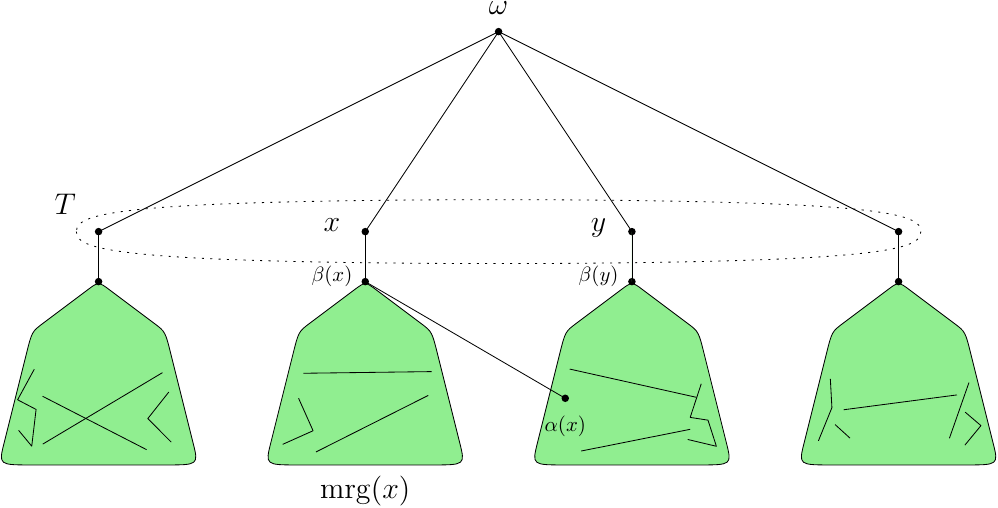}
\end{center}
\caption{
In $\mathbf M$, the children of the root $\omega$ are copies of the decomposition nodes, and the separate relation $\unlhd$ records their tree-decomposition order.
Under each such node~$x$, a subtree rooted at $\beta(x)$ contains the vertices of $\mrg(x)$.}
\label{fig:masterM}
\end{figure}

In Step~8, a transduction restores the edge relation of~$G_1$ inside each margin and the local trace order there; see~\zcref{fig:masterN}.
The remaining edges between margins and the remaining order comparisons are recovered in Step~9.

\begin{figure}[h!t]
\begin{center}
    \includegraphics[width=.75\textwidth]{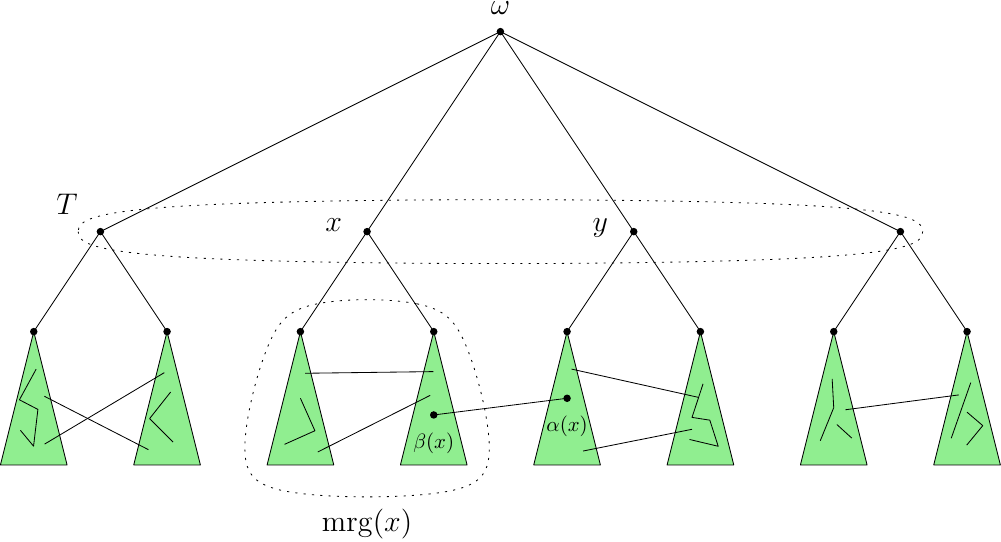}
\end{center}
\caption{In $\mathbf N$, each margin carries the restriction of
the edge relation of~$G_1$ and the local trace order of~$Y$.
The remaining order comparisons are recovered in Step~9.}
\label{fig:masterN}
\end{figure}

Finally, we use the tree version of Simon's factorization forest theorem and
the fact that simple $\MSO$ interpretations of colored rooted trees can be
simulated by first-order transductions of the associated tree orders
\cite{Colcombet2007}. This restores the complete tree order $\preceq_Y$ in
Step~9 and yields $(G,\preceq_Y,<_G)$ as a transduction of~$(H,\preceq_{Y_H},<_H)$. A related
margin-wise reconstruction was used in~\cite{schirrmacher2023first} for
separator logic model checking on classes excluding a topological minor.

\section{Preliminaries}

This section contains definitions and preliminary results that we use in our proofs.

\subsection{Graph theory}
Let $G$ be a graph.
A \emph{forest of $G$} is an acyclic subgraph $F\subseteq G$.
A forest is \emph{spanning} if $V(F)=V(G)$ and it is \emph{full} if the components of $F$ are equal to the components of~$G$ (in this paper, connected components are vertex sets). 
It is a \emph{tree} if it is connected. 
A \emph{rooted spanning forest} is a spanning forest with one distinguished vertex, called \emph{root}, in each connected component. 
If we want to make the roots explicit, we write~$(F,R)$, where $R$ denotes the set of roots and, if $F$ is a tree, we write~$(F,r)$, where $r$ is the unique root of $F$.

Let $G$ be a graph, and let $uv\in E(G)$ be an edge.
The \emph{contraction} of $uv$ is the operation that deletes the edge $uv$, identifies its endpoints $u$ and $v$ into a new vertex adjacent to every vertex that was adjacent to at least one of $u$ and $v$ and then removes any parallel edges.

We say that a graph $H$ is a \emph{minor} of $G$ and write
$H \leq_m G$,
if $H$ can be obtained from a subgraph of $G$ by a sequence of edge contractions.

Equivalently, $H \leq_m G$ if there exists a family $(B_v)_{v\in V(H)}$ of non-empty pairwise vertex-disjoint connected subsets of~$V(G)$, called \emph{branch sets}, such that for every edge $uv\in E(H)$, there is at least one edge of $G$ with one endpoint in $B_u$ and the other endpoint in $B_v$.
Such a family is called a \emph{minor model} of~$H$ in~$G$.
A minor model is \emph{surjective} if its branch sets partition~$V(G)$.
Equivalently, a surjective minor model is specified by a surjective map $f\colon V(G)\to V(H)$ such that each branch set $f^{-1}(v)$ induces a connected subgraph of~$G$ and, for every edge $uv\in E(H)$, some edge of~$G$ joins~$f^{-1}(u)$ and $f^{-1}(v)$.

The \emph{Hadwiger number} $h(G)$ of a graph $G$ is the maximum integer $t$ such that $K_t\leq_m G$.
Using the minor containment algorithm given in \cite{korhonen2024minor}, the Hadwiger number of a graph $G$ can be computed in time $f(h(G))\,|G|^{1+o(1)}$, where $f$ is a computable function.

\subsection{Order theory}

Let $X$ be a finite set.
A binary relation $\preceq$ on $X$ is a \emph{partial order} if it is reflexive
(for all $x\in X$, $x\preceq x$),
antisymmetric
(for all $x,y\in X$, if $x\preceq y$ and $y\preceq x$, then $x=y$),
and transitive
(for all $x,y,z\in X$, if $x\preceq y$ and $y\preceq z$, then $x\preceq z$).
In this case, the pair $(X,\preceq)$ is called a \emph{partially ordered set} (or \emph{poset}).
A \emph{cover} of a poset $(X,\preceq)$ is a pair $(u,v)$ of elements of $X$ such that $u\prec v$ and there exists no $z$ with $u\prec z\prec v$. We denote by $u\pred v$ the (definable) property that $(u,v)$ is a cover.
The \emph{cover graph} of a poset $(X,\preceq)$ is the oriented graph with vertex set $X$, whose arcs are the covers of~$(X,\preceq)$. The \emph{undirected cover graph} of $(X,\preceq)$ is the underlying undirected graph of the cover graph of~$(X,\preceq)$.
Given two elements $x,y\in X$ in a partial order, we write $x\parallel y$ if $x$ and $y$ are incomparable. 
For $B\subseteq X$, the restriction of~$\preceq$ to~$B$ is $\preceq\cap(B\times B)$. 

A partial order $\preceq$ on $X$ is a \emph{forest order} if, for every $y\in X$, the set
$
\{x\in X \mid x\preceq y\}$
is linearly ordered by $\preceq$.
The minimal elements of $\preceq$ are called the \emph{roots} of $\preceq$.
A forest order is a \emph{tree order} if it has a minimum (i.e., a unique root $r$).
Equivalently, a partial order $\preceq$ is a forest order (resp.\ a tree order) if its cover graph is a rooted forest (resp.\ a rooted tree) with all arcs oriented away from the roots.

As any upper-bounded set of elements in a forest order is a chain, we will allow
in formulas statements like 
$w=\max\{u\preceq v\colon P(u)\}$ and
$w=\max\{u\prec v\colon P(u)\}$ (where $P$ is a first-order definable property, possibly with parameters),  and these notations represent first-order definable statements. (We do not define $\max\{u\preceq v\colon P(u)\}$ directly, as 
$\{u\preceq v\colon P(u)\}$ might be empty.)

An \emph{ancestor} of an element $v$ in a forest order $F$ is an element $u$ with $u\preceq v$ and 
the \emph{parent} of $v$ is
\[
\pi_F(v)\is\begin{cases}
    v&\text{if $v$ is a root (i.e., a minimal element),}\\
    w&\text{if $w\pred v$.}
\end{cases}
\]

Note that we use the convention $\pi_F(v)=v$ at a root. 
This makes $\pi_F$ a total map on a forest-ordered set. For decomposition trees, we use the separate notation $\parent$ and the convention $\parent(r)=\bot$, with $\bag(\bot)=\emptyset$.

If $(X,\preceq)$ is a tree-ordered set, then the \emph{infimum} (or least common ancestor) of two elements of~$X$ always exists and is first-order definable. We denote it by $x\curlywedge y$.
For forest orders, we define the \emph{least common ancestor} relation by
\[
\mathsf{lca}(a;x,y)\is (a\preceq x)\wedge (a\preceq y) \wedge
\forall a'\bigl((a'\preceq x\wedge a'\preceq y)\rightarrow (a'\preceq a)\bigr).
\]
Using the notation introduced above, we can also define $\mathsf{lca}(a;x,y)$ as
\[
\mathsf{lca}(a;x,y)\is (a=\max\{v\preceq x\colon v\preceq y\}).
\]

Every rooted tree $(Y,r)$ naturally defines a tree order $\preceq_Y$ on $V(Y)$ by setting

\[
x\preceq_{Y} y
\quad\Longleftrightarrow\quad
\text{there is no path connecting $y$ and $r$ in $Y-x$.}
\]

Analogously, every rooted forest $(F,R)$ (where $R$ denotes the set of the roots) defines a forest order $\preceq_F$ on $V(F)$ by setting

\[
x\preceq_{F} y
\quad\Longleftrightarrow\quad
\begin{gathered}
\text{$x$ and $y$ belong to the same component of $F$, and, for the unique}\\
\text{root $r\in R$ of that component, there is no path from $y$ to $r$ in $F-x$.}
\end{gathered}
\]

Conversely, every finite forest order $\preceq$ is induced in this way by a rooted forest, which is the undirected cover graph of $\preceq$ rooted at minimal elements of $\preceq$.

A binary relation $<$ on $X$ is a \emph{(strict) linear order} if it is irreflexive
(for all $x\in X$, not~$x<x$),
transitive,
and total
(for all distinct $x,y\in X$, either $x<y$ or $y<x$).
In this case, the pair $(X,<)$ is called a \emph{linearly ordered set}.

\smallskip
Let $G$ be a graph.
We use the notation
\begin{itemize}
    \item $F,F_1,F_2,\dots$ for rooted spanning forests of $G$,
    \item $Y,Y_1,Y_2,\dots$ for rooted spanning trees of $G$ (which requires that $G$ is connected),
    \item $T$ for the decomposition tree of a tree decomposition of $G$ (defined below), and 
    \item $\preceq$ and $<$ for forest and linear orders on the vertex set of $G$ (possibly unrelated to the edges of the graph).
\end{itemize}
 
Let $\Cc$ be a class of graphs.
We define
\[
\Cc^\Upsilon
:=
\bigl\{
(G,\preceq_F)\;\big|\;
G\in\Cc,\ 
F \text{ is a full rooted spanning forest of } G
\bigr\},
\]
and 
\[
\Cc^{F}
:=
\bigl\{
(G,\preceq_F)\;\big|\;
G\in\Cc,\ 
F \text{ is a rooted spanning forest of } G
\bigr\}.
\]

\begin{remark}\label{rem:forest-tree-reduction}
We will use the following reductions between forests and trees, and between disconnected and connected structures.

Given a graph $G$ with a rooted spanning forest $F$, add a universal vertex $u$ and extend $F$ to a spanning tree $Y$ by joining $u$ to every root of $F$ and choosing $u$ as the new root.
For the resulting connected graph $\widehat G=G\oplus K_1$, we have
$(G,\preceq_F)=(\widehat G,\preceq_Y)[V(G)]$.
If $G$ excludes $K_t$ as a minor, then~$\widehat G$ excludes $K_{t+1}$ as a minor, since at most one branch set of a clique-minor model can contain~$u$.
Deleting $u$ is a first-order interpretation, which preserves monadic dependence, and taking an induced substructure does not increase twin-width.
The construction and the restriction take polynomial time.
Thus, for the results below whose bounds depend only on an excluded minor, it suffices to consider connected graphs equipped with rooted spanning trees.

For classes of tree-ordered graphs, the hereditary closure obtained by retaining the root and the unrestricted hereditary closure are dependent simultaneously, and monadically dependent simultaneously.
Indeed, every arbitrary induced substructure is obtained from a root-preserving induced substructure by possibly deleting its root, which is the definable least element.
Here we take induced substructures in the relational signature, the restricted order need not be induced by a spanning forest of the restricted graph.

Finally, a class closed under taking connected components is dependent if and only if its connected members are dependent, and the same holds for monadic dependence.
This follows from the first-order composition theorem due to Feferman and Vaught~\cite{feferman1959first} for disjoint unions and applies in particular to hereditary classes.
For relational structures, connectedness refers to the Gaifman graph of the full signature.
\end{remark}

\subsection{Basic finite model theory}
\subsubsection{Binary relational structures}

A \emph{binary relational signature} is a finite set of unary and binary relation symbols.  A \mbox{\emph{$\sigma$-structure}~$\mathbf M$} consists of a
finite set $M$, called its \emph{domain} or \emph{universe}, and an
interpretation \mbox{$R({\mathbf M})\subseteq M^r$} for every relation symbol
$R\in\sigma$ of arity $r\in\{1,2\}$.  Unary relations are also called
\emph{colors}. 
For every formula~$\varphi(\bar x)$, we write $\varphi(\mathbf M)
   :=\{\bar a\in M^{|\bar x|}:\mathbf M\models\varphi(\bar a)\}$.

The \emph{Gaifman graph} $\Gaif(\mathbf M)$ of a binary relational structure
$\mathbf M$ is the simple undirected graph with vertex set~$M$ in which two
distinct elements $u,v$ are adjacent if and only if there is a binary relation
symbol $R\in\sigma$ such that $(u,v)\in R({\mathbf M})$ or
$(v,u)\in R({\mathbf M})$.  

We regard a simple undirected graph $G$ as a structure with signature $\{E\}$,
universe $V(G)$, and with $E$ interpreted as the symmetric adjacency relation
of~$G$. 

Let $\mathbf M$ be a $\sigma$-structure and let $\tau\subseteq\sigma$.
The \emph{$\tau$-reduct} of $\mathbf M$, denoted by
$\mathbf M\!\upharpoonright_{\tau}$, is obtained by forgetting all relations
outside~$\tau$. If $\tau=\{R\}$, we abbreviate
$\mathbf M\!\upharpoonright_{\{R\}}$ by
$\mathbf M\!\upharpoonright_R$. Conversely, if
$\tau\subseteq\sigma$ and $\mathbf M$ is a $\tau$-structure, a
\emph{$\sigma$-expansion} of $\mathbf M$ is a $\sigma$-structure whose
$\tau$-reduct is~$\mathbf M$.  It is \emph{monadic} if all symbols in
$\sigma\setminus\tau$ are unary.  We write $\mathbf M_<$ for an expansion of
$\mathbf M$ by a linear order~$<$.  Finally, if $U\subseteq M$, then
$\mathbf M[U]$ denotes the substructure induced by~$U$.

Whenever a binary relational signature $\sigma$ contains a distinguished
binary relation symbol~$Z$, we consider only $\sigma$-structures $\mathbf M$
in which $Z(\mathbf M)$ is symmetric and irreflexive. Hence, the
$\{Z\}$-reduct~$\MZ$ is a simple undirected graph with vertex set~$M$ and edge relation~$Z(\mathbf M)$. Connectivity with
respect to~$Z$ always means connectivity in this graph.

\subsubsection{Logic and transductions}\label{subsec:logic-transductions}

Let $\sigma$ be a binary relational signature.
The set $\FO[\sigma]$ (or simply $\FO$ if $\sigma$ is clear from the context) of first-order formulas over~$\sigma$ is obtained from atomic formulas over~$\sigma$ by closing under the Boolean connectives and the quantifiers $\exists x$ and $\forall x$ ranging over elements.
Monadic second-order logic $\MSO$ extends first-order logic by quantifiers $\exists X$ and $\forall X$, ranging over sets of elements. 
The semantics are defined as usual and we refer to~\cite{libkin2004elements} for more background on logic. 

Let $\sigma$ and $\tau$ be finite binary relational signatures.
A \emph{simple interpretation} $\mathsf I$ of $\tau$-structures in
$\sigma$-structures is given by a $\sigma$-formula $\nu(x)$ and, for every
relation symbol $R\in\tau$ of arity $r\in\{1,2\}$, a $\sigma$-formula $\rho_R(x_1)$, or $\rho_R(x_1,x_2)$, respectively. For a $\sigma$-structure~$\mathbf M$ with
domain~$M$, the structure $\mathsf I(\mathbf M)$ has domain $\nu(\mathbf M)$
and interprets each such symbol~$R$ by
$ R(\mathsf I(\mathbf M))
 :=\rho_R(\mathbf M)\cap\nu(\mathbf M)^r$.
Unless another logic is specified, interpretations and transductions in this paper are first-order.

If $\Sigma$ is a finite set of unary relation symbols disjoint from~$\sigma$,
then a \emph{$\Sigma$-coloring} of a $\sigma$-structure~$\mathbf M$ is any
$(\sigma\cup\Sigma)$-expansion of~$\mathbf M$.  We write
$\Gamma_\Sigma(\mathbf M)$ for the set of all $\Sigma$-colorings
of~$\mathbf M$.

A \emph{copy operation} transforms a structure $\mathbf M$ with domain $M$
into a structure $\mathrm{Copy}_m(\mathbf M)$ with domain $M\times[m]$,
where $[m]=\{1,\dots,m\}$.  Each relation $R$ of arity~$r$ in the original
signature is interpreted component-wise, that is,
\[
R(\mathrm{Copy}_m(\mathbf M))
   :=
   \bigl\{((a_1,i),\dots,(a_r,i)) :
          (a_1,\dots,a_r)\in R(\mathbf M),\ i\in[m]\bigr\}.
\]
The copy operation also adds unary relation symbols $P_1,\ldots,P_m$ and a
binary relation symbol~$\simeq$, interpreted by
\[
 P_i(\mathrm{Copy}_m(\mathbf M)):=M\times\{i\}
 \quad\text{and}\quad
 (a,i)\simeq(b,j)\quad\Longleftrightarrow\quad a=b.
\]
Thus, the predicates $P_i$ distinguish the copies, while~$\simeq$ identifies
copies of the same original element.

A \emph{transduction} is a composition of copy operations, unary colorings, and simple interpretations. By the standard normalization~\cite{gajarsky2020first}, we may write it as $\mathsf T=(m,\Sigma,\mathsf I)$, where $m\geq1$, $\Sigma$ is a finite set of fresh unary predicates, and $\mathsf I$ is a simple interpretation. Thus,
\[
 \mathsf T(\mathbf M)=\{\mathsf I(\mathbf N):\mathbf N\in\Gamma_\Sigma(\mathrm{Copy}_m(\mathbf M))\}.
\]
We call this transduction \emph{non-copying} if $m=1$, and \emph{domain-preserving} if, in addition, its universe formula selects every input element.

We say that a structure $\mathbf N$ \emph{can be transduced from $\mathbf M$ by $\mathsf{T}$} if $\mathbf N\in\mathsf T(\mathbf M)$.
By extension, for a class $\Cc$ of structures, we define
\[
\mathsf T(\Cc)\is\bigcup_{\mathbf M\in\Cc}\mathsf T(\mathbf M).
\]
If $\Dd$ is a class of $\tau$-structures, we say that \emph{$\Dd$ can be transduced from $\Cc$} if there exists a transduction~$\mathsf T$ such that
$\Dd \subseteq \mathsf T(\Cc)$.
We remark that interpretations and transductions can be defined for all extensions of first-order logic that we consider in this paper. 

\medskip

We shall repeatedly carry auxiliary relations through a transduction.  The
following standard observation makes this precise.

\begin{fact}
    \label{fact:expand}
Let $\mathsf T$ be a non-copying, domain-preserving transduction from
$\sigma$-structures to $\tau$-structures, and let $\rho$ be a finite binary
relational signature disjoint from $\sigma\cup\tau$. Then $\mathsf T$
extends from $(\sigma\cup\rho)$-structures to
$(\tau\cup\rho)$-structures by carrying every relation in~$\rho$ over
unchanged.
\end{fact}

The next lemma says that, over a class excluding a fixed minor, all relational
information supported by the Gaifman graph can be encoded by a fixed
transduction. 

\begin{lemma}[special case of {\cite[Lemma~3.1]{lmcs_perm}}]
\label{lem:structure-from-gaif}
Let $\Cc$ be a class of graphs excluding a fixed minor and let $\sigma$ be a
finite binary relational signature.  There is a domain-preserving
transduction $\mathsf{Expand}_\sigma$ such that, for every $G\in\Cc$,
$\mathsf{Expand}_\sigma(G)$ contains every $\sigma$-structure $\mathbf M$
with domain~$V(G)$ and $\Gaif(\mathbf M)\subseteq G$.
\end{lemma}

We will multiple times use the following special case of the above lemma. 

\begin{corollary}
Let $\Cc$ be a class of graphs excluding a fixed minor.  There are fixed
transductions $\mathsf{Monotone}$ and $\mathsf{Orient}$ such that, for every
$G\in\Cc$, the class $\mathsf{Monotone}(G)$ contains every subgraph
of~$G$, and $\mathsf{Orient}(G)$ contains every orientation of~$G$.
Consequently every orientation of every subgraph of $G$ belongs to
$(\mathsf{Orient}\circ\mathsf{Monotone})(G)$.
\end{corollary}

\subsubsection{Dependence and monadic dependence}
Let $\mathcal L$ be a logic extending first-order logic.  An
$\mathcal L$-formula $\psi(\bar x;\bar y)$ is \emph{independent} in a class
$\Cc$ if, for every positive integer $n$, there are $\mathbf M\in\Cc$, tuples $\bar a_i\in M^{|\bar x|}$ for $i\in[n]$, and tuples $\bar b_J\in M^{|\bar y|}$ for $J\subseteq[n]$ such that
\[
\mathbf M\models \psi(\bar a_i,\bar b_J)\qquad\iff\qquad i\in J.
\]
The class $\Cc$ is \emph{$\mathcal L$-dependent} if no
$\mathcal L$-formula is independent in~$\Cc$.  It is \emph{monadically
$\mathcal L$-dependent} if the class of all expansions of members of~$\Cc$
by finitely many unary predicates is $\mathcal L$-dependent.  When the logic
is not displayed, dependence and monadic dependence refer to~$\FO$.

We mainly use the following characterization, which is specific here to
first-order logic and first-order transductions.
\begin{theorem}[\cite{baldwin1985second}]
A class $\Cc$ of structures is monadically dependent if and only if the class
of all finite graphs cannot be transduced from $\Cc$.
\end{theorem}

We note the following corollary, which immediately follows from the definition of transductions.

\begin{corollary}
    \label{lem:dep-preserved}
    Let $\Cc$ be a monadically dependent class of structures and let $\mathsf T$ be a transduction. 
    Then  $\mathsf T(\Cc)$ is  monadically dependent, too. 
\end{corollary}

For first-order logic, dependence and monadic dependence coincide on
hereditary classes of relational structures~\cite{braunfeld2022existential}.

\subsection{Monadic dependence of tree-ordered binary structures}
\subsubsection{Linearly ordered structures and twin-width}

We write $\operatorname{tww}(\mathbf M)$ for the twin-width of~$\mathbf M$. We do not use its definition directly, our arguments rely on the following characterization of ordered structures and its algorithmic consequence.

\begin{theorem}[store=thm:characterization-tww,note={\cite[Theorem 3]{bonnet2024twin}}]
    Let $\Cc$ be a class of finite linearly ordered binary structures over a finite signature. Then $\Cc$ has bounded twin-width if and only if $\Cc$ is monadically dependent.
\end{theorem}

\begin{theorem}[
store=thm:mc-tww,
note={\cite[Theorem~2]{bonnet2024twin} and
\cite[Theorem~1]{bonnet2022twwI}}]
There exist a computable function~$f$ and an absolute constant~$c$ such that,
given a finite linearly ordered binary structure $\mathbf M$ and an
$\FO$-sentence $\phi$, one can decide whether
$\mathbf M\models\phi$ in time
\[
 f\bigl(|\phi|,{\rm tww}(\mathbf M)\bigr)\cdot\|\mathbf M\|^c.
\]
\end{theorem}

For this statement, one first takes the reduct to the relation symbols
occurring in~$\phi$ together with the distinguished linear order.  Thus, the
dependence on the finite signature is absorbed by~$|\phi|$, while the
polynomial exponent remains absolute.

Having bounded twin-width is preserved by transduction, as stated in the next lemma.

\begin{lemma}[\cite{bonnet2022twwI}]
    \label{lem:tww-preserved}
    Let $\mathsf T$ be a fixed transduction. 
    There is a computable function $g$ such that the following holds. 
    If $\Cc$ is a class of binary structures with twin-width bounded by $d$, then $\mathsf T(\Cc)$ has twin-width bounded by $g(d)$. 
    In particular, every class of bounded twin-width is monadically dependent. 
\end{lemma}

The constructive version with the computable bound $g(d)$ follows from the constructive proof of \cite[Theorem~8.1]{bonnet2022twwI}. 

\subsubsection{Forest-ordered graphs excluding a minor}

The connection between forest-ordered classes with excluded minors and monadic dependence was made recently in~\cite{buffiere2026characterizations}.
A graph class $\Cc$ is \emph{weakly sparse} if there is a $t\in\mathbb{N}$ such that no graph in~$\Cc$ contains the complete bipartite graph $K_{t,t}$ as a subgraph.

\getkeytheorem{thm:minor-tows}

\pagebreak
We next consider minors of tree-ordered graphs.
Let~$(G,\preceq)$ be a tree-ordered graph (note that this tree order need not be related to the edges of~$G$).
An \emph{elementary $\prec$-contraction} identifies a non-root vertex $v$ with the unique element $u$ such that~$u \pred v$.
In the graph relation, the new vertex is adjacent to every vertex formerly adjacent to $u$ or~$v$ (loops are discarded and parallel edges are identified).
The order is the induced sub-order (i.e.\ the sub-order of $(V(G),\preceq)$ induced by $V(G)\setminus\{v\}$).
A \emph{$\prec$-contraction} is a possibly empty sequence of elementary $\prec$-contractions.
A \emph{tree-ordered minor} of $(G,\preceq)$ is any tree-ordered graph obtained by:
\begin{enumerate}[label=(\roman*),nosep]
  \item taking an induced tree-ordered subgraph of $(G,\preceq)$ that contains the root $r$;
  \item applying a $\prec$-contraction; and
  \item deleting an arbitrary set of edges from the graph relation, while retaining the resulting order.
\end{enumerate}
We write $\Minor_{\prec}(G,\preceq)$ for the set of all tree-ordered minors of $(G,\preceq)$, and for a class $\Cc$, we set $\Minor_{\prec}(\Cc) := \bigcup_{(G,\preceq)\in\Cc}\Minor_{\prec}(G,\preceq)$.
The following was shown in~\cite{buffiere2026characterizations}:

\begin{theorem}[{\cite[Theorem~1.3]{buffiere2026characterizations}}]\label{thm:bounded-degree-covergraph}
Let $D$ be an integer and let $\Cc$ be a class of tree-ordered weakly sparse graphs
whose $\prec$-reduct cover graphs have maximum degree at most $D$.
Then the following are equivalent:
\begin{enumerate}[label=(\arabic*),nosep]
  \item $\Cc$ is monadically dependent;
  \item $\Cc$ has bounded twin-width;
  \item the $E$-reduct of $\Minor_{\prec}(\Cc)$ is weakly sparse.
\end{enumerate}
\end{theorem}

The following is the constructive part of the  implication
$(3)\Rightarrow(2)$ in the proof of
\zcref{thm:bounded-degree-covergraph}.
In that proof, 
the ordered-grid bound, the finite chain-antichain bounds, and the coloring
bound for the conflict graph used there are effective. Tracing these bounds
yields a computable function of~$D$ and~$r$ only.

\begin{lemma}\label{lem:effective-bounded-degree-covergraph-original}
There exists a computable function $\beta$ such that the following holds.
Let $(G,\preceq)$ be a tree-ordered graph whose cover graph has maximum
degree at most~$D$.  If the $E$-reduct of every tree-ordered minor of
$(G,\preceq)$ excludes $K_{r,r}$ as a subgraph, then
\[
   \operatorname{tww}(G,\preceq)\leq\beta(D,r).
\]
\end{lemma}

\subsection{Tree decompositions}
\label{sec:tree-dec}
\subsubsection{Definitions and terminology}
A \emph{rooted tree decomposition} of a graph $G$ is a pair $\mathcal T=(T,\bag)$,
where $T$ is a rooted tree and
$\bag\colon V(T)\to 2^{V(G)}$
assigns to each node $x\in V(T)$ a set $\bag(x)\subseteq V(G)$, called the \emph{bag} of $x$, such that:
\begin{itemize}
    \item
    $
    V(G)=\bigcup_{x\in V(T)} \bag(x)$;
    
    \item for every edge $uv\in E(G)$, there exists $x\in V(T)$ with $\{u,v\}\subseteq \bag(x)$; and 
    \item for every vertex $v\in V(G)$, the set    
    $\{x\in V(T)\mid v\in \bag(x)\}$
    induces a non-empty connected subtree of $T$.
\end{itemize}

Let $\mathcal T=(T,\bag)$ be a rooted tree decomposition of a graph $G$.
Let $r$ be the root of $T$.
We denote by $\parent(\cdot)$ the parent function induced on $T$ and write $\parent(r)=\bot$ and set $\bag(\bot)=\emptyset$.

\pagebreak

For every node $x\in V(T)$, we define
\begin{itemize}
    \item its \emph{adhesion} $\adh(x):=\bag(x)\cap \bag(\parent(x))$, 
    \item its \emph{margin} $\mrg(x):=\bag(x)\setminus \adh(x)$, 
    \item its \emph{component} $\comp(x)=\bigcup_{y\succeq_T x}\mrg(y)$, and
    \item its \emph{torso} $\torso(x)$ as the graph on the vertex set $\bag(x)$ obtained from the induced subgraph $G[\bag(x)]$ by adding, for every neighbor $y$ of $x$ in $T$, all edges between vertices of
$\bag(x)\cap \bag(y)$. 
In other words, for every neighbor $y$ of $x$ in the decomposition tree, the adhesion $\bag(x)\cap \bag(y)$
is turned into a clique.
\end{itemize}

The \emph{adhesion} of $\mathcal T$ is the maximum size of an adhesion. 

\begin{definition}
\label{def:neat}
    A rooted tree decomposition $\mathcal T=(T,\bag)$ of a graph $G$ is \emph{neat} if
\begin{enumerate}[nosep]
    \item the margin $\mrg(x)$ of every node $x$ is non-empty,
    \item the margin $\mrg(x)$ of every node $x$ induces a connected subgraph of $G$,
    \item every vertex of the adhesion $\adh(x)$ of every node $x$ has a neighbor in $\mrg(x)$, and
    \item for every non-root node $x$, the adhesion of $x$ intersects the margin of its parent (i.e., $\adh(x)\cap\mrg(\parent(x))\neq\emptyset$).
\end{enumerate}    
\end{definition}

The following fact is standard and based on axioms of tree decompositions and the assumption that margins in neat tree decompositions are non-empty. 

\begin{fact}
\label{lem:margins-partition}
Let $\mathcal T=(T,\bag)$ be a neat tree decomposition of a graph $G$.
Then, for all nodes $x,y\in V(T)$,
we have 
$\mrg(x)=\mrg(y)$ if and only if 
$x=y$, and
 the family $
\{\mrg(x)\}_{x\in V(T)}$
is a partition of $V(G)$.
\end{fact}

Let $\mathcal T=(T,\bag)$ be a neat rooted tree decomposition  of a graph $G$.
For every non-root node $x\in V(T)$, we will later carefully choose a vertex
\[
\alpha(x)\in \adh(x)\cap \mrg(\parent(x)).
\]

Since every vertex of $\adh(x)$ has a neighbor in $\mrg(x)$, we can and will further fix a vertex
\[
\beta(x)\in \mrg(x)
\]
adjacent to $\alpha(x)$, which we call the \emph{vertex representing the margin of $x$}. 
For the root $r$, we fix an arbitrary vertex $\beta(r)\in \mrg(r)$. 
\medskip

It follows from \zcref{lem:margins-partition} that the map $\beta\colon V(T)\to V(G)$ is injective and that for every $v\in V(G)$, there exists a unique node $b(v)\in V(T)$ such that
\[
v\in \mrg(b(v)).
\]
The map $b:V(G)\rightarrow V(T)$ is called the \emph{placement map} of $\mathcal T$.

\subsubsection{Tree decompositions of graphs excluding a fixed minor}

We now state the algorithmic form of the Robertson--Seymour graph minor structure theorem while avoiding the technical details as much as possible. 

We use the usual bounded-width notion of an \emph{almost-embedding}; see
\cite[Section~5]{DujmovicMorinWood}.
For $k\in\mathbb N$, a surface~$\Sigma$, a graph~$G$, and a set
$A\subseteq V(G)$ with $|A|\leq k$, we say that $G$ is
\emph{$k$-almost embeddable in~$\Sigma$ with apex set~$A$} if $G-A$
consists of a graph embedded in~$\Sigma$ together with at most~$k$
attached pieces, called \emph{vortices}, each of width at most~$k$.
We leave the standard details to the cited definition.

We call an almost-embedding \emph{explicit} if its standard finite witness, including the associated decompositions where required, is supplied. 
We call an almost-embedding \emph{no-apex} if its apex set is empty. 
Finally, a graph is  \emph{$k$-almost embeddable} if it is $k$-almost embeddable in some surface of Euler genus at most~$k$.

\begin{theorem}[Robertson--Seymour; Grohe--Kawarabayashi--Reed]
\label{thm:rs-structure}
There is a computable function~$f$ and an algorithm with the following property. 
Given a graph~$J$ and a connected graph~$G$, the algorithm either outputs a model of $J$ as a minor of~$G$, or computes a supergraph $G'\supseteq G$ on the same vertex set as~$G$ and a tree decomposition $\mathcal T=(T,\bag)$ of~$G'$ of adhesion at most $k:=f(|J|)$ such that:
\begin{enumerate}[nosep]
    \item for every edge $xy\in E(T)$, the adhesion
    $\bag(x)\cap\bag(y)$ induces a clique in~$G'$, in particular, for every node $x\in V(T)$, the torso of~$x$ is exactly $G'[\bag(x)]$;
    \item for every node $x\in V(T)$, the algorithm computes a set $A_x\subseteq\bag(x)$ with $|A_x|\leq k$ and a graph~$P_x$ containing $G'[\bag(x)]-A_x$, together with an explicit no-apex $k$-almost-embedding of~$P_x$ in a surface of Euler genus at most~$k$.
\end{enumerate}
Moreover, the decomposition can be required to be neat. The running time is $f(|J|)\cdot |G|^2$.
\end{theorem}

This follows from the structure theorem
\cite[Theorem~3.2]{GroheKawarabayashiReed2013} and its algorithmic version
\cite[Theorem~6.1 and Section~7]{GroheKawarabayashiReed2013}.
The algorithm also supplies explicit almost-embeddings of the graphs~$P_x$; see~\cite[Proviso~3.1]{GroheKawarabayashiReed2013}.

For nonplanar excluded graphs, the underlying existence result goes back
to Robertson and Seymour~\cite{robertson2003graph}.
The additional neatness condition follows from the standard regularization with clique adhesions; see \cite[Lemma~2.8]{bojanczyk2016definability}. Each new bag is contained in an old one, so we restrict the corresponding apex set and retain the old no-apex almost-embedding as the containing graph~$P_x$.

Let $k\in\mathbb N$. 
An almost-embedding with apex set~$A$ is \emph{strong} if every neighbor outside~$A$ of an apex vertex belongs to one
of the attached pieces of the almost-embedding. Equivalently, apex vertices are adjacent only to other apex vertices and to vertices belonging to the
attached pieces. A graph is \emph{strongly $k$-almost embeddable} in a surface~$\Sigma$ if it admits a strong $k$-almost-embedding in~$\Sigma$. It is \emph{strongly $k$-almost embeddable} if this holds for some surface of Euler genus at most~$k$.

Let $G_1$ and $G_2$ be graphs and let $K_1\subseteq G_1$ and $K_2\subseteq G_2$ be cliques of the same size.
A \emph{clique-sum} of $G_1$ and $G_2$ is any graph obtained by identifying pairs of vertices in $K_1$ and $K_2$ and optionally deleting some edges of the identified clique. 
A particular case is when $K_1$ and $K_2$ are single vertices~$v_1$ and $v_2$ of $G_1$ and $G_2$, in which case the clique-sum consists of gluing $G_1$ and $G_2$ by identifying $v_1$ and $v_2$. 
Such a clique-sum is called a \emph{$1$-clique-sum}.

More generally, a graph is obtained by \emph{clique-sums} from a family of graphs if it can be constructed by repeatedly applying the clique-sum operation.

A graph $G$ is an \emph{apex graph} if there exists a vertex $v\in V(G)$ such that $G-v$ is planar.
The following theorem of Dujmović, Morin, and Wood characterizes proper minor-closed classes of graphs excluding an apex graph.

\begin{theorem}[Dujmović-Morin-Wood \cite{DujmovicMorinWood}]
\label{thm:dmw}
Let $\Cc$ be a proper minor-closed class of graphs.
Then the following are equivalent:
\begin{enumerate}[nosep]
    \item $\Cc$ excludes an apex graph as a minor;
    \item there exists $k\in \mathbb N$ such that every graph in $\Cc$ can be constructed by clique-sums of strongly $k$-almost embeddable graphs;
\end{enumerate}
\end{theorem}

We use the following no-apex consequence of this characterization; see~\cite[Theorem~2.2]{BonnetHickingbotham2025}.

\begin{corollary}[\cite[Theorem~2.2]{BonnetHickingbotham2025}]
\label{cor:bloc}
Let $\Cc$ be a class of graphs and let $k\in\mathbb N$.
If every maximal $2$-connected subgraph with at least three vertices of every
graph in~$\Cc$ is a subgraph of a no-apex $k$-almost embeddable graph, then
$\Cc$ excludes an apex graph as a minor.
\end{corollary}

The following result for classes excluding an apex graph was proved by
Milolidakis~\cite{MilolidakisThesis}. Its polynomial-time construction
also uses the strengthened structure theorem of Dvo\v{r}\'ak and
Thomas~\cite{dvorak2014list}.

\begin{theorem}[{\cite[Theorem 18]{MilolidakisThesis}}
    and \cite{dvorak2014list}]
\label{thm:milolidakis}
Let $\Cc$ be a class of graphs excluding an apex graph as a minor.
Then there exists a class $\Cc'$ of subcubic graphs excluding an apex
graph as a minor such that every graph $G\in\Cc$ is a minor of some
graph $H\in\Cc'$.

Furthermore, given $G\in\Cc$, such a graph $H\in\Cc'$ and a corresponding surjective minor model can be computed in polynomial time\footnote{The polynomial-time claim
uses the algorithmic structure theorem in~\cite{dvorak2014list}.
For the uniform form used below, suppose that $G$ is connected and,
for each block with at least three vertices, an explicit no-apex
$k$-almost-embedding of a containing graph is supplied.
Then $H$ can be chosen connected with $h(H)\leq q(k)$, and $H$ and
the minor-model map can be computed in time $p(k)N^{c_0}$,
where $p,q$ are computable functions, $c_0$ is an absolute constant,
and $N$ is the total encoding size of $G$ and the supplied witnesses.}.
\end{theorem}

\section{Rearranging tree orders}
\label{sec:technical}

We present two useful general transductions allowing to rearrange tree orders.

\subsection{Rewiring a tree order}

We introduce a small transduction that will be used several times below to modify a tree order locally.  
The intuitive operation is the following. 
We mark some vertices as roots of subtrees that should no longer stay attached to their old parents.  
In each relevant part of the tree, we also mark one attachment vertex.  
The operation then cuts every marked root from its old parent and attaches it as a new child of the corresponding attachment vertex.

Let $\mathbf X=(X,\preceq)$ be a tree-ordered set.  
We use the unary marks $S$, $A$, $R_{\rm old}$ with the following intended meaning.
The vertices marked by $S$ are called \emph{stop vertices}.  
They determine the parts of the tree in which rewiring is performed. 
Formally, we require that $S$ is an antichain. 
For $s\in S$, we write $U_s := \{x\in X : s\preceq x\}$.
Note that the sets $U_s$ are pairwise-disjoint.  
We call them the \emph{universes} of the rewiring.

The vertices marked by $A$ are called \emph{attachment vertices}. 
We require that every universe $U_s$ contains exactly one vertex marked by $A$. 
We denote this vertex by $a_s$.

The vertices marked by $R_{\rm old}$ are the old roots of the pieces that
will be moved. We require that every $r\in R_{\rm old}$ belongs to a unique
universe $U_s$ and that $\neg(r\preceq a_s)$. In particular, $r\neq a_s$,
the vertex~$r$ is not a root of the original tree, and $a_s$ does not belong
to the piece rooted at~$r$. Thus, deleting the cover edge from the parent of
$r$ to~$r$ and adding the edge $a_sr$ is well defined and creates no cycle.

The rewired tree is obtained from the cover tree of $\preceq$ as follows. 
For each $r\in R_{\rm old}\cap U_s$, delete the edge between $r$ and its old parent and add the edge between $a_s$ and $r$. 
All other cover edges are kept.  
Note that the vertices marked by $R_{\rm old}$ need not form an antichain.  
If two marked vertices lie on the same root-to-leaf path, then both are cut from their old parents and both are attached separately below the same attachment vertex.
Thus, the lower marked vertex is no longer considered part of the piece moved with the higher marked vertex.

A tree-ordered set obtained in this way from $\mathbf X$ is called a \emph{rewiring} of $\mathbf X$.

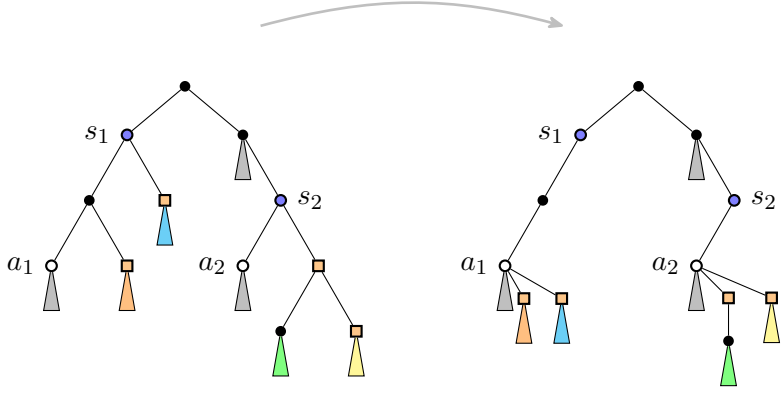
\begin{figure}[ht]
    \centering
       \begin{tikzpicture}
    \tikzset{node/.style={draw, circle, fill = black, minimum size = 3pt, inner sep=0pt, line width=1pt}}
    \tikzset{A/.style={draw=black, circle, fill = white, minimum size = 4pt,inner sep=0pt,line width=0.8pt}}
    \tikzset{S/.style={draw=black, circle, fill = white!50!blue, minimum size = 4pt,inner sep=0pt,line width=0.8pt}}
    \tikzset{R/.style={draw,line width=0.8pt,fill = yellow!50!red!50!white, minimum size = 4pt, inner sep=0pt}}

    \begin{scope}
        \node[node] (r) at (0,0) {};
        \node[label={left:$s_1$},S] (s1) at ($(r)+(220:1)$) {};
        \node[R] (r2) at ($(s1)+(-60:1)$) {};
        \node[node] (r2sibling) at ($(s1)+(240:1)$) {};
        \node[label={left:$a_1$},A] (a1) at ($(r2sibling)+(240:1)$) {};
        \node[R] (r1) at ($(r2sibling)+(-60:1)$) {};
        
        \node[node] (s1sibling) at ($(r)+(-40:1)$) {};
        \node[label={right:$s_2$},S] (s2) at ($(s1sibling)+(-60:1)$) {};
        \node[label={left:$a_2$},A] (a2) at ($(s2)+(240:1)$) {};
        \node[R] (r3) at ($(s2)+(-60:1)$) {};
        \node[R] (r4) at ($(r3)+(-60:1)$) {};
        \node[node] (r4sibling) at ($(r3)+(240:1)$) {};

        \draw (r) -- (s1) (r)-- (s1sibling);
        \draw (s1) -- (r2) (s1) -- (r2sibling);
        \draw (r2sibling) -- (a1) (r2sibling) -- (r1);
        \draw (s1sibling) -- (s2);
        \draw (s2) -- (a2) (s2) -- (r3);
        \draw (r3) -- (r4) (r3) -- (r4sibling);

        \draw[fill=gray!50!white] (a1) -- ($(a1)+(260:0.6)$) -- ($(a1)+(-80:0.6)$) -- (a1);
        \draw[fill=orange!50!white] (r1) -- ($(r1)+(260:0.6)$) -- ($(r1)+(-80:0.6)$) -- (r1);
        \draw[fill=cyan!50!white] (r2) -- ($(r2)+(260:0.6)$) -- ($(r2)+(-80:0.6)$) -- (r2);
        \draw[fill=gray!50!white] (a2) -- ($(a2)+(260:0.6)$) -- ($(a2)+(-80:0.6)$) -- (a2);
        \draw[fill=gray!50!white] (s1sibling) -- ($(s1sibling)+(260:0.6)$) -- ($(s1sibling)+(-80:0.6)$) -- (s1sibling);
        \draw[fill=yellow!50!white] (r4) -- ($(r4)+(260:0.6)$) -- ($(r4)+(-80:0.6)$) -- (r4);
        \draw[fill=green!50!white] (r4sibling) -- ($(r4sibling)+(260:0.6)$) -- ($(r4sibling)+(-80:0.6)$) -- (r4sibling);
    \end{scope}

    \begin{scope}[xshift=6cm]
        \node[node] (r) at (0,0) {};
        \node[label={left:$s_1$},S] (s1) at ($(r)+(220:1)$) {};
        \node[node] (r2sibling) at ($(s1)+(240:1)$) {};
        \node[label={left:$a_1$},A] (a1) at ($(r2sibling)+(240:1)$) {};
        \node[R] (r1) at ($(a1)+(-60:0.5)$) {};
        \node[R] (r2) at ($(r1)+(0.5,0)$) {};

        \node[node] (s1sibling) at ($(r)+(-40:1)$) {};
        \node[label={right:$s_2$},S] (s2) at ($(s1sibling)+(-60:1)$) {};
        \node[label={left:$a_2$},A] (a2) at ($(s2)+(240:1)$) {};
        \node[R] (r3) at ($(a2)+(-45:0.6)$) {};
        \node[R] (r4) at ($(r3)+(0:0.57)$) {};
        \node[node] (r4sibling) at ($(r3)+(270:0.57)$) {};

        \draw (r) -- (s1) (r)-- (s1sibling);
        \draw (a1) -- (r2) (s1) -- (r2sibling);
        \draw (r2sibling) -- (a1) (a1) -- (r1);
        \draw (s1sibling) -- (s2);
        \draw (s2) -- (a2) (a2) -- (r3);
        \draw (a2) -- (r4) (r3) -- (r4sibling);

        \draw[fill=gray!50!white] (a1) -- ($(a1)+(260:0.6)$) -- ($(a1)+(-80:0.6)$) -- (a1);
        \draw[fill=orange!50!white] (r1) -- ($(r1)+(260:0.6)$) -- ($(r1)+(-80:0.6)$) -- (r1);
        \draw[fill=cyan!50!white] (r2) -- ($(r2)+(260:0.6)$) -- ($(r2)+(-80:0.6)$) -- (r2);
        \draw[fill=gray!50!white] (a2) -- ($(a2)+(260:0.6)$) -- ($(a2)+(-80:0.6)$) -- (a2);
        \draw[fill=gray!50!white] (s1sibling) -- ($(s1sibling)+(260:0.6)$) -- ($(s1sibling)+(-80:0.6)$) -- (s1sibling);
        \draw[fill=yellow!50!white] (r4) -- ($(r4)+(260:0.6)$) -- ($(r4)+(-80:0.6)$) -- (r4);
        \draw[fill=green!50!white] (r4sibling) -- ($(r4sibling)+(260:0.6)$) -- ($(r4sibling)+(-80:0.6)$) -- (r4sibling);
    \end{scope}

    \draw[draw=gray!50!white,line width=1pt,-{Stealth[round]}] (1,0.8) [bend left=15] to (5,0.8);
   \end{tikzpicture}
    \caption{Rewiring some subtrees. The vertices $s_1,s_2$ are the stop vertices and the vertices $a_1,a_2$ are the corresponding attachment vertices. The vertices $R_{\rm old}$ are the four square vertices.}
    \label{fig:rewiring}
\end{figure}

\begin{lemma}\label{lem:rewire}
    There is a transduction $\rew$ such that 
    for every tree-ordered set $\mathbf X$, $\rew(\mathbf X)$ contains all the rewirings of~$\mathbf X$.
\end{lemma}
\begin{proof}
The transduction guesses the unary predicates $S,A,R_{\rm old}$. 
We first note some easily checked
definability results:
\begin{itemize}
    \item for an element~$u$ belonging to some universe, the unique stop
    vertex~$s$ such that $u\in U_s$ is definable;
    \item for every stop element $s$, the unique attachment vertex $a\in U_s$ is definable;
    \item for every element $u$ marked $R_{\rm old}$, the following set is definable: the set $X_u$ of all vertices $v$ with $u\preceq v$ such that no vertex $u'\neq u$ marked $R_{\rm old}$ satisfies $u\prec u'\preceq v$. These are the vertices below $u$ that shall be moved with $u$.
\end{itemize}
\medskip
The formula $\rho_\preceq(x,y)$ of the interpretation is now defined depending on several definable cases:
\begin{itemize}
    \item If $x$ and $y$ are not in the same universe, then $\rho_\preceq(x,y)\is (x\preceq y)$.
    \item Otherwise, let $s$ be the stop vertex with $x,y\in U_s$ and let $a$ be the only attachment vertex in~$U_s$. 
    \begin{itemize}
    \item If $x$ and $y$ are in the same set $X_u$, then $\rho_\preceq(x,y)\is (x\preceq y)$.
    \item If $x$ and $y$ are in different sets ($x \in X_u$ and $y \in X_v$), then $\rho_\preceq(x,y)\is \bot$.
    \item If $x$ is in a set $X_u$ and $y$ is not in any set $X_v$, then $\rho_\preceq(x,y)\is \bot$.
    \item If $x$ is not in any set $X_u$ and $y$ is in a set $X_v$, then $\rho_\preceq(x,y)\is (x  \preceq a)$.
    \item If neither $x$ nor $y$ are in any set $X_u$, then $\rho_\preceq(x,y)\is (x\preceq y)$.\qedhere
    \end{itemize}
\end{itemize}
\end{proof}

\subsection{Re-rooting a tree order}

We next define a transduction that allows us to re-root selected subtrees of a
tree order.

The intuitive operation is the following. 
We mark an antichain of vertices by a unary predicate~$R$.  
These vertices are the old roots of the subtrees that we want to modify.  
For every such old root $r$, we also mark exactly one descendant $n$ of $r$ by a unary predicate $N$. 
The subtree rooted at $r$ is then re-rooted at $n$: the path from $r$ to $n$ is reversed, all branches hanging from this path stay attached to the same vertex as before, and the whole subtree stays attached to the rest of the tree at the same place; see \zcref{fig:reroot}. 
The corresponding operation on (associated) tree orders is called a \emph{re-rooting}.

\begin{figure}[ht]
    \centering
    \includegraphics[width=.75\columnwidth]{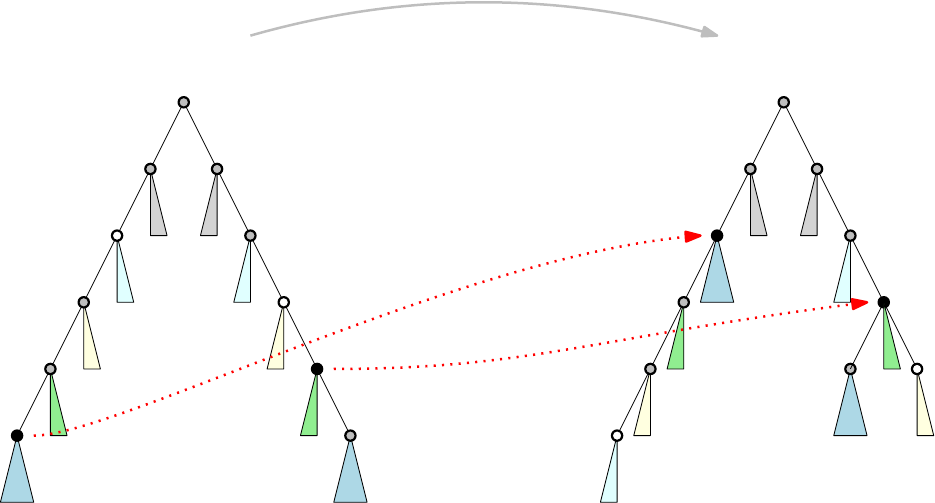}
    \caption{Re-rooting some subtrees. Old roots marked with $R$ are white, new roots marked with $N$ are black.}
    \label{fig:reroot}
\end{figure}

\begin{lemma}\label{lem:reroot-arbitrary}
There is a transduction $\mathsf T$ such that every re-rooting of a tree-ordered set $\mathbf X$ belongs to $\mathsf T(\mathbf X)$.
\end{lemma}
\begin{proof}
The transduction guesses the unary predicates $R$ and $N$.   
We then define, by a first-order interpretation, the new tree order.

Let $Y$ be the original rooted tree and let $Y^{\rm new}$ be the new one.
For a vertex $v$, denote by $P_v$ (resp.\ $P_v^{\rm new}$) the path from the root 
to $v$ in $Y$ (resp. in $Y^{\rm new}$).

Then it is easily checked that 
\begin{itemize}
    \item if $P_v$ does not contain a vertex marked $R$, then $P_v^{\rm new}=P_v$;
    \item if $P_v$ contains a (unique) vertex $r_v$ marked $R$ and if $n_v$ is the (unique) descendant of $r_v$ marked~$N$, then (see \zcref{fig:P_rerooted})
    \[
    V(P_v^{\rm new})=V(P_v)\mathbin{\Delta}\{z\colon (r_v\preceq z\preceq n_v)\wedge(z\neq (v\curlywedge n_v))\}.
    \]
\end{itemize}
\begin{figure}[h!t]
    \centering
    \includegraphics[width=0.5\linewidth]{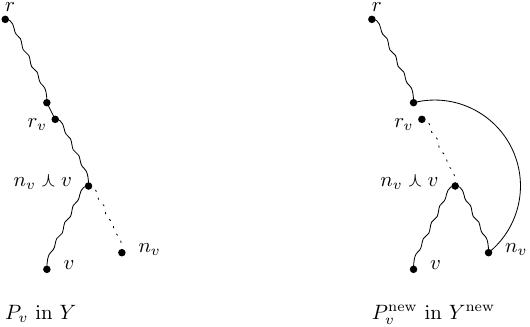}
    \caption{How the path from the root to $v$ changes when re-rooting the tree order (case where $P_v$ contains a vertex marked $R$).}
    \label{fig:P_rerooted}
\end{figure}

\pagebreak
We first define the auxiliary formula

\[\zeta(x,y,z,z') := 
    R(z)\wedge N(z')\wedge (z\preceq x\preceq z')\wedge (z\preceq y),\] which expresses that $x$ and $y$ are in the same subtree with root $z$ marked $R$ and $x$ is on the path linking the root to the vertex $z'$ marked $N$ in this subtree. (Hence, $z=r_y$ and $z'=n_y$.)

It follows from the above case analysis that if there exist no~$z,z'$ with
$\zeta(x,y,z,z')$, then $x\preceq y$ holds in the re-rooted tree order if
and only if it holds in the original one. If there exist~$z,z'$ with
$\zeta(x,y,z,z')$, then necessarily $z=r_y$ and $z'=n_y$, and
$x\preceq y$ holds in the re-rooted tree order if and only if
$y\curlywedge n_y\preceq x$ holds in the original tree order. Consequently, the new tree order is
defined by
\[
\rho_\preceq(x,y)\is
\bigl(\neg\exists z,z'\,\zeta(x,y,z,z')\wedge x\preceq y\bigr)
\vee
\bigl(\exists z,z'\,
   (\zeta(x,y,z,z')\wedge(y\curlywedge z'\preceq x))\bigr).
\]

It is easily checked that the transduction
$\mathsf T := (\{R,N\},\mathsf I)$
 meets the requirements of the lemma.
\end{proof}

\section{Proof of the main structure theorem}\label{sec:proof-main-thm}

This section is devoted to the proof of the structure theorems, which we restate here. 
\begin{theorem}[store=thm-main-structure,restate-keys={note=see \zcpageref[nocap]{thm-main-structure}}]
Let $\Cc$ be a weakly sparse class of graphs. Then the following are equivalent:
\begin{enumerate}[nosep]
    \item $\Cc$ excludes some graph as a minor;
    \item $\Cc^\Upsilon$ has bounded twin-width;
    \item $\Cc^F$ has bounded twin-width.
\end{enumerate}
\medskip 

Moreover, if these conditions hold, an expansion of $\Cc^F$ by a linear order with bounded twin-width can be computed in polynomial time. 
\end{theorem}

\begin{theorem}[store=thm-main-algo,restate-keys={note=see \zcpageref[nocap]{thm-main-algo}}]
There exist computable functions~$\xi,f$, a constant~$c$ and an algorithm that, for a given graph $G$ and a rooted spanning forest $F$ of $G$, computes in time $f(h(G))\cdot|G|^c$ a linear order $<$ on $G$, such that $(G,\preceq_F,<)$  has twin-width  at most $\xi(h(G))$.
\end{theorem}

By \zcref{rem:forest-tree-reduction}, it suffices to establish the following statement for connected graphs and spanning trees.

\begin{theorem}[store=thm-aux-structure,restate-keys={note=see \zcpageref[nocap]{thm-aux-structure}}]
Let $\sigma=\{E,\preceq,<\}$ be the signature obtained by expanding the signature $\{E\}$ of graphs with  two binary relations $\preceq$ and $<$.
There exist computable functions~$\xi,f$, a constant~$c$, and an algorithm that, for a given connected graph $G$ and a rooted spanning tree $Y$ of $G$, computes in time $f(h(G))\cdot|G|^c$ a $\sigma$-structure~$\mathbf G$ such that 
\begin{itemize}[nosep]
    \item the $\{E\}$-reduct of $\mathbf G$ is $G$,
    \item the $\preceq$-reduct of $\mathbf G$ is the tree order $\preceq_Y$,
    \item the $<$-reduct of $\mathbf G$ is a linear order,
    \item the twin-width of $\mathbf G$ is at most $\xi(h(G))$.
\end{itemize}
\end{theorem}

\begin{proof}[Proof that {\zcref[cap]{thm-aux-structure}} implies
{\zcref[cap]{thm-main-structure}} and {\zcref[cap]{thm-main-algo}}]
\label{pf:th1}

Let $\xi_0,f_0$, and $c_0$ be as provided by \zcref{thm-aux-structure}, and let $G$ be a graph equipped with a rooted spanning forest~$F$.
Apply \zcref{rem:forest-tree-reduction} to obtain a connected graph $\widehat G=G\oplus K_1$ with a rooted spanning tree~$Y$.
Thus, $h(\widehat G)\leq h(G)+1$, $|\widehat G|=|G|+1$, and restricting $\preceq_Y$ to $V(G)$ gives~$\preceq_F$.
Apply \zcref{thm-aux-structure} to~$(\widehat G,Y)$ and restrict the resulting linear order to~$V(G)$.
Since twin-width does not increase when passing to an induced substructure, the resulting structure $(G,\preceq_F,<)$ has twin-width at most $\xi_0(h(\widehat G))$.
After increasing~$f_0$ and $c_0$ if necessary to include the construction and restriction steps, define
\[
 \xi(k)=\max_{0\leq i\leq k+1}\xi_0(i),
 \qquad
 f(k)=2^{c_0}\max_{0\leq i\leq k+1}f_0(i).
\]
Then $\operatorname{tww}(G,\preceq_F,<)\leq\xi(h(G))$.
Moreover, $(|G|+1)^{c_0}\leq2^{c_0}|G|^{c_0}$, so the running time is at most $f(h(G))\cdot|G|^{c_0}$.
This proves \zcref{thm-main-algo}.

We now prove \zcref{thm-main-structure}.
The implication from (2) to (1) follows from
\zcref{thm:minor-tows}.
For the converse, suppose that $\Cc$ excludes a graph~$H$ as a minor,
and put $t=|V(H)|$.
Then $h(G)<t$ for every $G\in\Cc$. 
By \zcref{thm-main-algo}, for every $G\in\Cc$ and every rooted spanning
forest~$F$ of~$G$, there is a linear order~$<$ such that
$\operatorname{tww}(G,\preceq_F,<)\leq \xi(t-1)$.
Finally, twin-width does not increase when relations are removed.
Hence, $\operatorname{tww}(G,\preceq_F)
 \leq \operatorname{tww}(G,\preceq_F,<)
 \leq \xi(t-1)$.
Thus, $\Cc^F$ has bounded twin-width, proving (1) implies (3).
The implication from (3) to (2) follows from $\Cc^{\Upsilon}\subseteq\Cc^F$.
\end{proof}

\zcref{thm-main-algo} extends to binary structures.

\begin{theorem}[store=thm-main-algo-structures]
Fix a finite relational signature~$\sigma$ of arity at most two.
There exist computable functions~$f_\sigma,\xi_\sigma$, an absolute
constant~$c$, and an algorithm that, given a $\sigma$-structure
$\mathbf M$ and a rooted spanning forest~$F$ of
$G:=\Gaif(\mathbf M)$, computes in time
$f_\sigma(h(G))\cdot\|\mathbf M\|^c$ a linear order~$<$ such that
$(\mathbf M,\preceq_F,<)$ has twin-width at most
$\xi_\sigma(h(G))$.
\end{theorem}

\begin{proof}
Apply~\zcref{thm-main-algo} to $G:=\Gaif(\mathbf M)$ and~$F$.
We recover all relations of~$\mathbf M$ from a suitable unary expansion of~$G$ by first-order formulas,
while retaining both orders unchanged. This is the expansion
construction of~\zcref{lem:structure-from-gaif}, together
with~\zcref{fact:expand}. In its effective form, the number of
unary predicates and the defining formulas depend computably
only on $h(G)$ and~$\sigma$, and the required predicates can
be computed within the stated running time. Preservation of
bounded twin-width under first-order transductions then gives
the required twin-width bound.
\end{proof}

As $\FO$ model checking is known to be fixed-parameter tractable on monadically dependent classes of linearly ordered binary structures (\zcref{thm:mc-tww}), we obtain the following corollary.

\begin{corollary}[store=cor-model-checking,restate-keys={note=see \zcpageref[nocap]{cor-model-checking}}]
Fix a finite relational signature~$\sigma$ of arity at most two,
and let $\preceq$ be a fresh binary relation symbol.
There are a computable function~$f$ and an absolute
constant~$c$ such that the following holds.
Given a finite $\sigma$-structure~$\mathbf M$, a rooted spanning
forest~$F$ of $\Gaif(\mathbf M)$, and an $\FO$-sentence~$\varphi$
over $\sigma\cup\{\preceq\}$, one can decide whether
\[
   (\mathbf M,\preceq_F)\models\varphi
\]
in time
\[
   f\bigl(|\varphi|,h(\Gaif(\mathbf M))\bigr)
   \cdot\|\mathbf M\|^c.
\]
\end{corollary}

As outlined in \zcref{sec:outline_struc}, the proof 
of \zcref{thm-aux-structure}
consists of $9$ steps divided into $2$ phases.
For convenience, we reproduce the overview of the proof of \zcref{thm-aux-structure} and the short description of the two phases given in \zcref{sec:outline_struc}.

\begin{figure}[ht]
\[
\scalebox{.96}{
\begin{xy}
\xymatrix{
&&\mathcal T=(T,\bag)\ar@{..>}[dr]^{2}\\
\text{Phase I}&(G,Y)\ar@{..>}[ur]^(.4){1}\ar@{..>}[r]^(.45){1}&(G_1,Y)\ar@{..>}[r]^(.45){2}&(G_2,Y_2)\ar@{..>}[r]^(.45){3}&(H,Y_H)\ar@{..>}[r]^{4}&(H,\preceq_{Y_H})\ar@{..>}@/^/[dd]_{5}\ar@/_/[dd]_{\mathsf T_1}\\
\ar@{--}[rrrr]&&&&&\\
\text{Phase II}&(G,\preceq_Y,<_G)&{\mathbf N}\ar[l]_{9}^{\mathsf T_5}&{\mathbf M}\ar[l]_{8}^{\mathsf T_4}&(G_2,\preceq_{Y_2},<_{G})\ar[l]_(.55){7}^(.55){\mathsf T_3}& (H,\preceq_{Y_H},<_H)\ar[l]_(.45){6}^(.45){\mathsf T_2}
}
\end{xy}
}
\]
\caption{Outline of the proof of \zcref{thm-aux-structure} (copy of \zcref{fig:overview}). Dotted arrows are polynomial-time computations and plain arrows are transductions. (Precisely, $G\xrightarrow[\mathsf T]{~1~}H$ means that at step 1 of the proof, we prove $H\in\mathsf T(G)$.)
The main steps of the proof are:
    {\bf 1.}~construction of a neat tree decomposition of~$G$ based on the structure theorem of Robertson and Seymour, as well as an augmentation of $G$ so that every adhesion is a clique;
    {\bf 2.}~construction of a subgraph whose nontrivial blocks have explicit
    no-apex almost-embedding witnesses;
    {\bf 3.}~effective blockwise translation into a subcubic graph of bounded
    Hadwiger number;
    {\bf 4.}~computation of the tree-ordered expansion;
    {\bf 5.}~computation of a linear order that can be obtained by transduction $\mathsf T_1$;
    {\bf 6.}~transduction of the linearly ordered tree-ordered expansion of $G_2$ by $\mathsf T_2$;
    {\bf 7.}~transduction of a master structure that permits the graph to be handled margin-wise by~$\mathsf T_3$;
    {\bf 8.}~recovery of the target tree order on the margins by transduction $\mathsf T_4$;
    {\bf 9.}~full recovery of the tree order by transduction $\mathsf T_5$.}
    \label{fig:overview_rest}
\end{figure}

\begin{itemize}
    \item {\bf Phase I.} Given a connected graph $G\in \Hh_t$ with a rooted spanning tree $Y$, we construct in polynomial time a tree-ordered graph $(H,\preceq_{Y_H})$ with a linearly ordered expansion~$(H,\preceq_{Y_H},<_H)$ of~$(H,\preceq_{Y_H})$ such that the class
    \[\mathscr M_t = \{ (H,\preceq_{Y_H},<_H) \mid G\in \mathscr H_t,
    \text{$Y$ is a rooted spanning tree of $G$}\}\]
    is monadically dependent, and hence, as a linearly ordered class, has bounded twin-width.
    
    \item {\bf Phase II.} From $(H,\preceq_{Y_H},<_H)$ we transduce a linearly ordered
  expansion $(G,\preceq_Y,<_G)$ of the original tree-ordered graph. The order $<_G$ is
  inherited from the order $<_H$.
\end{itemize}

\subsection{Phase I}
\paragraph{Step 1}
We first compute the Hadwiger number $h(G)$ of $G$ in time $f(h(G))\cdot|G|^{1+o(1)}$ using the minor containment testing algorithm given in \cite{korhonen2024minor}, where $f$ is some computable function.

We now set $J=K_{h(G)+1}$ and call the algorithm of \zcref{thm:rs-structure}.
For this choice of $J$, the algorithm outputs a neat tree decomposition
$\mathcal T=(T,\bag)$ of a supergraph $G_1$ of~$G$, together with apex sets
$A_x$ for all $x\in V(T)$. Each adhesion induces a clique in~$G_1$, and
$G_1[\bag(x)]-A_x$ is contained in a graph~$P_x$ supplied with an explicit
no-apex $k$-almost-embedding, where $k$ is the parameter from
\zcref{thm:rs-structure}.

We shall need later that the supergraph $G_1$ still belongs to a
minor-excluded class whose excluded minor depends only on~$h(G)$.
Every torso at a node~$x$ is contained in the $k$-almost embeddable graph
obtained from~$P_x$ by adding $A_x$ as apices and all their incident edges
in the torso. Hence, every torso has Hadwiger number bounded by a computable
function of~$k$, see~\cite[Theorem~3.1]{JoretWood2013}.
Moreover, $G_1$ is obtained from its torsos by clique-sums, because all
adhesions are cliques. Since the Hadwiger number of a clique-sum is at most
the maximum of the Hadwiger numbers of its summands, see
\cite[Lemma~2.3]{JoretWood2013}, it follows that
$h(G_1)\leq q_1(h(G))$ for a computable function~$q_1$.

Our next goal is to define the vertices $\alpha(x)$ and $\beta(x)$ for all nodes $x$ of the decomposition, which will represent these nodes and their connections. 
For this, we compute an acyclic orientation of $G_1$ as follows. 
Recall that a graph is \emph{$d$-degenerate} if every non-empty subgraph contains a vertex of degree at most $d$. 
Equivalently, its vertices admit a degeneracy ordering, and such an ordering yields an acyclic orientation of the graph with maximum indegree at most $d$.  
Moreover, if $d$ is fixed, such an orientation can be computed in linear time by a standard greedy algorithm.

Let us bound the degeneracy inside the margins.
For every integer $q$, there is a (computable) number $d_{\rm ae}(q)$ such that every graph that is $q$-almost embeddable in a surface of Euler genus at most $q$ is $d_{\rm ae}(q)$-degenerate.
Indeed, as mentioned above, such graphs exclude a clique minor whose order is bounded by a function of~$q$. 
A theorem of Kostochka and Thomason bounds the average degree of
$K_t$-minor-free graphs~\cite{Th2}, which hence bounds the degeneracy of these graphs.
Since every torso is a subgraph of a $k$-almost embeddable graph as above,
it is $d_1$-degenerate for $d_1:=d_{\rm ae}(k)$.

Since the graph induced by a margin is a subgraph of the corresponding torso,
every margin~$\mrg(x)$ induces a $d_1$-degenerate graph in $G_1$.

We now orient the edges of $G_1$.  For every node $x\in V(T)$, we compute a degeneracy ordering of the graph $G_1[\mrg(x)]$ and orient all edges inside $\mrg(x)$ according to this ordering, so that every vertex has at most $d_1$ incoming edges from its own margin.

It remains to orient the edges whose endpoints lie in different margins.
Let $u\in\mrg(x)$ and $v\in\mrg(y)$, where $x\neq y$.  
Since $\mathcal T$ is a tree decomposition, the nodes $x$ and $y$ are comparable in the rooted tree order of $T$ whenever $uv$ is an edge of $G_1$. 
If $x\prec_T y$, then we orient the edge from $u$ to $v$; otherwise we orient it from $v$ to $u$.  
In other words, all cross-margin edges are oriented down the decomposition tree.

The resulting orientation $\vec G_1$ is acyclic.  
Indeed, inside each margin the orientation is acyclic by construction, while every edge between two different margins goes strictly down the tree order on $T$.  
Hence, no directed cycle can use a cross-margin edge, and no directed cycle can be contained in a single margin.

Finally, we bound the indegree.  
A vertex has at most $d_1$ incoming edges from its own margin.  
Its remaining incoming edges come from ancestor margins.  
By the properties of tree decompositions, all such ancestors that are adjacent to a vertex in $\mrg(x)$ must lie in the adhesion through which the bag $x$ is attached to its parent.  
This adhesion has size at most $k$. 
Thus, every vertex has at most $d := d_1+k $ incoming edges in $\vec G_1$.
For future reference, partition the arcs of $\vec G_1$ into binary relations
$E_1,\ldots,E_d$ so that every vertex has at most one incoming arc in each
$E_i$.  Such a partition is obtained by assigning distinct indices to the
incoming arcs at every vertex.

We now define vertices $\alpha(x)$ and $\beta(x)$ for every non-root node
$x\in V(T)$. Let $x$ be a non-root node and let $S_x:=\adh(x)\cap\mrg(\parent(x))$.
By neatness, $S_x$ is non-empty.

Since $\adh(x)$ is a clique in~$G_1$ and the orientation of~$G_1$ is
acyclic, the orientation induced on $\adh(x)$ is a transitive tournament
and therefore has a unique sink. Every vertex of $\adh(x)\setminus S_x$
belongs to the margin of a strict ancestor of $\parent(x)$. Hence every
edge between $\adh(x)\setminus S_x$ and~$S_x$ is oriented towards~$S_x$.
Consequently, no vertex of $\adh(x)\setminus S_x$ can be the sink of the
tournament induced on~$\adh(x)$. We define $\alpha(x)$ to be this unique
sink, thus
$\alpha(x)\in\adh(x)\cap\mrg(\parent(x))$.

By neatness, every vertex of $\adh(x)$ has a neighbor in $\mrg(x)$.
We choose such a neighbor of $\alpha(x)$ and denote it by~$\beta(x)$.
For the root~$r$ of~$T$, we choose $\beta(r)$ arbitrarily in~$\mrg(r)$.

\medskip
As $G$ is a spanning subgraph of $G_1$, the tree $Y$ is also a spanning tree of $G_1$.

For every node $x\in V(T)$, let $Y_x$ be the rooted forest whose order is the
restriction of $\preceq_Y$ to~$\mrg(x)$. 
We call $Y_x$ the \emph{trace of
$Y$ on the margin $\mrg(x)$}.  Notice that a trace edge need not be an edge
of~$G_1$: it may shortcut a path of~$Y$ whose internal vertices
lie outside~$\mrg(x)$.  All trace forests can be computed in linear time by traversing~$Y$
depth-first and maintaining, for each margin, its most recent vertex on the
current root-to-vertex path.

\paragraph{Step 2}
Our next goal is to delete edges (while ensuring that they can later be
recovered by a transduction) so that every non-trivial block of the resulting
graph is contained in a graph with an explicit no-apex almost-embedding
witness, such that we can apply \zcref{cor:bloc}. This will yield the graph~$G_2$.

Let $G_0$ be the graph obtained from $G_1$ by first deleting all edges between different margins, except the edges $\alpha(x)\beta(x)$, and then, for every $x\in V(T)$, deleting all edges with one endpoint in $\mrg(x)$ and the other one in $A_x\cap\mrg(x)$.

Fix a node $x\in V(T)$ and contract each connected component of $G_0[\mrg(x)]$ to one vertex.
Let~$Q_x$ denote the resulting quotient of $G_1[\mrg(x)]$.
By neatness, $G_1[\mrg(x)]$ is connected, so $Q_x$ is connected.
Moreover, every apex in $A_x\cap\mrg(x)$ is an isolated vertex of $G_0[\mrg(x)]$ and therefore gives a singleton quotient vertex.

Next, choose a spanning tree of~$Q_x$.
For every edge of this tree, choose an arbitrary corresponding edge of $E(G_1[\mrg(x)])\setminus E(G_0[\mrg(x)])$.
Let~$F_x$ be the set of chosen edges.
Then $G_0[\mrg(x)]+F_x$ is connected and every edge of $F_x$ is a bridge\footnote{A \emph{bridge} is an edge whose removal increases the number of connected components of the graph.} in $G_0[\mrg(x)]+F_x$.
Indeed, after contracting the connected components of $G_0[\mrg(x)]$, the graph $G_0[\mrg(x)]+F_x$ becomes precisely the chosen spanning tree of~$Q_x$.
See~\zcref{fig:spanningtreeQx}.

\begin{figure}[ht]
    \centering
    \includegraphics[width=0.38\linewidth]{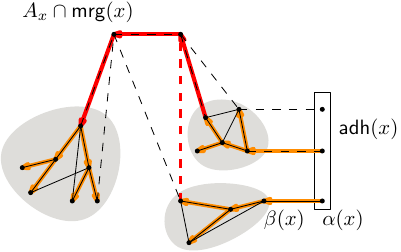}
    \caption{An illustration of chosen edges $F_x$ (in red).
    The picture illustrates the graph $G_1[\mrg(x)]$.
    The dashed edges are the ones that are deleted from $G_1$ to get $G_0$.
    The connected components of $G_0[\mrg(x)]$ that are not apices correspond to the three marked regions.
    The edges of $Y$ are highlighted and oriented away from the root of $Y$.
    The distinguished edges of $Y$ (red arcs), together with the distinguished dashed edge (also in red), form the set $F_x$ of chosen edges.}
    \label{fig:spanningtreeQx}
\end{figure}

We call the edges in $\bigcup_{x\in V(T)}F_x$ the \emph{apex connecting edges}.
Let~$G_2$ be the graph obtained from~$G_0$ by adding all apex connecting edges.

\medskip
We prove the following auxiliary lemma.

\begin{lemma}\label{lem:preserve-blocks}
Let $G$ be a connected graph and let $X\subseteq V(G)$.
Let $G'$ be the graph obtained from~$G$ by deleting all edges with at least one endpoint in~$X$.
Let $F\subseteq E(G)\setminus E(G')$ be inclusion-wise minimal such that $G'' := G'+F$ is connected.
Then the maximal $2$-connected subgraphs of $G''$ with at least three vertices are exactly the maximal $2$-connected subgraphs of $G'$ with at least three vertices.
\end{lemma}

\begin{proof}
By the inclusion-wise minimality of~$F$, every edge of~$F$ is a bridge of~$G''$.
Every edge of a $2$-connected graph with at least three vertices lies on a cycle, so no such subgraph of~$G''$ contains an edge of~$F$.
Thus, $G'$ and~$G''$ have exactly the same $2$-connected subgraphs with at least three vertices, and hence the same maximal ones.
\end{proof}

We next show that every $2$-connected subgraph of $G_2$ with at least three vertices is contained in a single margin. 
Each graph $G_2[\mrg(x)]$ is connected by the construction of~$F_x$. The only edges between distinct margins are the edges $\alpha(x)\beta(x)$ for non-root nodes~$x$, each joining $\mrg(x)$ to $\mrg(\parent(x))$. 
Contracting all margins therefore gives exactly~$T$. 
It follows that every cross-margin edge is a bridge of~$G_2$, so no $2$-connected subgraph with at least three vertices can use such an edge.

Now fix $x\in V(T)$.
We apply \zcref{lem:preserve-blocks} to the connected graph $G_1[\mrg(x)]$, with $X=A_x\cap\mrg(x)$, $G'=G_0[\mrg(x)]$, and $F=F_x$.
The graph $G_0[\mrg(x)]$ is obtained from $G_1[\mrg(x)]$ by deleting precisely the edges with at least one endpoint in $A_x\cap\mrg(x)$.
Moreover, $F_x$ is inclusion-wise minimal subject to $G_0[\mrg(x)]+F_x$ being connected, since every edge of $F_x$ is a bridge of this graph.
As $G_2[\mrg(x)]=G_0[\mrg(x)]+F_x$, \zcref{lem:preserve-blocks} yields that the maximal $2$-connected subgraphs of $G_2[\mrg(x)]$ with at least three vertices are exactly the maximal $2$-connected subgraphs of $G_0[\mrg(x)]$ with at least three vertices.

Every vertex of $A_x\cap\mrg(x)$ is isolated in $G_0[\mrg(x)]$.
Therefore, every maximal $2$-connected subgraph of $G_2[\mrg(x)]$ with at least
three vertices is contained in $G_1[\mrg(x)\setminus A_x]\subseteq P_x$,
for which Step~1 supplied
an explicit no-apex $k$-almost-embedding. Since every such block of $G_2$ is
contained in a single margin, \zcref{cor:bloc} yields the following corollary.

\begin{corollary}\label{crl:G_2-excludes-apex}
The graph $G_2$ excludes some apex graph $\Gamma_{h(G)}$ as a minor.
\end{corollary}

\begin{proof}
As shown above, every maximal $2$-connected subgraph of $G_2$ with at least
three vertices is contained in a graph with a no-apex $k$-almost-embedding.
The remaining blocks consist of single vertices or single edges.
Thus, \zcref{cor:bloc} applies to the class of all graphs $G_2$ obtained for
a fixed value of~$h(G)$, and the assertion follows.
\end{proof}

We now define the rooted spanning tree $Y_2$.  For every $x\in V(T)$, put
\[
   Y_x^0:=Y_x\cap G_0[\mrg(x)].
\]
If $C$ is a connected component of $G_0[\mrg(x)]$, then $Y_x^0[C]$ is a
forest.  Extend it to a spanning tree of~$C$, and let $F_{Y,x}$ be the set of
edges added in all these extensions.  We call them the \emph{spanning-tree
connecting edges}.  The tree $Y_2$ consists of all the resulting component
trees, all edges in $F_x$ for every $x\in V(T)$, and all edges
$\alpha(x)\beta(x)$ for non-root nodes $x\in V(T)$.  We root $Y_2$ at
$\beta(x_0)$, where $x_0$ is the root of~$T$.

\begin{lemma}
    $Y_2$ is a rooted spanning tree of $G_2$.
\end{lemma}
\begin{proof}
Fix $x\in V(T)$.  Inside each connected component of
$G_0[\mrg(x)]$, we chose a spanning tree.  By the construction of~$F_x$, the
graph obtained from $F_x$ after contracting these components is a tree.
Consequently, $Y_2[\mrg(x)]$ is connected and acyclic, and hence is a spanning
tree of the margin.

Finally, contract every tree $Y_2[\mrg(x)]$ to a single vertex.  The remaining
edges are precisely the edges $\alpha(x)\beta(x)$ connecting every non-root
margin to its parent margin, so the resulting quotient is~$T$.  It follows
that $Y_2$ is connected and that no cycle can use edges between distinct
margins.  Since there is no cycle inside a margin either, $Y_2$ is a rooted
spanning tree of~$G_2$.
\end{proof}

\paragraph{Step 3}

Every block of $G_2$ with at least three vertices is contained in a
graph~$P_x$ supplied in Step~1 with an explicit no-apex $k$-almost-embedding. The other
blocks have trivial such witnesses. Applying
\zcref{thm:milolidakis} to these witnesses computes a connected
subcubic graph~$H$ and a surjective minor-model map
$f\colon V(H)\to V(G_2)$ such that
$h(H)\leq\xi_c(h(G))$ for a computable function~$\xi_c$.
This computation takes time $f_1(h(G))\cdot |G|^{c_1}$ for a computable
function~$f_1$ and an absolute constant~$c_1$.

We now delete every edge $ab\in E(H)$ with $f(a)\neq f(b)$ and $f(a)f(b)\notin E(G_2)$.
The graph $H$ remains connected and subcubic, its Hadwiger number does not increase, and $f$ remains a surjective minor-model map.
Contracting the branch sets now gives exactly~$G_2$.
We now construct a spanning tree $Y_H$ of $H$ as follows: 
For every $v\in V(G_2)$, choose a spanning tree of $H[f^{-1}(v)]$. 
For every edge $uv\in E(Y_2)$, add one edge of~$H$ between $f^{-1}(u)$ and $f^{-1}(v)$. 
Their union $Y_H$ spans~$H$, because the branch sets partition $V(H)$. 
Contracting the chosen branch-set trees turns $Y_H$ into~$Y_2$. 
Hence, $Y_H$ is connected and acyclic, and is a spanning tree of~$H$.

\paragraph{Step 4}
We compute the expansion $(H,\preceq_{Y_H})$ of $H$ by the tree order $\preceq_{Y_H}$ defined by $Y_H$,
rooted at an arbitrary vertex of $f^{-1}(\beta(x_0))$, where $x_0$ is the root
of~$T$.  Contracting this branch set sends the root to~$\beta(x_0)$, which is
the root of~$Y_2$, hence the two root choices are compatible.

\paragraph{Step 5}

A linear ordering $<_H$ of $V(H)$ is computed from a DFS-traversal of $Y_H$.
This linear order can be obtained as a transduction of $(H,\preceq_{Y_H})$, as was stated without a proof in~\cite{buffiere2026characterizations}.
We give the proof for completeness. 

\begin{lemma}[\cite{buffiere2026characterizations}]
\label{lem:bitransducible-tree-linear}
    Let $\Cc$ be a class of connected graphs of maximum degree at most a
    fixed integer~$D$.  Then, for every $G\in\Cc$ and every rooted spanning
    tree~$Y$ of~$G$, every DFS preorder of~$Y$ can be transduced from
    $(G,\preceq_Y)$.
\end{lemma}
\begin{proof}
Let $G\in\Cc$, let $Y$ be a rooted spanning tree of~$G$, and fix a DFS preorder of~$Y$.
Every vertex has at most $D$ children.
Using unary predicates $C_1,\dots,C_D$, color the $i$th child of each vertex by~$C_i$, where the children are ordered according to the prescribed DFS preorder.
These colors can be computed in linear time from the corresponding ordered tree.
For distinct vertices $x,y$ that are comparable in~$\preceq_Y$, put the ancestor first.
Otherwise, let $z=x\curlywedge y$ and compare the two children of~$z$ leading to~$x$ and~$y$ by their colors.
The least common ancestor, these children, and their color comparison are first-order definable from the tree order and the unary predicates.
This defines exactly the prescribed DFS preorder.
\end{proof}

All computations in Steps~1-5 take time $f_2(h(G))\cdot |G|^{c_2}$ for a computable function~$f_2$ and an absolute constant~$c_2$.
Indeed, Steps~1 and~3 have bounds of this form, while the remaining computations are polynomial in the sizes of their inputs with an absolute exponent.
In particular, the size of~$H$ is bounded by the running time of Step~3, and its tree order and DFS preorder can be computed in quadratic time in~$|H|$.
We are now ready to conclude Phase I of the proof of~\zcref{thm-aux-structure}.

\begin{corollary}
    There is a computable function~$d$ such that the class
    \[\mathscr M_t = \{ (H,\preceq_{Y_H},<_H) \mid G\in \mathscr H_t, \text{$Y$~is a rooted spanning tree of $G$}\}\]
    has twin-width at most~$d(t)$.
\end{corollary}
\begin{proof}
Fix~$t$ and let $(H,\preceq_{Y_H},<_H)\in\mathscr M_t$.
By Step~3, $H$ is subcubic and $h(H)\leq\xi_c(t)$, where we may take~$\xi_c$ to be nondecreasing.

Observe that the graph reduct of every tree-ordered minor of $(H,\preceq_{Y_H})$ is an ordinary minor of~$H$.
Consequently, all these graph reducts exclude $K_{r,r}$ as a subgraph for $r:=\xi_c(t)+1$.
As $Y_H$ has maximum degree at most~$3$, \zcref{lem:effective-bounded-degree-covergraph-original} gives
$\operatorname{tww}(H,\preceq_{Y_H})\leq\beta(3,r)$.
By \zcref{lem:bitransducible-tree-linear}, the order~$<_H$ is obtainable by a fixed transduction.
Applying \zcref{lem:tww-preserved} therefore yields a computable bound~$d(t)$ on $\operatorname{tww}(H,\preceq_{Y_H},<_H)$.
\end{proof}

\subsection{Phase II}

The first phase of constructing the class $\Mm_t$ with bounded twin-width was algorithmic: starting from $(G,Y)$, we constructed a subcubic graph $H$ with a rooted spanning tree $Y_H$ and a DFS order. 
The purpose of the second phase is to show that we can reverse this construction by transductions. 
Starting from $(H,\preceq_{Y_H},<_H)$, we successively recover the tree-ordered and linearly ordered expansion $(G,\preceq_Y,<_G)$.
Since $\Mm_t$ has bounded twin-width and transductions preserve bounded twin-width, this will imply that the class of all produced structures $(G,\preceq_Y,<_G)$ has bounded twin-width. 
Additionally, as we obtain a linearly ordered class of bounded twin-width, we can do efficient model checking. 
We continue with the first step of the second phase.

\paragraph{Step 6}

The minor model of $G_2$ in~$H$ consists of pairwise-disjoint connected branch sets $f^{-1}(v)$, one for each $v\in V(G_2)$, which partition~$V(H)$.
Contracting these branch sets yields precisely~$G_2$, since all edges between branch sets corresponding to nonadjacent vertices of~$G_2$ have already been deleted.

For $v\in V(G_2)$, let $\iota(v)$ be the unique $\preceq_{Y_H}$-minimum of $f^{-1}(v)$.
This minimum exists because the branch set induces a subtree of~$Y_H$.
Contracting each branch set to its representative $\iota(v)$ turns $Y_H$ into~$Y_2$, by the construction of $Y_H$, and yields a tree-ordered graph whose $E$-reduct is precisely~$G_2$.

We describe the required transduction.
Mark the representatives $\iota(v)$ by a unary predicate~$Q$ and put
\[
 \chi(a,u):=Q(a)\wedge a\preceq_{Y_H}u
 \wedge\neg\exists b\,
 \bigl(Q(b)\wedge a\prec_{Y_H}b\preceq_{Y_H}u\bigr).
\]
Because the branch sets partition~$V(H)$ and each induces a subtree of~$Y_H$, the formula $\chi(\iota(v),u)$ holds exactly when $u\in f^{-1}(v)$.
Use $Q$ as the output domain and interpret adjacency by
\[
 \rho_E(a,b):=a\neq b\wedge
 \exists u\exists v\,
 \bigl(\chi(a,u)\wedge\chi(b,v)\wedge E(u,v)\bigr).
\]
The tree order is simply restricted to~$Q$.
We inherit the linear order from the representatives, that is, we define the linear order $<_{G}$ of $V(G)$ (recall that $V(G)=V(G_2)$) so that for every two distinct vertices $u,v\in V(G)$,
\[
   u <_{G} v
   \quad\Longleftrightarrow\quad
   \iota(u)<_H\iota(v).
\]
Thus, $(G_2,\preceq_{Y_2},<_{G})$ is obtained from $(H,\preceq_{Y_H},<_H)$ by a fixed transduction.

Consequently, the class of all expansions $(G_2,\preceq_{Y_2},<_{G})$ obtained in this
way from pairs $(G,Y)$ with $G\in\mathscr H_t$ and $Y$ a rooted spanning
tree of~$G$ can be transduced from~$\mathscr M_t$.

\paragraph{Step 7}

In the next step, we make the tree decomposition visible inside a transduced \emph{master structure}, which we obtain by a copying transduction.  
It contains one copy of the vertices of $G_2$ and a second kind of elements representing the bags of the decomposition. 
The bag element corresponding to a node $x$ is represented by a copy of the distinguished vertex $\beta(x)$. 
After rewiring, the vertices of the margin $\mrg(x)$ become exactly the descendants of the bag element representing $x$. 
This is similar to the augmented trees introduced in~\cite{pillipczuk2021connectivity}.

Formally, the \emph{master structure} $\mathbf M=\dec(G_2,\preceq_{Y_2},<_{G},\mathcal T)$ defined by $(G_2,\preceq_{Y_2},<_{G})$ and the tree decomposition~$\mathcal T$ has signature
$
  \sigma_{\rm mas}=\{V,B,R,P_\beta,E,\preceq,\unlhd,<,\mathsf{rep}\}$.
Here $V,B,R,P_\beta$ are unary predicates, and $V,B,R$ partition the domain.
The set $R(\mathbf M)$ contains a single element denoted~$\omega$.
The relation $\preceq$ is a tree order on the whole domain with root~$\omega$, $<$ is a linear order on $V(\mathbf M)$, $\unlhd$ is a tree order on $B(\mathbf M)$, and $\mathsf{rep}$ is a binary representative relation.
The structure is defined by bijections $\zeta_V:V(G_2)\rightarrow V(\mathbf M)$ and $\zeta_B:V(T)\rightarrow B(\mathbf M)$ satisfying the following conditions.
    \begin{itemize}
        \item $\zeta_V$ defines an isomorphism of the ${\{E,<\}}$-reduct of 
    $(G_2,\preceq_{Y_2},<_{G})$ and the ${\{E,<\}}$-reduct of  $\mathbf M[V(\mathbf M)]$;
    \item the set of children of $\omega$ (with respect to $\preceq$) is $B(\mathbf M)$;
    \item for every node $x\in V(T)$, the restriction of $\zeta_V$ to $\mrg(x)$ defines an isomorphism of 
    $(\mrg(x),\preceq_{Y_2})$ and 
    $(\{v\colon v\succ \zeta_B(x)\},\preceq)$;
    \item the mapping $\zeta_B$ induces an isomorphism of $(T,\preceq_T)$ and
    $(B(\mathbf M),\unlhd)$.
    \item $P_\beta$ marks exactly the vertices $\zeta_V(\beta(x))$, and
    $\mathsf{rep}(\zeta_B(x),\zeta_V(\beta(x)))$ holds for every
    $x\in V(T)$ and for no other pair.
    \end{itemize}

From this definition, it follows that, for all $x\in V(T)$ and
$u,v\in V(G_2)$,
\begin{align*}
v\in\mrg(x)
&\quad\Longleftrightarrow\quad
\mathbf M\models \zeta_B(x)\prec \zeta_V(v),\\
b(u)=b(v)
&\quad\Longleftrightarrow\quad
\mathbf M\models
\zeta_V(u)\curlywedge\zeta_V(v)\neq\omega.
\end{align*}

\begin{lemma}
    There is a transduction $\mathsf T$ such that 
    $\dec(G_2,\preceq_{Y_2},<_{G},\mathcal T)\in\mathsf T((G_2,\preceq_{Y_2},<_{G}))$.
\end{lemma}

\begin{proof}
    We first apply a first transduction consisting of the composition of a ${\sf Copy}_3$ operation (with equivalence relation between the clones denoted by $\simeq$) and the transduction defined as follows:
\begin{itemize}
    \item The universe of the structure will consist of three parts: a part marked $V$ corresponding to the domain of $G_2$ (this will be the first copy), a part marked $B$ corresponding to the image of $\beta$ (in the second copy and in bijection with the domain of $T$), and a global root marked $R$ (third copy of the root of the original tree-order). Hence, the domain is defined by the formula 
    \[
    \nu(x):=V(x)\vee B(x)\vee R(x).
    \]
    \item The relation $E$ is restricted to $V$:
    \[
    \rho_E(x,y):=V(x)\wedge V(y)\wedge E(x,y).
    \]
    \item The linear order $<$ is restricted to $V$:
    \[
    \rho_<(x,y):=V(x)\wedge V(y)\wedge (x<y).
    \]   
    \item The representative marker and relation are interpreted by
    \[
      \rho_{P_\beta}(x):=V(x)\wedge
         \exists y\,(B(y)\wedge x\simeq y)
    \]
    and
    \[
      \rho_{\mathsf{rep}}(x,y):=
         B(x)\wedge V(y)\wedge x\simeq y.
    \]
    \item A tree order on bags is defined as the restriction of the tree order $\preceq$ to the image of $\beta$ (which is essentially the set marked $B$):
    \[
    \rho_{\unlhd}(x,y):=B(x)\wedge B(y)\wedge (x\preceq y).
    \]
    \item Finally, a global tree order is defined as follows: the vertex marked $R$ is the root (hence smaller than all other vertices), and the copy marked $B$ of a vertex is introduced as a predecessor of its copy marked $V$ in the tree order inherited on $V$:
    \[
    \rho_\preceq(x,y)\is R(x)\vee (x=y)
      \vee \bigl(B(x)\wedge V(y)\wedge x\simeq y\bigr)
      \vee \exists x'\exists y'\,
        \bigl(x'\simeq x\wedge y'\simeq y\wedge x'\prec y'\bigr).
    \]
\end{itemize}
    The mapping $\zeta_V$ maps $v$ to its first copy, $\zeta_B$ maps $x$ to the second copy of~$\beta(x)$ and $\omega$ is the only vertex marked~$R$.

    Then, we apply the transduction $\rew$ from \zcref{lem:rewire}.
    Marking the root both $S$ and $A$ and marking all the vertices marked $B$ by $R_{\rm old}$, the rewiring detaches all the subtrees rooted at a vertex marked $B$ and moves them as  children of the root. This witnesses that our transduction produces the master structure.

    The whole process is displayed on \zcref{fig:ex}. 
    \end{proof}

\begin{figure}[H]
    \centering
    \includegraphics[width=\columnwidth]{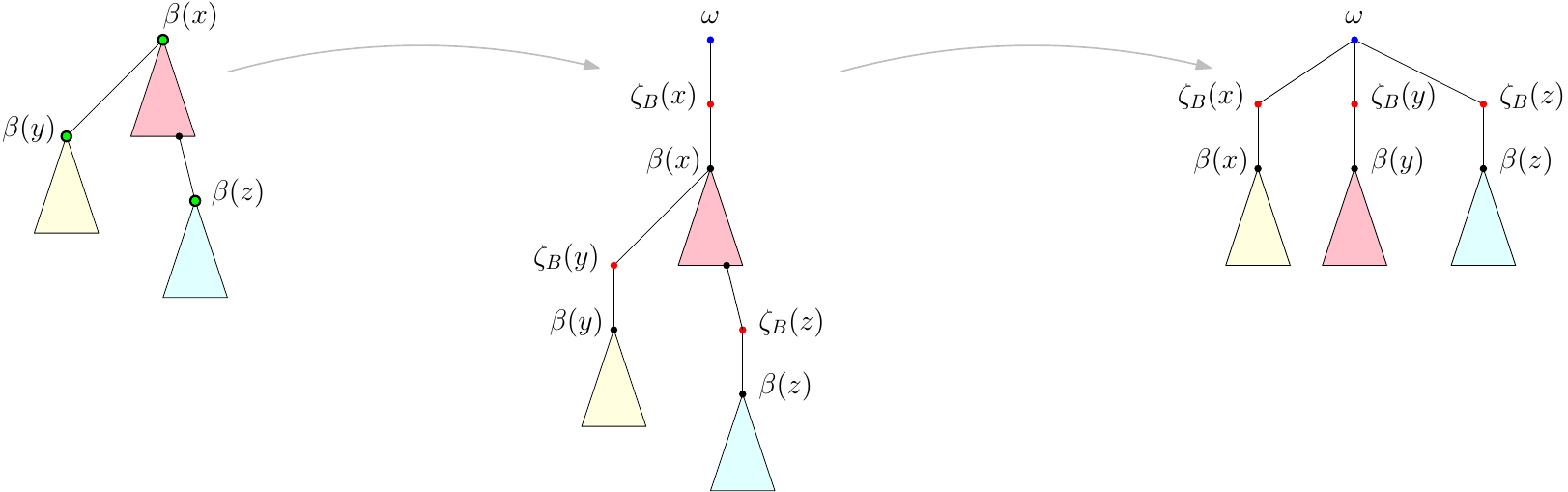}
    \caption{The master structure as a transduction.
    In the drawing, the graph $G_2$ is decomposed into three bags, $x$ (in pink), $y$ (in yellow), and $z$ (in blue), with respective minimum vertices $\beta(x),\beta(y)$, and $\beta(z)$. These minimum vertices are duplicated as $\zeta_B(x),\zeta_B(y)$, and $\zeta_B(z)$ and these new vertices are moved as direct children of a super-root $\omega$.
    }
    \label{fig:ex}
\end{figure}

\paragraph{Step 8}
At this point, the master structure still contains the simplified graph $G_2$
and the simplified tree order~$Y_2$.  We now undo the local simplifications
inside each margin.  We restore the edges of~$G_1$ incident with the
boundedly many apices and replace the auxiliary tree on each margin by the
part of the trace forest that is carried by actual edges of~$G_1$.

\begin{figure}[h!t]
\begin{center}
    \includegraphics[width=\textwidth]{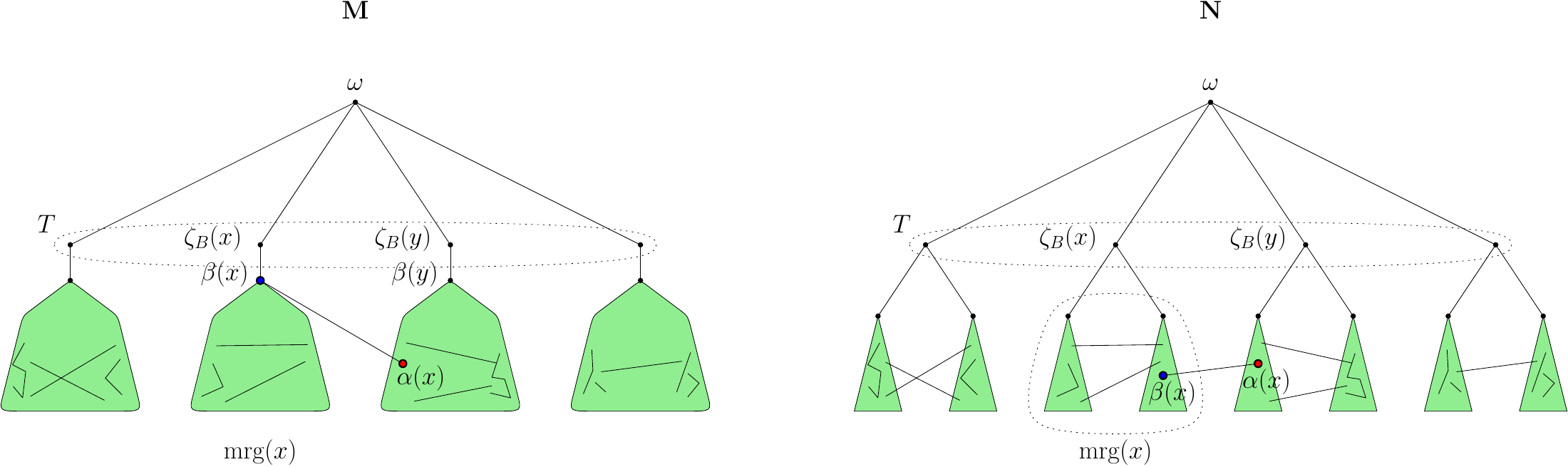}
\end{center}
\caption{In $\mathbf M$ (left), the first level under the root $\omega$ contains a copy
of the tree-decomposition order and, under each node $x$, the subtree rooted
at $\beta(x)$ contains the vertices of $\mrg(x)$.  In $\mathbf N$ (right),
the order on each margin is the local trace order induced by
$Y_x^{\mathrm{loc}}$, and the edge relation inside each margin is the one
induced by~$G_1$.
}
\label{fig:master}
\end{figure}

The structure $\mathbf N$ has the same signature and domain as~$\mathbf M$.
The predicates $V,B,R,P_\beta$, the relations $\unlhd,<,\mathsf{rep}$, and the
root~$\omega$ are unchanged. We next explain how $E$ and $\preceq$ change. For every $x\in V(T)$, put
\[
   Y_x^{\mathrm{loc}}:=Y_x\cap G_1[\mrg(x)]
\]
and let $\preceq_x^{\mathrm{loc}}$ be the forest-order induced by this
forest.  Inside $\mrg(x)$, the relation $\preceq$ of $\mathbf N$ is
$\preceq_x^{\mathrm{loc}}$, and every component root of this forest is a
child of the bag element $\zeta_B(x)$.  The order on the bag elements and
the order between distinct margin universes are retained from~$\mathbf M$.
For vertices $u,v$ marked~$V$, set
\[
\mathbf N\models E(u,v)
\quad\Longleftrightarrow\quad
\begin{cases}
   uv\in E(G_1),&\text{if }b(u)=b(v),\\
   \{u,v\}=\{\alpha(x),\beta(x)\}
      \text{ for some non-root }x\in V(T),&\text{if }b(u)\neq b(v).
\end{cases}
\]
Thus, the restriction of $\mathbf N$ to a margin carries the local trace order
and the edge relation induced by~$G_1$.  Between distinct margins, the only
graph edges retained at this stage are the distinguished attachment edges
$\alpha(x)\beta(x)$.

\begin{lemma}
There is a transduction $\mathsf T$ such that
$\mathbf N\in\mathsf T(\mathbf M)$.
\end{lemma}
\begin{proof}
For each of the at most $k$ apex positions, use a unary predicate to select the
corresponding apex in every margin and further unary predicates to select its
neighbors in $G_1$ within that margin.  Membership in a common margin is
first-order definable in~$\mathbf M$.  A fixed number of interpretation steps
therefore restores exactly the edge relation of $G_1$ inside every margin.

We next reconstruct the local trace forests.
In every margin, mark by $R_{\rm old}$ the child endpoint, with respect to the current tree order, of every edge belonging to $F_{Y,x}$ or~$F_x$.
Mark every bag element $\zeta_B(x)$ both $S$ and~$A$.
The rewiring transduction from \zcref{lem:rewire} deletes the corresponding parent--child edges and attaches every resulting component directly below~$\zeta_B(x)$.
The edges left inside the margin are exactly those of~$Y_x^0$.
The component roots now form an antichain, so \zcref{lem:reroot-arbitrary} applies simultaneously to all components.
We mark the minimum of each component in $\preceq_Y$ as its new root, obtaining the orientation inherited from~$Y_x$.

It remains to add the edges of
$Y_x^{\mathrm{loc}}\setminus Y_x^0$.
Every such edge is incident with a vertex of $A_x\cap\mrg(x)$, because
$G_0[\mrg(x)]$ was obtained from $G_1[\mrg(x)]$ precisely by deleting the
edges incident with these local apices.  Assign the local apices at most~$k$
fixed positions.  For one apex position at a time, first restore every trace
edge whose child is that apex: mark the apex as the component root to be
moved and mark its unique parent in $Y_x^{\mathrm{loc}}$ as the attachment
vertex.  In a second application, restore all trace edges whose parent is
that apex and whose child is not an apex: mark the apex as the attachment
vertex and the corresponding component roots as the roots to be moved.
Edges between two apices are handled in the first case, according to the
position of their child.

These operations use at most $2k$ applications of the fixed rewiring
transduction and are performed simultaneously in all margins.  The resulting
restriction of $\preceq$ to $\mrg(x)$ is exactly
$\preceq_x^{\mathrm{loc}}$.  Since the same colors are reused in every
margin, the whole construction is a fixed transduction depending only
on~$k$.
\end{proof}

The local trace order need not yet be the full restriction of
$\preceq_Y$. Indeed, a cover of the restriction of $\preceq_Y$ to a
margin need not be an edge of~$G_1$: the corresponding path in~$Y$
may leave the margin and later return to it. The following lemma
shows that every cover missing from the local trace forest has
an adhesion vertex on its corresponding path in~$Y$. This accounts
for all comparisons still missing within a margin, as we show
in~\eqref{eq:local-trace} below. In Step~9, we recover these
comparisons by recording the order relations between each margin
vertex and the vertices of its adhesion.
Here, the \emph{open $u$-$v$ path} is the unique $u$-$v$ path
in~$Y$ with its endpoints $u$ and $v$ removed.

\begin{lemma}\label{lem:local-trace}
Let $x\in V(T)$, and let $u\prec_Y v$ be a cover in the restriction of
$\preceq_Y$ to $\mrg(x)$. If the open $u$-$v$ path in~$Y$ avoids
$\adh(x)$, then $uv\in E(G_1)$.
\end{lemma}

\begin{proof}
If $uv\in E(Y)$, there is nothing to prove.
Otherwise, the open $u$-$v$ path is non-empty and avoids $\bag(x)=\mrg(x)\mathbin{\dot\cup}\adh(x)$, where $\adh(x)=\emptyset$ if $x$ is the root.
Indeed, it avoids $\mrg(x)$ because $u\prec_Yv$ is a cover in the restricted order, and it avoids $\adh(x)$ by assumption.
Let $C$ be the component of $G_1-\bag(x)$ containing this open path.
The nodes whose bags meet~$C$ form a connected subtree of $T-x$, and hence lie in one component~$D$ of~$T-x$.
If $q$ is the neighbor of~$x$ in~$D$, the connectedness of the bags containing each vertex implies $N_{G_1}(C)\cap\bag(x)\subseteq\bag(x)\cap\bag(q)$.
Both $u$ and~$v$ belong to this set.
The node~$q$ cannot be $\parent(x)$, since $u,v\in\mrg(x)$.
Thus, $q$ is a child of~$x$, and $u,v\in\adh(q)$.
As $\adh(q)$ is a clique in~$G_1$, we obtain $uv\in E(G_1)$.
\end{proof}

Consequently, for all $u,v\in\mrg(x)$,
\begin{equation}\tag{$\ast$}\label{eq:local-trace}
u\preceq_Y v
\quad\Longleftrightarrow\quad
u\preceq_x^{\mathrm{loc}}v
\ \lor\
\exists a\in\adh(x)\;
\bigl(u\preceq_Y a\preceq_Y v\bigr).
\end{equation}

For the forward implication, the case $u=v$ is immediate, so suppose
that $u\prec_Y v$. If the $u$-$v$ segment meets $\adh(x)$, we choose a
vertex $a\in\adh(x)$ on that segment. Otherwise, let
$u=z_0\prec_Y z_1\prec_Y\cdots\prec_Y z_\ell=v$ be the vertices of
$\mrg(x)$ on the segment, listed in their $Y$-order. For every
$i<\ell$, the pair $z_i\prec_Y z_{i+1}$ is a cover in the restriction
of $\preceq_Y$ to $\mrg(x)$, and its open path avoids $\adh(x)$. By
\zcref{lem:local-trace}, we have $z_iz_{i+1}\in E(G_1)$. Hence, by the
definition of $Y_x^{\mathrm{loc}}$, every such pair belongs to
$Y_x^{\mathrm{loc}}$. Taking the transitive closure gives
$u\preceq_x^{\mathrm{loc}}v$.

For the reverse implication, recall that
$\preceq_x^{\mathrm{loc}}$ is, by construction, contained in the
restriction of $\preceq_Y$ to $\mrg(x)$. Thus,
$u\preceq_x^{\mathrm{loc}}v$ implies $u\preceq_Y v$. If instead there
exists $a\in\adh(x)$ such that
$u\preceq_Y a\preceq_Y v$, then $u\preceq_Y v$ follows directly from
the transitivity of~$\preceq_Y$.

\paragraph{Step 9}

In the final step, we restore the edges and the tree order between different
margins.  We first need more than a definition of the set $\adh(x)$: a fixed
formula must be able to identify each of its elements separately.  This is
possible because every adhesion has size at most~$k$.  The first coordinate
in the index set below locates the margin containing an adhesion vertex, the
second identifies that vertex using the bounded-indegree orientation from
Step~1.

\begin{lemma}\label{lem:adhesion}
Let $I=[2k]\times\{0,\ldots,d\}$.  There is a monadic expansion
$\mathbf N^+$ of a fixed transduction of~$\mathbf N$, injections $\iota_x:\adh(x)\longrightarrow I$ for $x\in V(T)$,
and first-order formulas $\varphi^{\adh}_{c,i}(v,z)$, for $(c,i)\in I$, such
that
\[
\mathbf N^+\models
   \varphi^{\adh}_{c,i}(v,\zeta_B(x))
\quad\Longleftrightarrow\quad
v\in\adh(x)\text{ and }\iota_x(v)=(c,i).
\]
For the root of~$T$, whose adhesion is empty, all these formulas are false.
\end{lemma}

\begin{proof}
The $E$-reduct of~$\mathbf N$ is a subgraph of~$G_1$ together with the isolated bag elements and~$\omega$, and therefore excludes a fixed minor.
Apply \zcref{lem:structure-from-gaif} to this reduct and carry the other relations through by \zcref{fact:expand}.
Using the bounded-indegree orientation constructed in Step~1, we obtain directed relations $E_1,\ldots,E_d$ which agree with the corresponding oriented edges of~$G_1$ currently present in~$\mathbf N$.
Every vertex has at most one incoming $E_i$-edge for each~$i$.

For $v\in V(G_1)$, let
$X_v:=\{x\in V(T):v\in\bag(x)\}$.
This is a non-empty connected subtree of~$T$, and its minimum is the placement
node $b(v)$.  Consequently,
\[
 v\in\mrg(x)\Longleftrightarrow x=b(v),
 \qquad
 v\in\adh(x)\Longleftrightarrow x\in X_v\setminus\{b(v)\}.
\]

Let $\mathscr X$ be the indexed family of all connected components of
$T[X_v\setminus\{b(v)\}]$, over all vertices~$v$.  Every node $x$ belongs to
exactly one member of $\mathscr X$ for each vertex of~$\adh(x)$, and therefore
to at most~$k$ members of~$\mathscr X$.

We color the subtrees in $\mathscr X$ with colors in $[2k]$ so that subtrees
of the same color neither intersect nor contain adjacent vertices of~$T$.
This is done as follows. We process them in non-decreasing order of the depth of their minimum.  When a
subtree~$X$ with minimum~$x$ is processed, every previously processed subtree
intersecting or containing a vertex adjacent to~$X$ contains $x$ or
$\parent(x)$.  At most $k-1$ previously processed subtrees contain
$x$, and at most $k$ contain $\parent(x)$.  Hence, one of the $2k$ colors is
available.

For $x\in V(T)$, let $\Lambda(x)$ be the set of colors of the members of
$\mathscr X$ containing~$x$, encoded by unary predicates on
$\zeta_B(x)$.  If $c\in\Lambda(x)$, there is a unique color-$c$ subtree
containing~$x$. Let $\eta(x,c)$ denote its minimum.  The relation
$y=\eta(x,c)$ is first-order definable from $\unlhd$ and the unary color
predicates: $y$ is the first vertex of the maximal color-$c$ interval on the
ancestor chain ending at~$x$.  Moreover,
\[
 \{\eta(x,c):c\in\Lambda(x)\}
 =\{y\preceq_T x:\ y\text{ is non-root and }
       \adh(x)\cap\mrg(\parent(y))\neq\emptyset\}.
\]

Fix $c\in\Lambda(x)$ and put $y:=\eta(x,c)$.  The color~$c$ identifies a
unique vertex $v\in\adh(x)\cap\mrg(\parent(y))$.  The relation
$\mathsf{rep}$ explicitly identifies $\beta(y)$ from~$\zeta_B(y)$, even
though $\beta(y)$ need not be the unique child of~$\zeta_B(y)$ after Step~8.
Using a unary color on $\zeta_B(y)$, record the index $j\in[d]$ for which
$\alpha(y)$ is the unique $E_j$-in-neighbor of~$\beta(y)$.  Thus,
$\alpha(y)$ is first-order definable from~$\zeta_B(y)$.

Since $v,\alpha(y)\in\adh(y)$ and $\adh(y)$ is a clique of~$G_1$, either
$v=\alpha(y)$ or the orientation contains the arc $v\alpha(y)$. The latter
holds because $\alpha(y)$ is the sink of the transitive tournament on
$\adh(y)$.  Set $i=0$ in the first case, and otherwise let $i\in[d]$ be the
unique index for which $E_i(v,\alpha(y))$ holds.  A unary predicate on
$\zeta_B(x)$ records this index for each~$c$.  Define
$\iota_x(v):=(c,i)$.  Distinct vertices of $\adh(x)$ correspond to distinct
members of the indexed family $\mathscr X$ through~$x$, and these receive
distinct colors, hence $\iota_x$ is injective.  The construction is expressed
by the formulas $\varphi^{\adh}_{c,i}$.
\end{proof}

We can now restore all edges of $G_1$ between different margins.  Let
$v\in\mrg(x)$.  Every neighbor of~$v$ in an ancestor margin belongs to
$\adh(x)$.  For every $(c,i)\in I$, use a unary predicate on~$v$ to record
whether $v$ is adjacent in~$G_1$ to the unique vertex of $\adh(x)$ with index
$(c,i)$.  The placement node~$x$ is definable from~$v$ through its bag
representative $\zeta_B(x)$, so \zcref{lem:adhesion} interprets exactly all
missing cross-margin edges.  Together with the edges already present inside
the margins, this recovers the full edge relation of~$G_1$.

It remains to restore the global tree order.
Order $I=[2k]\times\{0,\ldots,d\}$ lexicographically.
For each $x\in V(T)$, let $\overline{\adh}(x)$ be the $I$-indexed tuple whose coordinate~$\lambda$ is the unique vertex of $\adh(x)$ with index~$\lambda$, or a formal absence marker if no such vertex exists.
For $x\preceq_Ty$, let $\theta(x,y)$ record which coordinates of these two tuples are present, all equalities between present coordinates, and all comparisons in~$\preceq_Y$ between present coordinates.
We call this finite record the \emph{order type} of $\overline{\adh}(y)$ over $\overline{\adh}(x)$. 

\begin{claim}
If $x\preceq_T y\preceq_T z$, then $\theta(x,z)$ is determined by
$\theta(x,y)$ and $\theta(y,z)$.
\end{claim}
\begin{claimproof}
Let $a\in\adh(x)$ and $c\in\adh(z)$.  Put $W_y:=\{w\in V(G_1):y\preceq_Tb(w)\}$.
The adhesion $\adh(y)$ separates $W_y$ from all vertices outside $W_y\cup\adh(y)$.
We have $a\notin W_y$.
Moreover, either $c\in W_y$ or $c\in\adh(y)$, since the bags containing~$c$ form a connected subtree and contain~$z$.
It follows that the $a$-$c$ path in~$Y$ meets $\adh(y)$, possibly at an endpoint.
Consequently,
\[
 a\preceq_Yc\quad\Longleftrightarrow\quad
 \exists b\in\adh(y)\;a\preceq_Yb\preceq_Yc.
\]
The same argument with the order reversed determines whether $c\preceq_Ya$.
Both comparisons are determined by $\theta(x,y)$ and $\theta(y,z)$, and equality holds exactly when both do.
The internal comparisons in the two end tuples are already recorded in the corresponding intermediate types.
\end{claimproof}

After padding the adhesion tuples by dummy coordinates, the set of order types is
finite, and the claim gives a finite transition table for propagating order types
along ancestor paths. 
The types~$\theta(x,x)$ and those on parent-child
pairs can be encoded by finitely many unary predicates on~$T$.  The resulting
propagation is $\MSO$-definable by the standard equivalence between finite
automata and monadic second-order logic on finite
paths~\cite{ThatcherWright1968,Doner1970,Thomas1997}.  By
Colcombet's theorem~\cite[Theorem~2]{Colcombet2007}, all relations~$\theta(x,y)$ are therefore first-order definable after a further monadic
expansion of~$(T,\preceq_T)$, and hence of~$\mathbf N$ through the bag
representatives.  These translations are effective, so the required formulas
and unary predicates depend computably only on~$k$.

We also record the order type of every margin vertex over its indexed
adhesion.  For $\lambda\in I$, write~$a_\lambda(x)$ for the vertex of
$\adh(x)$ with index~$\lambda$, when it exists.  Add unary predicates
$P^\uparrow_\lambda$ and $P^\downarrow_\lambda$ with the meanings
\[
\begin{aligned}
P^\uparrow_\lambda(w)
&\quad\Longleftrightarrow\quad
w\in\mrg(x),\ a_\lambda(x)\text{ exists, and }
w\preceq_Ya_\lambda(x),\\
P^\downarrow_\lambda(w)
&\quad\Longleftrightarrow\quad
w\in\mrg(x),\ a_\lambda(x)\text{ exists, and }
a_\lambda(x)\preceq_Yw.
\end{aligned}
\]
Only $2|I|$ unary predicates are required, so this is a fixed monadic
expansion.

We now define the desired order in three stages.
First, let
\[
\begin{split}
 R_1(u,v):\Longleftrightarrow{}& b(u)=b(v)\ \wedge\\
 &\left(
 u\preceq v
 \ \lor\
 \bigvee_{\lambda\in I}
 \bigl(P^\uparrow_\lambda(u)\wedge
       P^\downarrow_\lambda(v)\bigr)
 \right),
\end{split}
\]
where $\preceq$ in the first disjunct is the local trace order carried by $\mathbf N$.
By \eqref{eq:local-trace},
\[
 R_1(u,v)
 \quad\Longleftrightarrow\quad
 b(u)=b(v)\text{ and }u\preceq_Yv.
\]

Next suppose that $b(v)\prec_T b(u)$, and let $z$ be the unique child of $b(v)$ with $z\preceq_T b(u)$.
To recover the order type of~$v$ over $\adh(z)$, split this adhesion as
\[
 \adh(z)=
 \bigl(\adh(z)\cap\adh(b(v))\bigr)
 \mathbin{\dot\cup}
 \bigl(\adh(z)\setminus\adh(b(v))\bigr).
\]
Comparisons with the first part are given by the unary code for the type of $v$ over $\adh(b(v))$.
Since $b(v)=\parent(z)$, the second part is contained in $\mrg(b(v))$, and its comparisons with~$v$ are supplied by~$R_1$.
Put $y=b(u)$.
The adhesion separators give
\[
 u\preceq_Yv
 \quad\Longleftrightarrow\quad
 \exists a\in\adh(y)\;\exists b\in\adh(z)\;
 u\preceq_Ya\preceq_Yb\preceq_Yv.
\]
For the forward implication, the $u$-$v$ path in~$Y$ meets $\adh(y)$ and then $\adh(z)$, possibly at the same vertex.
The reverse implication follows from transitivity.
The comparisons of~$u$ with~$a$ are given by its unary code, those of~$a$ with~$b$ by the composed type $\theta(z,y)$, and those of~$b$ with~$v$ by the type just recovered.
Thus, the displayed condition is first-order definable.
The case $b(u)\prec_T b(v)$ is analogous, exchanging $u$ and~$v$ and reversing the comparisons.
Together with $R_1$ for equal placement nodes, this defines a relation $R_2$ such that
\[
 R_2(u,v)\Longleftrightarrow
 \bigl(b(u)\preceq_T b(v)\vee b(v)\preceq_T b(u)\bigr)
 \wedge u\preceq_Yv.
\]

The last case uses the following separator observation.
\begin{lemma}\label{lem:path-common-placement}
Let $u,v\in V(G)$ have incomparable placement nodes, and let $z=b(u)\curlywedge_T b(v)$.
The $u$-$v$ path in~$Y$ meets $\bag(z)$.
Every $w\in\bag(z)$ satisfies
\[
   b(w)\preceq_T z\preceq_T b(u),b(v).
\]
Consequently, if $u\preceq_Yv$, then there is a vertex $w$ with $b(w)\preceq_T b(u),b(v)$ and $u\preceq_Y w\preceq_Y v$.
\end{lemma}
\begin{proof}
The two placement nodes lie in distinct child subtrees of~$z$.
The tree-decomposition axioms imply that $\bag(z)$ separates the vertices placed in these two parts of~$T$.
Since $Y\subseteq G\subseteq G_1$, the $u$-$v$ path in~$Y$ must therefore meet~$\bag(z)$.
If $w\in\bag(z)$, the connected set of bags containing~$w$ includes~$z$, so its minimum $b(w)$ is an ancestor of~$z$.
Finally, when $u\preceq_Yv$, every vertex $w$ on the $u$-$v$ segment of~$Y$ satisfies $u\preceq_Y w\preceq_Y v$.
\end{proof}

It follows that the full tree order is defined by
\[
 R_3(u,v):\Longleftrightarrow
 R_2(u,v)\vee
 \exists w\,
 \bigl(b(w)\preceq_T b(u)\wedge b(w)\preceq_T b(v)
       \wedge R_2(u,w)\wedge R_2(w,v)\bigr).
\]
Indeed, the forward implication in the incomparable-placement case follows from \zcref{lem:path-common-placement}, while the reverse implication follows from transitivity of~$\preceq_Y$.
Thus, $R_3(u,v)$ holds exactly when $u\preceq_Yv$.

We have now recovered $(G_1,\preceq_Y,<_G)$. Fix~$t$ and consider inputs~$G$
with $h(G)\leq t$. By the bound proved in Phase~I,
$h(G_1)\leq q_1(t)$
for the computable function~$q_1$ from Phase~I. Hence, the effective transduction from \zcref{lem:structure-from-gaif}, extended as in \zcref{fact:expand}, recovers
the edge relation of~$G$ from~$G_1$ while retaining $\preceq_Y$ and~$<_G$.

All transductions used in Steps~6--9 are given by explicit formulas, and the
numbers of copies, colors, and formulas used in them are bounded computably in~$t$. Their composition with the preceding expansion
transduction therefore forms an effectively computable family
$(\mathsf R_t)_{t\in\mathbb N}$. Since
$\operatorname{tww}(\mathscr M_t)\leq d(t)$ for a computable bound $d(t)$, we overall get a computable bound $\xi(t)$ from \zcref{lem:tww-preserved}. 
This completes the proof.
\hfill\qed

\bigskip

We also obtain \zcref{thm:minor-transduction}. For every fixed bound~$t$,
given a connected graph~$G$ with $h(G)\leq t$ and a rooted spanning
tree~$Y$ of~$G$, the construction computes a connected subcubic graph~$H$
with a rooted spanning tree~$Y_H$ and a DFS preorder~$<_H$.
The Hadwiger number of~$H$ is bounded by a computable function of~$t$.
Step~5 supplies the fixed transduction~$\mathsf L$, while
Steps~6--9 supply the transduction~$\mathsf R_t$, such that
\[
    (H,\preceq_{Y_H},{<_H})\in\mathsf L(H,\preceq_{Y_H})
    \qquad\text{and}\qquad
    (G,\preceq_Y,{<_G})\in\mathsf R_t(H,\preceq_{Y_H},{<_H})
\]
for the linear order~$<_G$ produced by the construction.
The family $(\mathsf R_t)_{t\in\mathbb N}$ is effectively computable,
and composing $\mathsf R_t$ with~$\mathsf L$ and forgetting both
orders recovers~$G$.

For an arbitrary graph~$G$ with $h(G)\leq t$ and a rooted spanning forest~$F$ of~$G$, apply the connected construction with bound~$t+1$ to the connected extension $(\widehat G,Y)$ from \zcref{rem:forest-tree-reduction}.
Deleting the marked universal vertex from the interpreted output recovers $(G,\preceq_F,<_G)$, where $<_G$ is the restriction of the computed linear order.
The auxiliary graph~$H$ is unchanged and remains connected and subcubic.

\bigskip
We now consider the special case of monotone classes of graphs.
\begin{theorem}[store=thm-main-monotone]
Let $\Cc$ be a monotone class of graphs. Then the following are equivalent:
\begin{enumerate}[nosep]
    \item $\Cc$ excludes some graph as a minor;
    \item $\Cc^{\Upsilon}$ has bounded twin-width;
    \item $\Cc^F$ has bounded twin-width;
    \item for every $G\in\Cc$, there is a choice of a full rooted spanning forest $F_G$ of $G$
     such that the class
    $\{(G,\preceq_{F_G})\mid G\in\Cc\}$ is monadically dependent.
\end{enumerate}
\end{theorem}
\begin{proof}
The first condition implies the third by \zcref{thm-main-structure}, since a minor-excluded class is weakly sparse.
The third implies the second because $\Cc^{\Upsilon}\subseteq\Cc^F$.
The second implies the fourth by choosing, for every graph, any full rooted spanning forest and using that bounded twin-width implies monadic dependence.

It remains to prove that the fourth condition implies the first.
Suppose for a contradiction that~$\Cc$ is monotone and excludes no fixed minor.
Choose a sequence $(H_n)_{n\in\mathbb N}$ of connected cubic graphs of unbounded twin-width.
Such a sequence exists because the class of cubic graphs has unbounded twin-width~\cite{bonnet2022twwII}.

For every $n$, some graph in~$\Cc$ contains $H_n$ as a minor.  Taking an
edge-minimal minor model and using monotonicity, we conclude that~$\Cc$
contains a subdivision $S_n$ of~$H_n$.  Let $F_n$ be the full rooted spanning
forest supplied by the fourth condition.  Since $S_n$ is connected, $F_n$ is
a rooted spanning tree.

Every subdivision vertex has degree two and is incident with an edge of the
spanning tree~$F_n$.  On a path of $S_n$ that subdivides one edge of $H_n$,
the tree $F_n$ omits at most one edge: omitting two would disconnect the
vertices between them.  If it omits one, contract every cover edge of $F_n$
on that path; if it omits none, leave one cover edge uncontracted and contract
the others.  These choices can be made simultaneously so that every
contraction class contains exactly one original vertex of~$H_n$.  Contracting
the selected cover edges therefore collapses the internal vertices of every
subdivided path while retaining one edge between its two original endpoints.
It produces a tree-ordered copy~$(H_n,\preceq_{T_n})$, where $T_n$ is a
spanning tree of~$H_n$ and hence has maximum degree at most three.  The
selected cover edges can be guessed by unary colors and the contractions are
realized by the fixed tree-order contraction transduction from~\cite{buffiere2026characterizations}.  Thus, the resulting
$E$-reduct is exactly~$H_n$ (with parallel edges, if any, suppressed by the
convention that graphs are simple).

Thus, the class $\{(H_n,\preceq_{T_n}):n\in\mathbb N\}$ is a transduction of
$\{(S_n,\preceq_{F_n}):n\in\mathbb N\}$ and would be monadically dependent.
It is a weakly sparse class of tree-ordered graphs whose cover graphs have
maximum degree at most three.  By
\zcref{thm:bounded-degree-covergraph}, it would consequently have bounded
twin-width.  This is impossible because its $E$-reducts are the graphs
$H_n$, which have unbounded twin-width.  Hence, $\Cc$ must exclude a fixed
minor.
\end{proof}

Note that \zcref{thm-main-monotone} does not extend to hereditary
weakly sparse classes. This can be seen by considering the weakly sparse example
used in~\cite{bonnet2022twin} in order to prove that \zcref{thm:minor-tww} does not
extend to hereditary classes.

\section{Sub-connectivity logic}\label{sec:subconn-logic}

We now introduce \emph{sub-connectivity logic}~$\FOscon$ on binary relational
structures. Let $\sigma$ be a finite binary relational signature containing a
distinguished binary relation symbol~$Z$. The logic is evaluated on
$\sigma$-structures $\mathbf M$ in which $Z(\mathbf M)$ is symmetric and
irreflexive. We identify $\MZ$ with the simple undirected
graph with vertex set~$M$ and edge relation~$Z(\mathbf M)$. In particular,
$\MZ$ is a subgraph of $\Gaif(\mathbf M)$.

The logic $\FOscon[\sigma]$ is obtained from $\FO[\sigma]$ by allowing, for
every tuple $\bar z$, the additional atomic formula
\[
   \scon(x,y;\bar z).
\]

On a $\sigma$-structure $\mathbf M$, the atom
$\scon(x,y;z_1,\ldots,z_k)$ holds if $x$ and $y$ are vertices of the same
connected component of $\mathbf M\!\upharpoonright_Z-\bar z$. In particular,
the atom is false if $x$ or $y$ is among the deleted vertices. Paths are
allowed to have length zero, so $\scon(x,x;\bar z)$ holds precisely when $x$
is not deleted. If $\bar z$ is empty, we simply write $\scon(x,y)$.

\subsection{Equivalence to FO with forest orders}\label{sec:FOFO}

To characterize the expressive power of $\FOscon$, we note the following key lemma about \FOconn. 

\begin{lemma}
\label{lem:con=FO}
    For every nonnegative integer $k$,
    there exists a formula $\psi_k(x,y,\bar z)$ in $\FO[\{E,\preceq\}]$ with $|\bar z|=k$ and the following property:
    for every rooted depth-first search forest $F$ of a graph $G$ and every $u,v\in V(G)$ and $\bar w\in V(G)^k$, we have
    \[
    G\models \conn(u,v,\bar w)\quad\iff\quad (G,\preceq_F)\models \psi_k(u,v,\bar w).
    \]

     On the other hand,
     there exists a formula $\varrho(x,y)$ of $\FOconn$ (for graphs with a unary predicate~$R$),  such that if $Y$ is a forest, $R$ marks one root in each connected component of $Y$, and $\preceq$ is the associated forest order, then
     
\[
(Y,\preceq)\models u\preceq v\quad\iff (Y,R)\models \varrho(u,v).
\]
\end{lemma}

\begin{proof} 
Let $G$ be a graph, let $k\geq 0$, and let $F$ be a rooted depth-first spanning forest of~$G$.
We write $\preceq$ for the forest
order $\preceq_F$.
We use the standard DFS property for the graph $G$:  every edge of $G$ has comparable endpoints (with respect to  $\preceq$). 

We first define, in first-order logic over $\{E,\preceq\}$, when two vertices
lie in the same component of the forest $F-\{z_1,\dots,z_k\}$. Define
\[
\mathsf{Comp}_k(x,y;\bar z):=
\exists a\Bigl(
   \mathsf{lca}(a;x,y)
   \wedge
   \bigwedge_{i=1}^k
      \neg(a\preceq z_i\wedge z_i\preceq x)
   \wedge
   \bigwedge_{i=1}^k
      \neg(a\preceq z_i\wedge z_i\preceq y)
\Bigr).
\]
This formula holds exactly when $x$ and $y$ lie in the same tree of~$F$ and
the unique tree path from $x$ to $y$ avoids all vertices occurring in
$\bar z$. Hence,
\[
   (G,\preceq_F)\models \mathsf{Comp}_k(u,v;\bar w)
\]
if and only if $u$ and $v$ lie in the same connected component of
$F-\{w_1,\dots,w_k\}$.

Next, define a one-step relation between components of $F-\{z_1,\dots,z_k\}$:
\[
\begin{aligned}
\mathsf{Step}_{k}(x,y;\bar z) :=\ 
& \mathsf{Comp}_{k}(x,y;\bar z) \\
& {}\vee
\exists u\,\exists w\bigl(
   \mathsf{Comp}_{k}(x,u;\bar z)
   \wedge
   E(u,w)
   \wedge
   \mathsf{Comp}_{k}(w,y;\bar z)
\bigr).
\end{aligned}
\]

Thus, $\mathsf{Step}_{k}(x,y;\bar z)$ expresses that $x$ and $y$ are either in the same component of $F-\{z_1,\dots,z_k\}$, or that one can move from the component of $x$ to the component of $y$ using one edge of $G-\{z_1,\dots,z_k\}$.

If $k=0$, then $x$ and $y$ are connected in $G$ exactly when they lie in the same tree of $F$, which is expressed by $\mathsf{Comp}_0(x,y)$.
For $k\geq 1$, we claim that $x$ and $y$ are connected in $G-\{z_1,\ldots,z_k\}$ if and only if they can be connected by at most $2k$ applications of $\mathsf{Step}_k$.
The backward implication is immediate from the definition of $\mathsf{Step}_k$: each step is realized by a path in $G-\{z_1,\ldots,z_k\}$.

For the forward implication, work in the connected component of $G$ containing $x$ and $y$, restrict~$F$ to this component, and ignore every $z_i$ outside it.
The restricted forest is a rooted depth-first spanning tree of this component.
We say that a connected component of $F-\{z_1,\dots,z_k\}$ is \emph{internal} if it contains the parent (in $\preceq$) of some $z_i$. Hence, there are at most $k$ internal connected components. Let $P$ be an $x$-$y$ path in
$G-\{z_1,\dots,z_k\}$ minimizing the number of edges whose endpoints lie in
different connected components of $F-\{z_1,\dots,z_k\}$. Note that every
connected component appears at most once. Moreover, if $uv$ is an
inter-component edge of $P$, then, since $uv\in E$ and $F$ is a DFS forest
of $G$, the vertices $u$ and $v$ are comparable in $\preceq$. As $u$ and $v$
lie in different components of $F-\{z_1,\dots,z_k\}$, the
$\preceq$-chain between them contains at least one $z_i$. It follows that at
least one of the components containing $u$ and $v$ in
$F-\{z_1,\dots,z_k\}$ is internal.
As the $k$ internal components are each visited at most once in $P$, and each touches at most $2$ inter-component edges, this yields at most $2k$ inter-component edges in $P$, and at most $2k+1$
connected components of $F\setminus\{z_1,\dots,z_k\}$.

Consequently, $x$ and $y$ are connected in $G-\{z_1,\dots,z_k\}$ if and only if the following first-order formula holds:
For $k=0$, put
\[
   \mathsf{Conn}_0(x,y):=\mathsf{Comp}_0(x,y).
\]
For $k\geq 1$, put
\[
\mathsf{Conn}_{k}(x,y;\bar z) :=
\exists v_1\cdots \exists v_{2k+1}\,\Bigl((v_1=x)\wedge (v_{2k+1}=y)\wedge
\bigwedge_{j=1}^{2k}
   \mathsf{Step}_{k}(v_j,v_{j+1};\bar z)\Bigr).
\]
Since $\mathsf{Step}_k$ is reflexive on the undeleted vertices, a sequence with fewer than $2k$ steps can be padded by repetitions.

Thus, for every $u,v\in V(G)$ and $\bar w\in V(G)^k$, we have
    \[
    G\models \conn(u,v,\bar w)\quad\iff\quad (G,\preceq_F)\models \mathsf{Conn}_k(u,v,\bar w).
    \]

This completes the forward direction.
For the converse direction, assume that $(Y,\preceq)$ is a forest equipped with a rooted forest order defined by some rooting of $Y$, and let $R$ mark the roots.  

The relation $u\preceq v$ holds if and only if $v$ is connected to the marked
root of its tree and deleting~$u$ destroys this connection. Thus, define
\[
\varrho(x,y):=
   \exists r\bigl(
      R(r)
      \wedge \conn(r,y)
      \wedge \neg\conn(r,y;x)
   \bigr).
\]
Then
\[
(Y,\preceq)\models u\preceq v
\quad\Longleftrightarrow\quad
(Y,R)\models\varrho(u,v).\qedhere
\]
\end{proof}

The next theorem yields the promised equivalence for arbitrary binary
relational structures.  All relations from the original signature are
carried through the translations unchanged.

\begin{theorem}[store={thm:scon=FOFO},
    restate-keys={note=see \zcpageref[nocap]{thm:scon=FOFO}}]
\mbox{}
\smallskip

\FOscon~is equivalent to \FO~with DFS spanning forest orders in the following sense.
\medskip

\noindent
\textup{\underline{\FO~with DFS spanning forest orders can simulate \FOscon:}}
\smallskip

\noindent
Let $\sigma$ be a finite binary relational signature containing a
distinguished binary relation symbol~$Z$ and not containing~$\preceq$.

For every formula $\varphi(\bar x)$ of $\FOscon[\sigma]$, there is an
effectively computable formula
$\psi(\bar x)\in\FO[\sigma\cup\{\preceq\}]$ such that the following holds.
Let $\mathbf M$ be a $\sigma$-structure in which $Z(\mathbf M)$ is symmetric
and irreflexive, and let $F$ be any rooted depth-first search forest of
$\MZ$. Then, for every tuple
$\bar v\in M^{|\bar x|}$,
\[
   \mathbf M\models\varphi(\bar v)
   \quad\Longleftrightarrow\quad
   (\mathbf M,\preceq_F)\models\psi(\bar v).
\]
\medskip

\noindent

\textup{\underline{\FOscon~can simulate \FO~with spanning forest orders:}}

\smallskip

\noindent
Let $\tau$ be a finite binary
relational signature containing neither $Z$, $R$, nor~$\preceq$.

For every formula
$\chi(\bar x)\in\FO[\tau\cup\{\preceq\}]$, there is an
effectively computable formula $\theta(\bar x)$ of
$\FOscon[\tau\cup\{Z,R\}]$, where $Z$ and $R$ are fresh binary and unary
relation symbols, respectively, with the
following property. Let $\mathbf M$ be a $\tau$-structure, let $\preceq$ be a
rooted forest order on~$M$, let $Z_\preceq$ be the set of undirected cover edges
of $\preceq$, 
and let $R_\preceq$ mark exactly the roots of~$\preceq$.
Let $\mathbf M_\preceq$ be the expansion of $\mathbf M$ in which
$ Z(\mathbf M_\preceq)=Z_\preceq$ and $R(\mathbf M_\preceq)=R_\preceq$.
Then, for every $\bar v\in M^{|\bar x|}$,
\[
   (\mathbf M,\preceq)\models\chi(\bar v)
   \quad\Longleftrightarrow\quad
   \mathbf M_\preceq\models\theta(\bar v).
\]
\end{theorem}

\begin{proof}
For each $k$, let $\psi_k^Z$ be obtained from $\psi_k$ in
\zcref{lem:con=FO} by replacing every atom $E(u,v)$ by $Z(u,v)$.
For the first translation, replace every atom
$\scon(x,y;z_1,\ldots,z_k)$ by $\psi_k^Z(x,y;\bar z)$. Ordinary
atomic formulas, Boolean connectives, and quantifiers are left unchanged.

For the converse, leave the symbols of $\tau$ unchanged and replace each
atom $x\preceq y$ by
\[
   \exists r\bigl(
      R(r)
      \wedge \scon(y,r)
      \wedge \neg\scon(y,r;x)
   \bigr).
\]
Because $Z(\mathbf M_\preceq)$ is exactly the symmetric edge relation of the
cover forest of $\preceq$, correctness follows from the second part of
\zcref{lem:con=FO}.
\end{proof}

\subsection{Model checking}\label{sec-MC}

We now derive the model-checking result for classes excluding a fixed minor, which follows immediately from \zcref{thm:scon=FOFO,cor-model-checking}.

\begin{theorem}[store=thm:main-mc,
    restate-keys={note=see \zcpageref[nocap]{thm:main-mc}}]
There exist a computable function $f$ and an absolute constant $c$ such that
the following holds. Given a finite binary relational structure $\mathbf M$
over a finite signature containing the distinguished binary relation
symbol~$Z$, where $Z(\mathbf M)$ is symmetric and irreflexive, and an
$\FOscon$-sentence~$\varphi$ over its signature, one can decide whether
$\mathbf M\models\varphi$ in time
\[
   f\bigl(|\varphi|,h(\Gaif(\mathbf M))\bigr)\cdot\|\mathbf M\|^c,
\]
where $\|\mathbf M\|$ denotes the size of the input encoding. In particular,
$\FOscon$ model checking is fixed-parameter tractable, parameterized by formula length, on every class of binary structures whose Gaifman graphs exclude a
fixed minor.
\end{theorem}

\begin{proof}
Let $\sigma_\varphi$ be the set of relation symbols other than~$Z$ that occur
in~$\varphi$, and let $\mathbf M_\varphi$ be the
$(\sigma_\varphi\cup\{Z\})$-reduct of~$\mathbf M$. This does not change
whether $\varphi$ holds, while
$|\sigma_\varphi|\leq|\varphi|$,
$\|\mathbf M_\varphi\|\leq\|\mathbf M\|$, and
\[
 h(\Gaif(\mathbf M_\varphi))
 \leq h(\Gaif(\mathbf M)).
\]

Compute a rooted depth-first search forest $F_Z$ of
$\mathbf M_\varphi\!\upharpoonright_Z$, including every isolated vertex as a
one-vertex tree. Since $\mathbf M_\varphi\!\upharpoonright_Z$ is a subgraph of
$\Gaif(\mathbf M_\varphi)$, the forest $F_Z$ is a rooted spanning forest of
$\Gaif(\mathbf M_\varphi)$. By \zcref{thm:scon=FOFO}, one can effectively
compute a first-order sentence~$\psi$ such that
\[
   \mathbf M\models\varphi
   \quad\Longleftrightarrow\quad
   (\mathbf M_\varphi,\preceq_{F_Z})\models\psi.
\]
Moreover, $|\psi|$ is bounded by a computable function of~$|\varphi|$.
Applying \zcref{cor-model-checking} to the right-hand side proves the claimed
running time. The computation of the reduct, the Gaifman graph, and the DFS
forest takes polynomial time with an absolute exponent and is absorbed
into~$f$ and~$c$.
\end{proof}

\subsection{Monadic dependence}\label{subsec:dependence}

We now show that on weakly sparse graph classes $\FOscon$ characterizes classes with excluded minors. Recall that a class is \emph{monadically dependent for $\FOscon$} if every expansion by finitely many unary predicates is $\FOscon$-dependent, that is, no $\FOscon$-formula
has the independence property on the class of all such expansions.

For a class $\mathscr K$ of relational structures, write
$\Gaif(\mathscr K):=\{\Gaif(\mathbf M)\mid \mathbf M\in\mathscr K\}$.  For a graph class
$\Cc$, let
\[
   \Cc_Z
   :=
   \{(G,Z)\mid G\in\Cc,\ Z\subseteq E(G)\}.
\]

\begin{theorem}[store=thm:dependence,
    restate-keys={note=see \zcpageref[nocap]{thm:dependence}}]
Let $\mathscr K$ be a class of finite binary relational structures over a fixed
signature containing the distinguished binary relation symbol~$Z$, such that
$Z(\mathbf M)$ is symmetric and irreflexive for every
$\mathbf M\in\mathscr K$. If
$\Gaif(\mathscr K)$ excludes a fixed minor, then $\mathscr K$ is monadically
dependent for $\FOscon$.

Moreover, if $\Cc$ is a weakly sparse class of graphs, then
\[
   \Cc_Z\text{ is monadically dependent for }\FOscon
   \quad\Longleftrightarrow\quad
   \Cc\text{ excludes a fixed minor}.
\]
\end{theorem}

\begin{proof}
Assume first that $\Gaif(\mathscr K)$ excludes a fixed minor. For every
$\mathbf M\in\mathscr K$, choose a rooted depth-first search forest
$F_{\mathbf M}$ of $\MZ$, and let
\[
   \Dd:=
   \{(\mathbf M,\preceq_{F_{\mathbf M}}):\mathbf M\in\mathscr K\}.
\]
Each $F_{\mathbf M}$ is a rooted spanning forest of
$\Gaif(\mathbf M)$. By \zcref{thm-main-structure}, the class of all
forest-ordered expansions of the graphs in $\Gaif(\mathscr K)$ has bounded
twin-width and is therefore monadically dependent. By
\zcref{lem:structure-from-gaif,fact:expand}, the class~$\Dd$ is contained in
a fixed transduction of that forest-ordered graph class. Consequently,
$\Dd$ is monadically dependent as well.

Suppose, towards a contradiction, that $\mathscr K$ is monadically
independent for~$\FOscon$.  Then some $\FOscon$-formula
$\varphi(\bar x;\bar y)$ has the independence property on a class of monadic
expansions of~$\mathscr K$.  Equip every underlying structure with the fixed
DFS forest order chosen above.  The added unary predicates make the resulting
structures a monadic expansion of~$\Dd$.  By the first translation in
\zcref{thm:scon=FOFO}, the formula $\varphi$ is replaced uniformly by a
first-order formula $\varphi^\dagger$ with exactly the same truth values on
the witnessing tuples.  Hence, $\varphi^\dagger$ has the independence property
on a monadic expansion of~$\Dd$, contradicting the monadic dependence of
$\Dd$.  This proves the first assertion.

Now let $\Cc$ be a weakly sparse graph class.
If $\Cc$ excludes a fixed minor, then $\Gaif(\mathbf M)=G$ for every $\mathbf M\in\Cc_Z$ whose $E$-reduct is~$G$, and the first assertion applies.
Conversely, suppose that $\Cc$ does not exclude a minor.
By \zcref{thm:minor-tows}, some first-order formula has the independence property on a monadic expansion of $\Cc^{\Upsilon}$.
For every $(G,\preceq_F)\in\Cc^{\Upsilon}$, let $\mathbf M_{G,F}\in\Cc_Z$ be the expansion of $G$ in which $Z(\mathbf M_{G,F})$ is the symmetric edge relation of $F$.
Keep the unary predicates of the witnessing expansion and add a fresh unary predicate $R$ marking the roots of $F$.
The second translation in \zcref{thm:scon=FOFO} replaces the witnessing formula by an $\FOscon$-formula with exactly the same truth values on the witnessing tuples.
Hence, this formula has the independence property on a monadic expansion of $\Cc_Z$, so $\Cc_Z$ is not monadically dependent for $\FOscon$.
\end{proof}

We now show hardness of $\FOscon$ model checking on monotone classes that do not exclude a minor, under the efficient encoding hypothesis below.
The encoding algorithm need not output a minor model.
We use an alternative construction that avoids having to find one by using a subcubic coding graph whose minor-minimal host is necessarily a subdivision.

\begin{theorem}[store=thm:main-hardness-new,
    restate-keys={note=see \zcpageref[nocap]{thm:main-hardness-new}}]
Assume that $\Cc$ is monotone, does not exclude any fixed minor, and admits
efficient encoding, that is, there is a polynomial-time algorithm that, given
a graph~$H$, computes a graph $F\in\Cc$ such that $H\leq_m F$ and
$H\not\leq_m F'$ for every proper subgraph $F'$ of~$F$.
Then deciding, given a graph $F\in\Cc$, a subset $Z_F\subseteq E(F)$, and an
$\FOscon$-sentence~$\varphi$, whether $(F,Z_F)\models\varphi$ is
$\mathsf{AW}[\star]$-hard, parameterized by~$|\varphi|$.
\end{theorem}

\begin{proof}
We reduce from first-order model checking on finite linearly-ordered partial
matchings. Thus, the input structure has the form
\[
    \mathbf M=(M,E(\mathbf M),{<}(\mathbf M)),
\]
where $E(\mathbf M)$ is the edge relation of a graph of maximum degree at
most~$1$ and ${<}(\mathbf M)$ is a linear order on~$M$. First-order model
checking on this class is $\mathsf{AW}[\star]$-hard
by~\cite[Theorem~1]{eickmeyer2020modelchecking}.

Enumerate $M$ as
$M=\{v_1,\ldots,v_n\}$ so that
\[
    (v_i,v_j)\in {<}(\mathbf M)
    \quad\Longleftrightarrow\quad
    i<j.
\]
We first construct from~$\mathbf M$ a connected graph $G(\mathbf M)$ all of
whose vertices have degree either~$1$ or~$3$.

Start with the path whose vertex sequence is
\[
    r,v_1,\ldots,v_n,t.
\]
For each $i\in[n-1]$ such that at least one of $v_i$ and~$v_{i+1}$ is
incident with an edge of $E(\mathbf M)$, subdivide the edge $v_iv_{i+1}$
once, and attach a new leaf to the subdivision vertex. If both $v_i$
and~$v_{i+1}$ are isolated in~$\mathbf M$, leave the edge
$v_iv_{i+1}$ unchanged. Add the edge $v_iv_j$ for every
$v_iv_j\in E(\mathbf M)$, and attach one new leaf to every vertex~$v_i$
that is isolated in~$\mathbf M$. Finally, attach two new leaves to~$t$,
and attach to~$r$ two new vertices, each of which receives two new leaf
neighbors. Every vertex of $G(\mathbf M)$ has degree either~$1$ or~$3$.

\begin{figure}[h!t]
    \centering
    \includegraphics[width=0.75\linewidth]{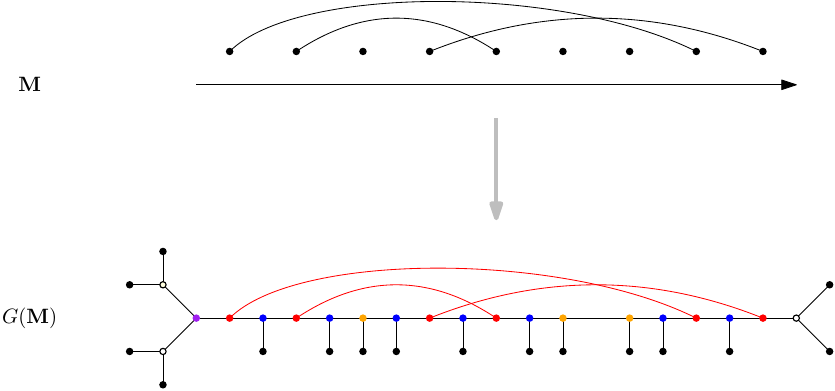}
    \caption{The coding of an ordered partial matching~$\mathbf M$ by
    $G(\mathbf M)$. Degree-$1$ vertices are black; degree-$3$ vertices with
    one black neighbor are preliminarily blue, and those with two black
    neighbors are white. The unique vertex with two white neighbors is
    purple, and all remaining degree-$3$ vertices are red. Finally, a
    preliminarily blue vertex with no red neighbor is recolored orange.
    The red and orange vertices form the interpreted domain, the red-red
    edges are precisely the matching edges, and the purple vertex anchors
    the order encoded by the decorated spine. In the subdivision~$F(\mathbf M)$, the admissible edge set is obtained by deleting one
    edge from each path corresponding to a red-red edge.}
    \label{fig:coding}
\end{figure}

The elements of~$M$ can be recognized intrinsically in~$G(\mathbf M)$.
First color the degree-$1$ vertices black. Among the degree-$3$ vertices,
color those with exactly one black neighbor preliminarily blue and those
with exactly two black neighbors white. 
Color the unique vertex with exactly
two white neighbors purple, and color all remaining degree-$3$ vertices red.
Finally, recolor orange every preliminarily blue vertex having no red
neighbor.

The purple vertex is~$r$. The red vertices are precisely the vertices~$v_i$
that are incident with an edge of $E(\mathbf M)$, whereas the orange
vertices are precisely the isolated vertices of~$\mathbf M$. Indeed, an
auxiliary subdivision vertex is introduced between $v_i$ and~$v_{i+1}$
only if at least one of them is matched, and hence it has a red neighbor and
remains blue. On the other hand, an unmatched vertex~$v_i$ has no red
neighbor: whenever one of its neighbors in the order is matched, the
corresponding spine edge is subdivided. Consequently, the red and orange
vertices are exactly $v_1,\ldots,v_n$. Moreover, the edges with two red
endpoints are exactly the edges of $E(\mathbf M)$. In particular,
$G(\mathbf M)$ with its red-red edges deleted is a tree.

Apply the assumed encoding algorithm to~$G(\mathbf M)$, and denote the
resulting graph by~$F(\mathbf M)$. Thus,
\[
    F(\mathbf M)\in\Cc,\qquad
    G(\mathbf M)\leq_m F(\mathbf M),
\]
and no proper subgraph of~$F(\mathbf M)$ contains $G(\mathbf M)$ as a
minor. We use the following elementary observation, which is well known. 

\medskip
\noindent\emph{Claim.}
Let $K$ be a connected graph all of whose vertices have degree either~$1$
or~$3$. If $K\leq_m F$, but $K\not\leq_m F'$ for every proper subgraph
$F'$ of~$F$, then $F$ is a subdivision of~$K$.

\medskip
By the claim, $F(\mathbf M)$ is a subdivision of~$G(\mathbf M)$. Notice
that the minor model was used only to prove this structural fact, the
reduction does not compute it. Since $G(\mathbf M)$ has no degree-$2$
vertices, its vertices correspond exactly to the vertices of
$F(\mathbf M)$ whose degree is~$1$ or~$3$. We can therefore recover the
suppressed graph and all its subdivided edge paths in linear time: starting
at every vertex of degree~$1$ or~$3$, traverse each incident maximal path
whose internal vertices have degree~$2$. The resulting suppressed graph is
isomorphic to~$G(\mathbf M)$, and the coloring described above can be
computed on it without finding an isomorphism with the labeled copy of
$G(\mathbf M)$.

For every red-red edge of the suppressed graph, choose one edge on its
corresponding path in~$F(\mathbf M)$. Let $D_{\mathbf M}$ be the set of
chosen edges and put
\[
    Z_{\mathbf M}:=E(F(\mathbf M))\setminus D_{\mathbf M}.
\]
Let $\mathbf A_{\mathbf M}$ be the $\{E,Z\}$-structure with
\[
    E(\mathbf A_{\mathbf M})=E(F(\mathbf M))
    \qquad\text{and}\qquad
    Z(\mathbf A_{\mathbf M})=Z_{\mathbf M}.
\]

Since deleting the red-red edges from $G(\mathbf M)$ leaves a tree, $\mathbf A_{\mathbf M}\!\upharpoonright_Z$ is a spanning tree of
$F(\mathbf M)$. Indeed, every path replacing a matching edge is cut once,
and its two remaining pieces hang from the subdivision of that tree.

We now explain how to recover~$\mathbf M$ by a fixed
$\FOscon$-interpretation. The branch vertices are precisely the vertices
of degree~$1$ or~$3$ in~$F(\mathbf M)$. Since the admissible graph is a
tree, two distinct branch vertices are consecutive on an admissible path
if deleting any other branch vertex leaves them connected. This is
expressible in~$\FOscon$ and recovers the suppressed graph with its
red-red edges removed. We can then apply the coloring described above,
using degrees in~$F(\mathbf M)$ for the degree tests and the recovered
adjacency for the tests on neighbors. The omitted red-red edges do not
affect these tests on neighbors. Thus, the red and orange vertices give
the domain of the interpreted structure.

To recover the matching, consider an edge outside~$Z_{\mathbf M}$.
Its endpoints lie in the two remaining halves of a subdivided matching
edge. Each endpoint is joined to a unique branch vertex by an admissible
path containing no other branch vertex. This also includes the case
where the endpoint is itself a branch vertex. The same deletion test
expresses this condition in~$\FOscon$. Two red vertices are therefore
matched precisely when they correspond in this way to the endpoints
of an edge outside~$Z_{\mathbf M}$.

Finally, the purple vertex~$r$ marks the beginning of the spine.
For distinct red or orange vertices~$x,y$, we have
\[
    x<y \quad\Longleftrightarrow\quad \scon(r,x;y).
\]
Indeed, deleting~$y$ leaves~$r$ connected to~$x$ exactly when~$x$
occurs before~$y$ on the spine.

Given a first-order sentence~$\varphi$, we replace its matching and
order atoms by these definitions and relativize its quantifiers to
the red and orange vertices. The resulting $\FOscon$-sentence
$\varphi^\dagger$ satisfies
\[
    \mathbf M\models\varphi
    \quad\Longleftrightarrow\quad
    \mathbf A_{\mathbf M}\models\varphi^\dagger.
\]
All formulas defining the interpretation are fixed, so
$|\varphi^\dagger|=O(|\varphi|)$.

The graph $G(\mathbf M)$ has size linear in~$|\mathbf M|$, and the encoding
algorithm computes $F(\mathbf M)$ in polynomial time. Suppressing the
degree-$2$ paths, computing the intrinsic coloring, and constructing
$Z_{\mathbf M}$ all take linear time in~$|F(\mathbf M)|$. We have therefore
obtained an fpt-reduction from first-order model checking on ordered
partial matchings, which proves the theorem.
\end{proof}

\section{Combinatorial reconfiguration}\label{sec:reconfiguration}
\newcommand{\legal}{\mathrm{legal}}

We now present the algorithmic application to combinatorial reconfiguration, discussed in the introduction. 

Fix a graph $G$ and sets $S,T\subseteq V(G)$. We think of $S$ as a set of
initial tokens placed on the graph. A \emph{token sliding sequence of length
$\ell$} from~$S$ to~$T$ is a sequence $S=S_0,\dots,S_\ell=T$ satisfying, for
every $i\in[\ell]$,
\[
    S_i=(S_{i-1}\setminus\{u_i\})\cup\{v_i\}
\]
for some $u_i\in S_{i-1}$, $v_i\in V(G)\setminus S_{i-1}$, and
$u_iv_i\in E(G)$. This means that each step slides one token to an unoccupied
neighbor.

Let $\mathcal L$ be a logic over the signature of colored graphs.
The \textsc{$\mathcal L$-Modulator-Reconfiguration} problem is defined as follows.
Given a (possibly colored) graph $G$, sets $S,T\subseteq V(G)$, a movement budget $b\in\N$, and an $\mathcal L$-sentence $\psi$, decide whether there exists a token sliding sequence $S=S_0,\dots,S_\ell=T$ from $S$ to $T$ of length $\ell\le b$ such that
\[
G - S_i \models \psi
\qquad\text{for every } 0\le i\le\ell .
\]
When the formula $\psi$ is fixed, we also speak of the \textsc{$\psi$-Modulator-Reconfiguration} problem.

\medskip 

Let us show how this captures some standard reconfiguration problems, and why
the use of \FOscon\ is natural here. Let $P_S$ be a fresh unary relation
symbol, and let $\mathbf G_S$ be the expansion of $G$ in which
\[
    P_S(\mathbf G_S)=S
    \qquad\text{and}\qquad
    Z(\mathbf G_S)
    :=
    \{(u,v)\in E(\mathbf G_S) : u\notin S \text{ and } v\notin S\}.
\]
Thus, $\mathbf G_S\!\upharpoonright_Z$ is obtained from $G-S$ by adding the vertices of
$S$ as isolated vertices. In the following formulas, $Z$ and $P_S$ denote
the corresponding relation symbols, interpreted in $\mathbf G_S$ as above.
Over the course of the reconfiguration sequence, the newly occupied vertices
will be deleted in the $\scon$ atoms, while connections through vertices that
have become free will be handled explicitly. The examples state the desired
properties for the initial token configuration.

As a first example, \textsc{Feedback Vertex Set Reconfiguration} is obtained by requiring that
$G-S$ is acyclic.  In $\FOscon$ this can be expressed as
\[
    \psi_{\mathrm{fvs}}
    :=
    \forall x\forall y\forall z\,
    \Bigl(
        Z(x,y)\wedge Z(x,z)\wedge y\neq z
        \rightarrow
        \neg \scon(y,z;x)
    \Bigr).
\]
This forbids cycles in $\mathbf G_S\!\upharpoonright_Z$. Since the vertices of~$S$ are
isolated in this graph, this is equivalent to saying that $G-S$ is acyclic,
and hence says exactly that~$S$ is a feedback vertex set.

As a second example, suppose that two distinct terminals are marked by unary
relation symbols~$P_s$ and $P_t$. Then \textsc{$s$-$t$ Cut Reconfiguration}
is captured by the $\FOscon$ sentence
\[
\begin{aligned}
    \psi_{s,t}
    :=
    \exists s\exists t\Bigl(
        &P_s(s)\wedge P_t(t)\wedge s\neq t
        \wedge \neg P_S(s)\wedge \neg P_S(t) \\
        &\wedge \forall x(P_s(x)\rightarrow x=s)
        \wedge \forall x(P_t(x)\rightarrow x=t)
        \wedge \neg \scon(s,t)
    \Bigr).
\end{aligned}
\]
The subformula $\neg\scon(s,t)$ says that there is no $s$-$t$ path in
$G-S$, while the conjuncts $\neg P_S(s)$ and $\neg P_S(t)$ ensure that the
terminals themselves are not deleted.
\medskip

We show the following meta theorem. 

\begin{theorem}[store=thm:modulatorreconfiguration,restate-keys={note=see \zcpageref[nocap]{thm:modulatorreconfiguration}}]
There exist a computable function $f$ and an absolute constant~$c$ such that,
given a possibly unary-colored graph~$G$, sets $S,T\subseteq V(G)$, a budget
$b\in\N$, and a sentence $\psi\in\FOconn$, one can decide the
\textsc{$\psi$-Modulator-Reconfiguration} problem in time
\[
f(|\psi|,h(G),b)\cdot |G|^c.
\]
\end{theorem}

\begin{proof}
Let $P_S$ and $P_T$ be fresh unary relation symbols and $Z$ a fresh binary
relation symbol. Given an instance $(G,S,T,b,\psi)$, let
$\mathbf G_{S,T}$ be the expansion of~$G$ in which
$P_S(\mathbf G_{S,T})=S$, $P_T(\mathbf G_{S,T})=T$, 
and $Z(\mathbf G_{S,T})
    :=
    \{uv\in E(G):\{u,v\}\cap S=\emptyset\}$.
The Gaifman graph of $\mathbf G_{S,T}$ is exactly~$G$.

Let $S'$ be a configuration occurring after at most~$b$ token slides, and
put
$F:=S\setminus S'$ and $B:=S'\setminus S$.
Since every slide vacates one occupied vertex and occupies one adjacent
unoccupied vertex,
\[
    S'=(S\setminus F)\cup B,
    \qquad
    F\subseteq S,\quad B\subseteq V(G)\setminus S,
    \quad |F|=|B|\leq b.
\]

We represent $F$ and~$B$ by tuples of at most~$b$ vertices; membership is
then expressed by disjunctions of equalities.
For every such pair $(F,B)$, the assertion $G-\bigl((S\setminus F)\cup B\bigr)\models\psi$
is expressible by an $\FOscon$-formula whose length depends only on
$|\psi|$ and~$b$. For a connectivity atom of the $\FOconn$-sentence $\psi$, observe that every
simple path in the remaining graph can be split into at most $2b$ individual
edges incident with~$F$ and at most $2b+1$ subpaths using only edges from~$Z$
and avoiding~$B$. It can therefore be expressed by a chain of at most
$4b+1$ edge or sub-connectivity atoms.

It is now entirely standard to express a token sliding sequence of length at
most~$b$. For every $\ell\leq b$, quantify representations $(F_i,B_i)$ as
above of configurations $S_i=(S\setminus F_i)\cup B_i$, require
$S_0=S$ and $S_\ell=T$, and require $G-S_i\models\psi$ for every
$0\leq i\leq\ell$ using the translation above.
For every $i\in[\ell]$, require adjacent vertices
$u_i\in S_{i-1}$ and $v_i\notin S_{i-1}$ such that
$S_i=(S_{i-1}\setminus\{u_i\})\cup\{v_i\}$.

Taking the disjunction over all $\ell\leq b$ gives an effectively computable $\FOscon$-sentence $\varphi_{\psi,b}$ such that
\[
    \mathbf G_{S,T}\models\varphi_{\psi,b}
    \quad\Longleftrightarrow\quad
    (G,S,T,b)
    \text{ is a yes-instance of }
    \textsc{$\psi$-Modulator-Reconfiguration}.
\]
Note that the length of $\varphi_{\psi,b}$ depends only on $|\psi|$ and~$b$.
Since the Gaifman graph of $\mathbf G_{S,T}$ is~$G$, the claimed running time
follows from~\zcref{thm:main-mc}.
\end{proof}

The examples given above give the following corollary. 
\begin{corollary}
\textsc{Vertex-Cut-Reconfiguration} and \textsc{Feedback-Vertex-Set-Reconfiguration} are fixed-parameter tractable when parameterized by $b$ on graph classes excluding a fixed minor.
\end{corollary}

\subsection{Reconfiguration variants and solution discovery}
Our algorithm for the \textsc{$\psi$-Modulator-Reconfiguration} problem~(\zcref{thm:modulatorreconfiguration}) reduces the problem to model checking a $\FOscon$ formula.
This approach is highly flexible and shows the power of algorithmic meta-theorems: with simple modifications to the constructed formula, we can easily adapt the algorithm to other variations of the reconfiguration framework.
A common variation is to change the rules for how one token set $S_i$ can be transformed to the next set $S_{i+1}$.
Popular models include
\begin{itemize}
    \item the \emph{token sliding} model, where in each step a token is moved along a single edge~\cite{nishimura18reconfiguration},
    \item the \emph{token jumping} model, where in each step a token moves to an arbitrary vertex~\cite{nishimura18reconfiguration},
    \item the \emph{token addition and removal} model, where
    in each step either a token is added to an arbitrary vertex, or an existing token is removed~\cite{nishimura18reconfiguration}, and
    \item the \emph{token sliding along a path}\footnote{To the best of our knowledge, there is no established name for this model.} model, where in each step a token moves along an arbitrarily long path of unoccupied vertices~\cite{cualinescu2008reconfigurations,cooper2024parameterized}.
\end{itemize}
We have given the proof of \zcref{thm:modulatorreconfiguration} for the token sliding model, but it can be easily adapted to the other models.  
Especially the last model fits $\FOscon$ perfectly. 

Another emerging variant of the reconfiguration framework is the solution
discovery framework, which was recently introduced
in~\cite{fellows2025onsd} and further studied
in~\cite{bousquet2025algorithmic,GroblerMMMRSS24,GroblerMMNRS24,saito2026solution,von2026separating}.
Here $S$ is viewed as a tentative solution that may fail to satisfy the target
property and should be transformed into an a priori unknown solution~$T$.
More formally, in the \textsc{$\psi$-Modulator-Discovery} problem the input
consists of $G,S,b$, and~$\psi$, and the question is whether, for some
$\ell\leq b$, there exists a token sliding sequence
\[
    S=S_0,\ldots,S_\ell=T
\]
such that $G-T\models\psi$. No feasibility requirement is imposed on the
intermediate sets. The proof of \zcref{thm:modulatorreconfiguration} can be easily adapted to prove the following theorem. The result entails fixed-parameter tractability of
\textsc{Feedback-Vertex-Set-Discovery} and
\textsc{Vertex-Cut-Discovery} on classes excluding a fixed minor,
parameterized by the budget~$b$.

\begin{theorem}[store=thm:modulatordiscovery,restate-keys={note=see \zcpageref[nocap]{thm:modulatordiscovery}}]
There exist a computable function $f$ and an absolute constant $c$ such
that, given a possibly unary-colored graph~$G$, a set $S\subseteq V(G)$, a
budget $b\in\N$, and a sentence $\psi\in\FOconn$, one can decide the
\textsc{$\psi$-Modulator-Discovery} problem in time
\[
    f(|\psi|,h(G),b)\cdot |G|^c.
\]
\end{theorem}

\section{Further applications and discussion}

In this section, we discuss two extensions of sub-connectivity logic and its
connection with regular path queries. Throughout, $Z$ denotes a fresh binary
relation symbol. If $Z$ is interpreted by a symmetric and irreflexive
relation $F\subseteq M^2$, we write $\mathbf M[Z:=F]$ for the resulting
expansion. The applicability and running-time bound of the positive
model-checking theorem depend on this expansion through
$h(\Gaif(\mathbf M[Z:=F]))$; unless explicitly required below,
$Z(\mathbf M[Z:=F])=F$ need not be contained in any relation of the original
signature.

\subsection{\texorpdfstring{\GFOscon}{Guarded-FO+subconn}}\label{subsec-guarded-fragment}

We now describe a new logic, which we call {\GFOscon}, which generalizes
$\FOscon$ while still admitting efficient model checking. The idea is to fix the Gaifman graph of a structure and allowing to update the $Z$-relation in a controlled way. 

We consider a binary relational signature $\sigma$ containing a special
binary relation symbol~$E$, and let $Z\notin\sigma$ be a fresh binary relation
symbol. We restrict ourselves to $\sigma$-structures $\mathbf M$ in which
$E(\mathbf M)$ is symmetric and irreflexive and, for every binary relation
symbol $R\in\sigma$,
\[
 R(\mathbf M)\setminus\{(v,v):v\in M\}\subseteq E(\mathbf M).
\]
Thus, the $E$-reduct is precisely the Gaifman graph of~$\mathbf M$, which we consider as fixed. 

We define the syntax inductively. Put $\mathcal L_0:=\FO[\sigma]$. Inductively, let
$\mathcal L_{r+1}$ be the Boolean closure of $\mathcal L_r$ together with all
formulas
\[
 [\phi]\,\psi(\bar x),
\]
where $\phi(p,q)\in\mathcal L_r$ has precisely the free variables $p$ and $q$
and $\psi(\bar x)\in\FOscon[\sigma\cup\{Z\}]$, with~$Z$ as the distinguished
sub-connectivity relation symbol of this full signature. Set
\[
 \GFOscon[\sigma]:=\bigcup_{r\geq0}\mathcal L_r.
\]
The semantics is inductive on~$r$. For $\phi(p,q)\in\mathcal L_r$, define
\[
 Z_{\phi,\mathbf M}
 :=
 \{(a,b)\in M^2:\mathbf M\models
 E(a,b)\wedge(\phi(a,b)\vee\phi(b,a))\}.
\]
Then
\[
 \mathbf M\models[\phi]\,\psi(\bar v)
 \quad\Longleftrightarrow\quad
 \mathbf M[Z:=Z_{\phi,\mathbf M}]\models\psi(\bar v).
\]

\begin{theorem}\label{thm:MC-guarder-logic}
There exist a computable function $g$ and an absolute constant $C$ such that,
given a finite binary signature~$\sigma$, a $\sigma$-structure $\mathbf M$ of
the kind just described, and a sentence $\chi\in\GFOscon[\sigma]$, one can
decide whether $\mathbf M\models\chi$ in time
\[
 g\bigl(|\chi|,|\sigma|,h(\Gaif(\mathbf M))\bigr)\cdot\|\mathbf M\|^C.
\]
\end{theorem}
\begin{proof}
Evaluate the subformulas in increasing rank and memoize the binary relation
defined by every guard. To compute the relation defined by a guard
$\phi(p,q)$, evaluate it on all ordered pairs $(a,b)\in M^2$. Assigned
vertices are encoded by fresh unary singleton predicates. Rank-zero formulas
are evaluated by \zcref{thm:main-mc} after expanding $\mathbf M$ by the empty
interpretation of~$Z$. Every scheme occurrence
$[\eta]\,\psi(\bar x)$ is evaluated, once $Z_{\eta,\mathbf M}$ has been
computed, by applying \zcref{thm:main-mc} to
$\mathbf M[Z:=Z_{\eta,\mathbf M}]$ with the appropriate singleton
predicates. Because the construction uses only Boolean closure, a scheme
occurrence inside a guard with free variables $p,q$ has no other free
variables; its truth table therefore has at most $|M|^2$ entries. A
top-level scheme occurrence in the input sentence has no free variables and requires only one evaluation. Memoization ensures that each occurrence is
processed only once.

Every relation $Z_{\eta,\mathbf M}$ is contained in $E(\mathbf M)$, and the
additional singleton predicates are unary. Hence, every structure supplied to
\zcref{thm:main-mc} has the same Gaifman graph as~$\mathbf M$. Let $c$ be the
absolute exponent from that theorem. The factor $|M|^2$ needed to tabulate a
binary guard raises the polynomial exponent to at most $c+2$, independently
of $\chi$, $\sigma$, and the excluded minor. All remaining dependence is
absorbed into a computable function of
$|\chi|$, $|\sigma|$, and $h(\Gaif(\mathbf M))$, which proves the stated bound
with the absolute constant $C:=c+2$.
\end{proof}

For example, let $\eta_1(x,y)$ and $\eta_2(x,y)$ be first-order
formulas defining two equivalence relations on the vertices of a
graph, and let $\psi(\bar x)$ be an $\FOscon$-formula. If we think
about equivalence classes as vertex colors, then, for each
$i\in\{1,2\}$, $[\eta_i]\,\psi(\bar x)$ interprets the
$\scon(x,y;\bar z)$ atom as asserting the existence of a
monochromatic path linking $x$ to $y$ while avoiding~$\bar z$,
with respect to the colors defined by~$\eta_i$.
For example, the sentence
\begin{equation*}
 \bigl[\eta_1\bigr]\,
 \forall x\forall y\,
 \bigl(\eta_1(x,y)\rightarrow\scon(x,y)\bigr)
 \;\wedge\;
 \bigl[\eta_2\bigr]\,
 \forall x\forall y\,
 \bigl(\eta_2(x,y)\rightarrow\scon(x,y)\bigr)
\end{equation*}
states that for each of the two equivalence relations,
all vertices in each equivalence class are pairwise connected
by monochromatic paths. Thus, every class induces a connected
subgraph. The two conjuncts are evaluated independently, each
with its own sub-connectivity relation, which remains fixed
throughout the evaluation of its quantified formula.

\subsection{More than one set of admissible edges}\label{subsec-hard-several-color}

Let $k$ be a positive integer and let
$\tau_k:=\{E,Z_1,\ldots,Z_k\}$, 
where $Z_1,\ldots,Z_k$ are distinct binary relation symbols. We consider
$\tau_k$-structures $\mathbf M$ in which $E(\mathbf M)$ and every
$Z_i(\mathbf M)$ are symmetric and irreflexive. For $i\in[k]$, let
$G_{Z_i}(\mathbf M)$ be the graph with vertex set~$M$ and edge set
\[
   E(G_{Z_i}(\mathbf M))
   :=
   \bigl\{\{u,v\}\in\tbinom{M}{2}:(u,v)\in Z_i(\mathbf M)\bigr\}.
\]
The logic $\FOscon^k$ is obtained from $\FO[\tau_k]$ by adding, for every
$i\in[k]$, atoms $\scon_i(x,y;\bar z)$ asserting that $x$ and $y$ belong to
the same connected component of $G_{Z_i}(\mathbf M)-V(\bar z)$.

For a graph class $\Cc$, let $\Cc_{\bar Z}$ be the class of all
$\tau_k$-structures $\mathbf M$ whose $E$-reduct is a graph $G\in\Cc$ and
for which
\[
   E(G_{Z_i}(\mathbf M))\subseteq E(G)
   \qquad\text{for every }i\in[k].
\]
The interpretations of the symbols $Z_i$ need not be disjoint.

\begin{theorem}
Let $k\geq2$. For every monotone class of graphs $\Cc$, the class
$\Cc_{\bar Z}$ is $\FOscon^k$-dependent if and only if $\Cc$ has bounded
tree-width.
\end{theorem}
\begin{proof}
Assume first that $\Cc$ has unbounded tree-width. For $n\geq2$, fix an
ordering of $2^{[n]}$ whose first element is~$\emptyset$, and let $H_n$ be
the $n\times2^n$ grid with vertices $v_{i,J}$ indexed by
$[n]\times2^{[n]}$ in this order. Attach a pendant vertex $w_{i,J}$ to
$v_{i,J}$ whenever $i\in J$, and put
$W:=\{w_{i,J}:i\in J\subseteq[n]\}$.

For every fixed planar graph $H$, the class of $H$-minor-free graphs has
bounded tree-width~\cite{RobertsonSGM5}.
Hence, some graph in~$\Cc$ contains $H_n$ as a minor.
By monotonicity, $\Cc$ contains an inclusion-minimal subgraph~$H'_n$
supporting such a model. In this minimal model, every branch set is a tree
and exactly one edge is retained for each edge of~$H_n$. 
Let $f\colon V(H'_n)\to V(H_n)$ be the map specifying the resulting surjective minor model.
Minimality
implies that each pendant branch set $f^{-1}(w)$, $w\in W$, is a single
degree-one vertex, and that $H'_n$ has no other degree-one vertices.

Define edge sets $\mathcal Z_1$ and $\mathcal Z_2$ as follows. For each
$xy\in E(H'_n)$:
\begin{itemize}
    \item if $x,y\in f^{-1}(v_{i,J})$ for some grid vertex $v_{i,J}$,
    then $xy$ belongs to both $\mathcal Z_1$ and $\mathcal Z_2$;
    \item if $x\in f^{-1}(v_{i,J})$ and $y\in f^{-1}(v_{i,J'})$, where
    $v_{i,J}$ and $v_{i,J'}$ are horizontally adjacent in $H_n$, then
    $xy\in\mathcal Z_1$;
    \item if $x\in f^{-1}(v_{i,J})$ and $y\in f^{-1}(v_{i',J})$, where
    $v_{i,J}$ and $v_{i',J}$ are vertically adjacent in $H_n$, then
    $xy\in\mathcal Z_2$.
\end{itemize}
Let $\mathbf M_n\in\Cc_{\bar Z}$ be the expansion of $H'_n$ in which
$Z_1(\mathbf M_n)$ and $Z_2(\mathbf M_n)$ are the symmetric relations
induced by $\mathcal Z_1$ and $\mathcal Z_2$, respectively; if $k>2$,
interpret the remaining symbols $Z_i$ by the empty relation. Let
\[
\phi(x,y):=(\exists z)\,
   \bigl(\scon_1(x,z)\wedge\theta(z)\wedge\scon_2(z,y)\bigr),
\]
where
\[
   \theta(z):=
   \exists w\bigl(E(z,w)\wedge
   \forall u\,(E(w,u)\rightarrow u=z)\bigr)
\]
expresses that $z$ has a neighbor of degree~$1$ in the $E$-reduct.

For each $i\in[n]$, select $a_i\in f^{-1}(v_{i,\emptyset})$, and for each
$J\subseteq[n]$, select $b_J\in f^{-1}(v_{1,J})$. Horizontal
$Z_1(\mathbf M_n)$-connectivity fixes the row, vertical
$Z_2(\mathbf M_n)$-connectivity fixes the column, and $v_{i,J}$ has a
pendant neighbor precisely when $i\in J$. Consequently,
\[
   \mathbf M_n\models\phi(a_i,b_J)
   \quad\Longleftrightarrow\quad i\in J.
\]
Thus, $\phi$ has the independence property on $\Cc_{\bar Z}$.

Conversely, assume that $\Cc$ has bounded tree-width. Since every
$Z_i(\mathbf M)$ is contained in $E(\mathbf M)$, the Gaifman graphs of the
structures in $\Cc_{\bar Z}$ have bounded tree-width as well. Every atom $\scon_i(x,y;\bar z)$ is expressible in monadic second-order logic, and hence every $\FOscon^k$-formula translates into an $\MSO$-formula.
The standard tree encoding of structures of bounded tree-width, together with the type bounds for $\MSO$ on trees~\cite{paszke2020vc}, implies dependence also after adding finitely many arbitrary unary predicates.
Consequently, $\Cc_{\bar Z}$ is monadically $\FOscon^k$-dependent.
\end{proof}

We remark that if the interpretations of $Z_1,\ldots,Z_k$ are required to be pairwise
disjoint, the preceding proof no longer applies because every internal
branch-set edge was placed in both~$\mathcal Z_1$ and~$\mathcal Z_2$. For a
minor-closed class the argument can be repaired: unbounded tree-width implies
that the class contains every planar graph $H_n$ itself, whose horizontal and
vertical grid edges can be assigned to $Z_1$ and $Z_2$, respectively, and
are disjoint.

\subsection{Connections to RPQs}\label{sec-rpq}

We now compare our logic with the undirected endpoint semantics of \emph{regular path
queries} (RPQs).
Let~$\Sigma$ be a finite label alphabet and let
$\tau_\Sigma:=\{E\}\cup\{E_a:a\in\Sigma\}$, where all displayed objects are
binary relation symbols. Consider a $\tau_\Sigma$-structure $\mathbf G$ in
which $E(\mathbf G)$ and all $E_a(\mathbf G)$ are symmetric and irreflexive,
and $E_a(\mathbf G)\subseteq E(\mathbf G)$ for every $a\in\Sigma$; an edge
may carry more than one label. Given a regular expression~$Q$ over~$\Sigma$,
a pair $(x,y)$ is an answer to~$Q$ if there are a walk
$x=v_0,v_1,\ldots,v_m=y$ in the $E$-reduct of~$\mathbf G$ and labels
$a_1,\ldots,a_m\in\Sigma$ such that
$(v_{i-1},v_i)\in E_{a_i}(\mathbf G)$ for every $i\in[m]$ and
$a_1\cdots a_m\in L(Q)$. Regular path
queries play a central role in graph databases and are used in query
languages such as SPARQL.

RPQs and $\FOscon$ serve different expressive purposes. RPQs can constrain
the complete label word of a witnessing walk, for example by
$B^\ast R^\ast B^\ast$; such a constraint is not available as a single
sub-connectivity atom. Conversely, $\FOscon$ combines its connectivity atoms
with the full first-order logical overlay, which is not part of an RPQ.

With at least two labels whose edge relations may overlap, allowing arbitrary
RPQ atoms subsumes reachability with respect to two possibly overlapping
admissible-edge relations. Hence, \zcref{subsec-hard-several-color} shows that,
on monotone graph classes, unbounded tree-width implies monadic independence.
Conversely, every fixed RPQ endpoint relation is $\MSO$-definable via its finite
automaton and a standard product-reachability construction. Thus, expansions
of bounded-tree-width graph classes remain monadically dependent, and bounded
tree-width is the resulting dependence frontier. An
$\mathsf{AW}[\star]$-hardness statement on classes of unbounded tree-width
additionally requires an efficient construction of the relevant grid-minor
encodings, including their minor models.

Building on the query-log analysis of Bonifati, Martens, and
Timm~\cite{bonifati2020analytical}, Martens and
Trautner~\cite[Definition~2]{martens2019bridging} call the following
expressions \emph{simple transitive expressions} (STEs). A nonempty set
$A\subseteq\Sigma$ is an \emph{atomic
expression}; an expression of the form $A_1\cdots A_m$ or
$A_1?\cdots A_m?$ is \emph{bounded}, where $m\geq0$ and every $A_i$ is
atomic. An STE has the form
\[
   B_{\mathrm{pre}}\,T^\ast\,B_{\mathrm{suff}},
\]
where $B_{\mathrm{pre}}$ and $B_{\mathrm{suff}}$ are bounded expressions and
$T$ is either $\varepsilon$ or an atomic expression. For a fixed STE, the
endpoint relations defined by $B_{\mathrm{pre}}$ and $B_{\mathrm{suff}}$ are
first-order definable; denote the corresponding formulas by
$\beta_{\mathrm{pre}}$ and $\beta_{\mathrm{suff}}$. If $T$ is a nonempty
set of labels, put
\[
   \theta_T(p,q):=\bigvee_{a\in T}E_a(p,q).
\]
Then the undirected endpoint semantics of the STE is expressed in
$\GFOscon$ by
\[
   \Psi(x,y)
   :=
   [\theta_T]\,
   \exists u\exists w\,
   \bigl(
      \beta_{\mathrm{pre}}(x,u)
      \wedge \scon(u,w)
      \wedge \beta_{\mathrm{suff}}(w,y)
   \bigr).
\]
If $T=\varepsilon$, the middle condition is simply $u=w$, and the query is
first-order definable.

The subsequent analysis in~\cite{martens2019bridging} reports that $99.992\%$
of the $247{,}404$ property paths in the analyzed query logs are STEs or are
trivially equivalent to STEs. Moreover, every remaining expression type in the reported table
except $(ab)^\ast$ is a union of STEs. Thus, the guarded logic covers nearly
all of those logged property-path shapes under the undirected endpoint
semantics.

\paragraph{Dealing with RPQs not expressible in the guarded fragment of the logic}

There is no known \GFOscon\xspace formula expressing that a set is bipartite, or a path to be of even length, and these properties are likely not expressible, similarly to \FOconn\xspace \cite[Theorem 3.12]{schirrmacher2023first}.
However, we can always use the following technique. First compute in linear time a bounded amount of unary (or 0-ary) predicates, second use \cref{thm:MC-guarder-logic} on the graph equipped with these new predicates. We now showcase this reasoning on 
the parity-restricted endpoint relation without vertex deletions, and show that it is still
tractable. Indeed, given the admissible-edge graph, mark by a fresh unary
predicate $N$ all vertices lying in non-bipartite connected components, and
choose one bipartition class in every bipartite connected component and mark
one side by a fresh unary predicate $B$. Then the existence of an even-length
admissible-edge walk is expressed by
\[
  \mathsf{Even}_Z(x,y):=
  \scon(x,y)\wedge
  \bigl(N(x)\vee(B(x)\leftrightarrow B(y))\bigr).
\]
The predicates $N$ and $B$ can be computed in linear time and, being unary,
do not change the Gaifman graph, and the formula can be model-checked efficiently with \cref{thm:MC-guarder-logic} on minor-excluded classes. 

Going further in this direction, we may also consider atoms
\[
  \mathsf{even\text{-}scon}(x,y;\bar z)
\]
asserting the existence of an even-length walk in the admissible-edge graph
after deleting the vertices occurring in $\bar z$. Whether model checking for
this deletion-aware extension remains tractable on minor-excluded classes is
a natural question, which remains open.

\newpage 
\bibliography{ref}

@article{bonifati2020analytical,
  title = {An analytical study of large {SPARQL} query logs},
  author = {Bonifati, Angela and Martens, Wim and Timm, Thomas},
  journal = {The VLDB Journal},
  volume = {29},
  number = {2--3},
  pages = {655--679},
  year = {2020},
  publisher = {Springer},
  xdoi = {10.1007/s00778-019-00558-9}
}

@article{martens2019bridging,
  title = {Bridging theory and practice with query log analysis},
  author = {Martens, Wim and Trautner, Tina},
  journal = {ACM SIGMOD Record},
  volume = {48},
  number = {1},
  pages = {6--13},
  year = {2019},
  publisher = {ACM New York, NY, USA},
  xdoi = {10.1145/3371316.3371319}
}

@book{Flum2006ParameterizedCT,
  title = {Parameterized Complexity Theory},
  author = {J{\"o}rg Flum and Martin Grohe},
  publisher = {Springer},
  year = {2006},
  series = {Texts in Theoretical Computer Science. An EATCS Series},
  xdoi = {10.1007/3-540-29953-X}
}

@article{ThatcherWright1968,
  author = {James W. Thatcher and Jesse B. Wright},
  title = {Generalized finite automata theory with an application to a decision problem of second-order logic},
  journal = {Mathematical Systems Theory},
  volume = {2},
  number = {1},
  pages = {57--81},
  year = {1968},
  xdoi = {10.1007/BF01691346}
}

@article{Doner1970,
  author = {John Doner},
  title = {Tree acceptors and some of their applications},
  journal = {Journal of Computer and System Sciences},
  volume = {4},
  number = {5},
  pages = {406--451},
  year = {1970},
  xdoi = {10.1016/S0022-0000(70)80041-1}
}

@incollection{books/cu/p/Heuvel13,
  author = {Jan van den Heuvel},
  title = {The Complexity of Change},
  booktitle = {Surveys in Combinatorics 2013},
  series = {London Mathematical Society Lecture Note Series},
  volume = {409},
  pages = {127--160},
  publisher = {Cambridge University Press},
  year = {2013},
  xdoi = {10.1017/CBO9781139506748.005}
}

@article{journals/cosrev/BousquetMNS24,
  author = {Nicolas Bousquet and Amer E. Mouawad and Naomi Nishimura and Sebastian Siebertz},
  title = {A Survey on the Parameterized Complexity of Reconfiguration Problems},
  journal = {Computer Science Review},
  volume = {53},
  pages = {100663},
  year = {2024},
  xdoi = {10.1016/j.cosrev.2024.100663}
}

@article{journals/algorithmica/MouawadN0SS17,
  author = {Amer E. Mouawad and Naomi Nishimura and Venkatesh Raman and Narges Simjour and Akira Suzuki},
  title = {On the Parameterized Complexity of Reconfiguration Problems},
  journal = {Algorithmica},
  volume = {78},
  number = {1},
  pages = {274--297},
  year = {2017},
  xdoi = {10.1007/s00453-016-0159-2}
}

@article{journals/tcs/BousquetHKMS23,
  author = {Nicolas Bousquet and Felix Hommelsheim and Yusuke Kobayashi and Moritz M{\"u}hlenthaler and Akira Suzuki},
  title = {Feedback Vertex Set Reconfiguration in Planar Graphs},
  journal = {Theoretical Computer Science},
  volume = {979},
  pages = {114188},
  year = {2023},
  xdoi = {10.1016/j.tcs.2023.114188}
}

@article{journals/rairo-ro/GomesNS25,
  author = {Guilherme de Castro Mendes Gomes and S{\'e}rgio H. Nogueira and Vin{\'i}cius Fernandes dos Santos},
  title = {Some Results on Vertex Separator Reconfiguration},
  journal = {RAIRO Operations Research},
  volume = {59},
  number = {1},
  pages = {725--742},
  year = {2025},
  xdoi = {10.1051/ro/2025006}
}

@article{journals/jcss/GomesLMMOSZ24,
  author = {Guilherme C. M. Gomes and Cl{\'e}ment Legrand-Duchesne and Reem Mahmoud and Amer E. Mouawad and Yoshio Okamoto and Vin{\'i}cius Fernandes dos Santos and Tom C. van der Zanden},
  title = {Minimum Separator Reconfiguration},
  journal = {Journal of Computer and System Sciences},
  volume = {146},
  pages = {103574},
  year = {2024},
  xdoi = {10.1016/j.jcss.2024.103574}
}

@inproceedings{conf/iwpec/MouawadNRW14,
  author = {Amer E. Mouawad and Naomi Nishimura and Venkatesh Raman and Marcin Wrochna},
  title = {Reconfiguration over Tree Decompositions},
  booktitle = {Parameterized and Exact Computation -- IPEC 2014},
  series = {Lecture Notes in Computer Science},
  volume = {8894},
  pages = {246--257},
  publisher = {Springer},
  year = {2014},
  xdoi = {10.1007/978-3-319-13524-3_21}
}

@article{conf/esa/GimaIKO22,
  author = {Tatsuya Gima and Takehiro Ito and Yasuaki Kobayashi and Yota Otachi},
  title = {Algorithmic Meta-Theorems for Combinatorial Reconfiguration Revisited},
  journal = {Algorithmica},
  volume = {86},
  number = {11},
  pages = {3395--3424},
  year = {2024},
  xdoi = {10.1007/s00453-024-01261-0}
}

@article{journals/jcss/LokshtanovMPRS18,
  author = {Daniel Lokshtanov and Amer E. Mouawad and Fahad Panolan and M. S. Ramanujan and Saket Saurabh},
  title = {Reconfiguration on Sparse Graphs},
  journal = {Journal of Computer and System Sciences},
  volume = {95},
  pages = {122--131},
  year = {2018},
  xdoi = {10.1016/j.jcss.2018.02.004}
}

@inproceedings{kreutzer2010lower,
  title = {Lower bounds for the complexity of monadic second-order logic},
  author = {Kreutzer, Stephan and Tazari, Siamak},
  booktitle = {2010 25th Annual IEEE Symposium on Logic in Computer Science},
  pages = {189--198},
  year = {2010},
  organization = {IEEE},
  xdoi = {10.1109/LICS.2010.39}
}

@article{eickmeyer2020modelchecking,
  author = {Eickmeyer, Kord and van den Heuvel, Jan and Kawarabayashi, Ken-ichi and Kreutzer, Stephan and {Ossona de Mendez}, Patrice and Pilipczuk, Micha{\l} and Quiroz, Daniel A. and Rabinovich, Roman and Siebertz, Sebastian},
  title = {Model-Checking on Ordered Structures},
  journal = {ACM Transactions on Computational Logic},
  volume = {21},
  number = {2},
  pages = {11:1--11:28},
  year = {2020},
  xdoi = {10.1145/3360011}
}

@article{JoretWood2013,
  author = {Gwena{\"e}l Joret and David R. Wood},
  title = {Complete Graph Minors and the Graph Minor Structure Theorem},
  journal = {Journal of Combinatorial Theory, Series B},
  volume = {103},
  number = {1},
  pages = {61--74},
  year = {2013},
  xdoi = {10.1016/j.jctb.2012.09.001}
}

@article{ganian2010there,
  title = {Are there any good digraph width measures?},
  author = {Ganian, Robert and Hlin{\v{e}}n{\'y}, Petr and Kneis, Joachim and Meister, Daniel and Obdr{\v{z}}{\'a}lek, Jan and Rossmanith, Peter and Sikdar, Somnath},
  journal = {Journal of Combinatorial Theory, Series B},
  volume = {116},
  pages = {250--286},
  year = {2016},
  publisher = {Elsevier},
  xdoi = {10.1016/j.jctb.2015.09.001}
}

@article{journals/combinatorics/Siebertz18,
  author = {Sebastian Siebertz},
  title = {Reconfiguration on Nowhere Dense Graph Classes},
  journal = {The Electronic Journal of Combinatorics},
  volume = {25},
  number = {3},
  pages = {P3.24},
  year = {2018},
  xdoi = {10.37236/7458}
}

@incollection{Thomas1997,
  author = {Wolfgang Thomas},
  title = {Languages, Automata, and Logic},
  booktitle = {Handbook of Formal Languages},
  volume = {3},
  publisher = {Springer},
  address = {Berlin},
  pages = {389--455},
  year = {1997},
  xdoi = {10.1007/978-3-642-59126-6_7}
}

@inproceedings{von2026separating,
  title = {Separating Feasibility and Movement in Solution Discovery: The Case of Path Discovery},
  author = {von Bergen, Hanno and Fastenau, Larissa and Gerhard, Enna and Lorenz, Nicola and Maaz, Stephanie and Mouawad, Amer E. and Rabinovich, Roman and Schirrmacher, Nicole and Schmand, Daniel and Siebertz, Sebastian and Trinh, Mai},
  year = {2026},
  booktitle = {51st International Symposium on Mathematical Foundations of Computer Science (MFCS 2026)},
  series = {Leibniz International Proceedings in Informatics},
  volume = {386},
  pages = {82:1--82:14},
  publisher = {Schloss Dagstuhl -- Leibniz-Zentrum f{\"u}r Informatik},
  xdoi = {10.4230/LIPIcs.MFCS.2026.82}
}

@inproceedings{SauST25para,
  author = {Ignasi Sau and
                  Giannos Stamoulis and
                  Dimitrios M. Thilikos},
  title = {Parameterizing the quantification of {CMSO:} model checking on minor-closed
                  graph classes},
  booktitle = {Proceedings of the 2025 Annual {ACM-SIAM} Symposium on Discrete Algorithms (SODA)},
  pages = {3728--3742},
  publisher = {{SIAM}},
  year = {2025},
  xdoi = {10.1137/1.9781611978322.124}
}

@inproceedings{SauSSSTV26,
  author = {Sau, Ignasi and Schirrmacher, Nicole and Siebertz, Sebastian and Stamoulis, Giannos and Thilikos, Dimitrios M. and Vigny, Alexandre},
  title = {Model Checking for Low Monodimensionality Fragments of {CMSO} on Topological-Minor-Free Graph Classes},
  booktitle = {41st Annual Symposium on Logic in Computer Science (LICS 2026)},
  series = {Leibniz International Proceedings in Informatics (LIPIcs)},
  volume = {380},
  pages = {80:1--80:27},
  year = {2026},
  publisher = {Schloss Dagstuhl -- Leibniz-Zentrum f{\"u}r Informatik},
  xdoi = {10.4230/LIPIcs.LICS.2026.80}
}

@article{gajarsky2020first,
  title = {First-order interpretations of bounded expansion classes},
  author = {Gajarsk{\'y}, Jakub and Kreutzer, Stephan and Ne{\v{s}}et{\v{r}}il, Jaroslav and {Ossona de Mendez}, Patrice and Pilipczuk, Micha{\l} and Siebertz, Sebastian and Toru{\'n}czyk, Szymon},
  journal = {ACM Transactions on Computational Logic},
  volume = {21},
  number = {4},
  pages = {29:1--29:41},
  year = {2020},
  publisher = {ACM New York, NY, USA},
  xdoi = {10.1145/3382093}
}

@article{feferman1959first,
  author = {Feferman, Solomon and Vaught, Robert L.},
  title = {The first order properties of products of algebraic systems},
  journal = {Fundamenta Mathematicae},
  volume = {47},
  number = {1},
  pages = {57--103},
  year = {1959},
  xdoi = {10.4064/fm-47-1-57-103}
}

@article{robertson2003graph,
  title = {{Graph minors. XVI. Excluding a non-planar graph}},
  author = {Robertson, Neil and Seymour, Paul D.},
  journal = {Journal of Combinatorial Theory, Series B},
  volume = {89},
  number = {1},
  pages = {43--76},
  year = {2003},
  publisher = {Elsevier},
  xdoi = {10.1016/S0095-8956(03)00042-X}
}

@misc{buffiere2026characterizations,
  title = {Characterizations of monadically dependent tree-ordered weakly sparse structures},
  author = {Buffi{\`e}re, Hector and Lin, Yuquan and Ne{\v{s}}et{\v{r}}il, Jaroslav and {Ossona de Mendez}, Patrice and Siebertz, Sebastian},
  year = {2026},
  eprint = {2601.16039},
  archiveprefix = {arXiv},
  xdoi = {10.48550/arXiv.2601.16039}
}

@article{schirrmacher2023first,
  title = {First-order logic with connectivity operators},
  author = {Schirrmacher, Nicole and Siebertz, Sebastian and Vigny, Alexandre},
  journal = {ACM Transactions on Computational Logic},
  volume = {24},
  number = {4},
  pages = {30:1--30:23},
  year = {2023},
  publisher = {ACM New York, NY, USA},
  xdoi = {10.1145/3595922}
}

@inproceedings{pillipczuk2021connectivity,
  author = {Micha{\l} Pilipczuk and Nicole Schirrmacher and Sebastian Siebertz and Szymon Toru{\'n}czyk and Alexandre Vigny},
  title = {Algorithms and Data Structures for First-Order Logic with Connectivity
                  Under Vertex Failures},
  booktitle = {49th International Colloquium on Automata, Languages, and Programming,
                  {ICALP} 2022},
  series = {Leibniz International Proceedings in Informatics (LIPIcs)},
  volume = {229},
  pages = {102:1--102:18},
  publisher = {Schloss Dagstuhl - Leibniz-Zentrum f{\"{u}}r Informatik},
  year = {2022},
  xdoi = {10.4230/LIPIcs.ICALP.2022.102}
}

@inproceedings{GroheKawarabayashiReed2013,
  author = {Martin Grohe and
               Ken{-}ichi Kawarabayashi and
               Bruce Reed},
  title = {A Simple Algorithm for the Graph Minor Decomposition -- Logic Meets Structural Graph Theory},
  booktitle = {Proceedings of the Twenty-Fourth Annual ACM-SIAM Symposium on Discrete Algorithms (SODA)},
  pages = {414--431},
  publisher = {{SIAM}},
  year = {2013},
  xdoi = {10.1137/1.9781611973105.30}
}

@article{DujmovicMorinWood,
  author = {Vida Dujmovi{\'c} and Pat Morin and David R. Wood},
  title = {Layered Separators in Minor-Closed Graph Classes with Applications},
  journal = {Journal of Combinatorial Theory, Series B},
  volume = {127},
  pages = {111--147},
  year = {2017},
  xdoi = {10.1016/j.jctb.2017.05.006}
}

@mastersthesis{MilolidakisThesis,
  author = {Orestis C. Milolidakis},
  title = {Splittability within Minor-Closed Classes to Graphs of Low Maximum Degree},
  year = {2025},
  school = {National and Kapodistrian University of Athens}
}

@inproceedings{paszke2020vc,
  author = {Adam Paszke and Micha{\l} Pilipczuk},
  title = {{VC Density of Set Systems Definable in Tree-Like Graphs}},
  booktitle = {45th International Symposium on Mathematical Foundations
               of Computer Science (MFCS 2020)},
  series = {Leibniz International Proceedings in Informatics (LIPIcs)},
  volume = {170},
  pages = {78:1--78:13},
  year = {2020},
  publisher = {Schloss Dagstuhl -- Leibniz-Zentrum f{\"u}r Informatik},
  xdoi = {10.4230/LIPIcs.MFCS.2020.78}
}

@inproceedings{bojanczyk2016definability,
  title = {Definability equals recognizability for graphs of bounded treewidth},
  author = {Boja{\'n}czyk, Miko{\l}aj and Pilipczuk, Micha{\l}},
  booktitle = {Proceedings of the 31st Annual ACM/IEEE Symposium on Logic in Computer Science (LICS 2016)},
  pages = {407--416},
  year = {2016},
  xdoi = {10.1145/2933575.2934508}
}

@inproceedings{Colcombet2007,
  author = {Thomas Colcombet},
  title = {A Combinatorial Theorem for Trees: Applications to Monadic Logic and Infinite Structures},
  booktitle = {Automata, Languages and Programming -- 34th International Colloquium (ICALP 2007)},
  series = {Lecture Notes in Computer Science},
  pages = {901--912},
  publisher = {Springer},
  year = {2007},
  xdoi = {10.1007/978-3-540-73420-8_77},
  volume = {4596}
}

@article{courcelle1992monadic,
  title = {The monadic second-order logic of graphs {VII}: {G}raphs as relational structures},
  author = {Courcelle, Bruno},
  journal = {Theoretical Computer Science},
  volume = {101},
  number = {1},
  pages = {3--33},
  year = {1992},
  publisher = {Elsevier},
  xdoi = {10.1016/0304-3975(92)90148-9}
}

@book{libkin2004elements,
  title = {Elements of finite model theory},
  author = {Libkin, Leonid},
  year = {2004},
  publisher = {Springer},
  series = {Texts in Theoretical Computer Science. An EATCS Series},
  xdoi = {10.1007/978-3-662-07003-1}
}

@article{bonnet2024twin,
  title = {Twin-width {IV}: ordered graphs and matrices},
  author = {Bonnet, {\'E}douard and Giocanti, Ugo and {Ossona de Mendez}, Patrice and Simon, Pierre and Thomass{\'e}, St{\'e}phan and Toru{\'n}czyk, Szymon},
  journal = {Journal of the ACM},
  volume = {71},
  number = {3},
  pages = {21:1--21:45},
  year = {2024},
  publisher = {ACM New York, NY},
  xdoi = {10.1145/3651151}
}

@inproceedings{GroblerMMMRSS24,
  author = {Mario Grobler and
                  Stephanie Maaz and
                  Nicole Megow and
                  Amer E. Mouawad and
                  Vijayaragunathan Ramamoorthi and
                  Daniel Schmand and
                  Sebastian Siebertz},
  title = {Solution Discovery via Reconfiguration for Problems in {P}},
  booktitle = {51st International Colloquium on Automata, Languages, and Programming (ICALP 2024)},
  series = {Leibniz International Proceedings in Informatics (LIPIcs)},
  pages = {76:1--76:20},
  publisher = {Schloss Dagstuhl -- Leibniz-Zentrum f{\"{u}}r Informatik},
  year = {2024},
  volume = {297},
  xdoi = {10.4230/LIPIcs.ICALP.2024.76}
}

@inproceedings{GroblerMMNRS24,
  author = {Mario Grobler and
                  Stephanie Maaz and
                  Amer E. Mouawad and
                  Naomi Nishimura and
                  Vijayaragunathan Ramamoorthi and
                  Sebastian Siebertz},
  title = {Kernelization Complexity of Solution Discovery Problems},
  booktitle = {35th International Symposium on Algorithms and Computation (ISAAC 2024)},
  series = {Leibniz International Proceedings in Informatics (LIPIcs)},
  pages = {36:1--36:17},
  year = {2024},
  volume = {322},
  publisher = {Schloss Dagstuhl -- Leibniz-Zentrum f{\"u}r Informatik},
  xdoi = {10.4230/LIPIcs.ISAAC.2024.36}
}

@misc{bousquet2025algorithmic,
  title = {On Algorithmic Meta-Theorems for Solution Discovery: Tractability and Barriers},
  author = {Bousquet, Nicolas and Mouawad, Amer E. and Maaz, Stephanie and Nishimura, Naomi and Siebertz, Sebastian},
  year = {2025},
  eprint = {2510.17344},
  archiveprefix = {arXiv},
  xdoi = {10.48550/arXiv.2510.17344}
}

@inproceedings{saito2026solution,
  title = {Solution Discovery for Vertex Cover, Independent Set, Dominating Set, and Feedback Vertex Set},
  author = {Saito, Rin and Sommer, Anouk and Suga, Tatsuhiro and Suzuki, Takahiro and Tamura, Yuma},
  booktitle = {SOFSEM 2026: Theory and Practice of Computer Science -- 51st International Conference on Current Trends in Theory and Practice of Computer Science},
  pages = {432--446},
  year = {2026},
  series = {Lecture Notes in Computer Science},
  volume = {16448},
  publisher = {Springer},
  xdoi = {10.1007/978-3-032-17801-5_32}
}

@article{courcelle1990,
  author = {Bruno Courcelle},
  title = {The monadic second-order logic of graphs. {I}. {Recognizable sets of finite graphs}},
  journal = {Information and Computation},
  volume = {85},
  number = {1},
  pages = {12--75},
  year = {1990},
  xdoi = {10.1016/0890-5401(90)90043-H}
}

@article{courcelle1993,
  author = {Bruno Courcelle},
  title = {The monadic second-order logic of graphs. {III}. {Tree-decompositions, minors and complexity issues}},
  journal = {RAIRO -- Theoretical Informatics and Applications},
  volume = {26},
  number = {3},
  pages = {257--286},
  year = {1992},
  xdoi = {10.1051/ita/1992260302571}
}

@article{arnborg1991,
  author = {Stefan Arnborg and Jens Lagergren and Detlef Seese},
  title = {Easy problems for tree-decomposable graphs},
  journal = {Journal of Algorithms},
  volume = {12},
  number = {2},
  pages = {308--340},
  year = {1991},
  xdoi = {10.1016/0196-6774(91)90006-K}
}

@article{seese1996,
  author = {Detlef Seese},
  title = {Linear time computable problems and first-order descriptions},
  journal = {Mathematical Structures in Computer Science},
  volume = {6},
  number = {6},
  pages = {505--526},
  year = {1996},
  xdoi = {10.1017/S0960129500070079}
}

@article{frick2001,
  author = {Markus Frick and Martin Grohe},
  title = {Deciding first-order properties of locally tree-decomposable structures},
  journal = {Journal of the ACM},
  volume = {48},
  number = {6},
  pages = {1184--1206},
  year = {2001},
  xdoi = {10.1145/504794.504798}
}

@article{dvorak2013,
  author = {Zden{\v{e}}k Dvo{\v{r}}{\'a}k and Daniel Kr{\'a}l' and Robin Thomas},
  title = {Testing first-order properties for subclasses of sparse graphs},
  journal = {Journal of the ACM},
  volume = {60},
  number = {5},
  pages = {36:1--36:24},
  year = {2013},
  xdoi = {10.1145/2499483}
}

@article{grohe2017,
  author = {Martin Grohe and Stephan Kreutzer and Sebastian Siebertz},
  title = {Deciding first-order properties of nowhere dense graphs},
  journal = {Journal of the ACM},
  volume = {64},
  number = {3},
  pages = {17:1--17:32},
  year = {2017},
  xdoi = {10.1145/3051095}
}

@article{bonnet2022twwI,
  author = {Bonnet, {\'E}douard and Kim, Eun Jung and Thomass{\'e}, St{\'e}phan and Watrigant, R{\'e}mi},
  title = {Twin-width {I}: tractable {FO} model checking},
  journal = {Journal of the ACM},
  volume = {69},
  number = {1},
  pages = {3:1--3:46},
  year = {2022},
  xdoi = {10.1145/3486655},
  publisher = {ACM New York, NY}
}

@inproceedings{dreier2023structurally,
  author = {Jan Dreier and Nikolas M{\"a}hlmann and Sebastian Siebertz},
  title = {First-Order Model Checking on Structurally Sparse Graph Classes},
  booktitle = {Proc. of the 55th Annual ACM Symposium on Theory of Computing (STOC)},
  pages = {567--580},
  year = {2023},
  xdoi = {10.1145/3564246.3585186}
}

@inproceedings{dreier2024stable,
  author = {Jan Dreier and Ioannis Eleftheriadis and Nikolas M{\"a}hlmann and Rose McCarty and Micha{\l} Pilipczuk and Szymon Toru{\'n}czyk},
  title = {First-Order Model Checking on Monadically Stable Graph Classes},
  booktitle = {Proc. of the 65th IEEE Annual Symposium on Foundations of Computer Science (FOCS)},
  pages = {21--30},
  year = {2024},
  xdoi = {10.1109/FOCS61266.2024.00012}
}

@inproceedings{dreier2024flip,
  author = {Jan Dreier and Nikolas M{\"a}hlmann and Szymon Toru{\'n}czyk},
  title = {Flip-Breakability: A Combinatorial Dichotomy for Monadically Dependent Graph Classes},
  booktitle = {Proc. of the 56th Annual ACM Symposium on Theory of Computing (STOC)},
  pages = {1550--1560},
  year = {2024},
  xdoi = {10.1145/3618260.3649739}
}

@misc{pilipczuk2025survey,
  author = {Micha{\l} Pilipczuk},
  title = {Graph classes through the lens of logic},
  year = {2025},
  archiveprefix = {arXiv},
  eprint = {2501.04166}
}

@misc{bojanczyk2021separator,
  author = {Miko{\l}aj Boja{\'n}czyk},
  title = {Separator logic and star-free expressions for graphs},
  year = {2021},
  archiveprefix = {arXiv},
  eprint = {2107.13953}
}

@inproceedings{golovach2022dpminor,
  author = {Petr A. Golovach and Giannos Stamoulis and Dimitrios M. Thilikos},
  title = {Model-Checking for First-Order Logic with Disjoint Paths Predicates in Proper Minor-Closed Graph Classes},
  booktitle = {Proc. of the 2023 Annual ACM-SIAM Symposium on Discrete Algorithms (SODA)},
  year = {2023},
  pages = {3684--3699},
  xdoi = {10.1137/1.9781611977554.ch141},
  publisher = {SIAM}
}

@inproceedings{schirrmacher2024topodp,
  author = {Nicole Schirrmacher and Sebastian Siebertz and Giannos Stamoulis and Dimitrios M. Thilikos and Alexandre Vigny},
  title = {Model Checking Disjoint-Paths Logic on Topological-Minor-Free Graph Classes},
  booktitle = {Proc. of the 39th Annual ACM/IEEE Symposium on Logic in Computer Science (LICS)},
  year = {2024},
  pages = {68:1--68:12},
  articleno = {68},
  numpages = {12},
  publisher = {ACM},
  xdoi = {10.1145/3661814.3662089}
}

@inproceedings{bojanczyk2025lowrank,
  author = {Miko{\l}aj Boja{\'n}czyk and Micha{\l} Pilipczuk and Wojciech Przybyszewski and Marek Soko{\l}owski and Giannos Stamoulis},
  title = {Low rank {MSO}},
  year = {2026},
  booktitle = {41st Annual Symposium on Logic in Computer Science (LICS 2026)},
  series = {Leibniz International Proceedings in Informatics},
  volume = {380},
  pages = {22:1--22:21},
  publisher = {Schloss Dagstuhl -- Leibniz-Zentrum f{\"u}r Informatik},
  xdoi = {10.4230/LIPIcs.LICS.2026.22}
}

@article{fellows2025onsd,
  title = {On solution discovery via reconfiguration},
  journal = {Journal of Computer and System Sciences},
  volume = {157},
  pages = {103747},
  year = {2026},
  author = {Fellows, Michael R. and Grobler, Mario and Megow, Nicole and Mouawad, Amer E. and Ramamoorthi, Vijayaragunathan and Rosamond, Frances A. and Schmand, Daniel and Siebertz, Sebastian},
  xdoi = {10.1016/j.jcss.2025.103747},
  publisher = {Elsevier}
}

@inproceedings{cruz1987,
  author = {Isabel F. Cruz and Alberto O. Mendelzon and Peter T. Wood},
  title = {A Graphical Query Language Supporting Recursion},
  booktitle = {Proc. of the ACM SIGMOD International Conference on Management of Data (SIGMOD)},
  pages = {323--330},
  year = {1987},
  publisher = {ACM},
  xdoi = {10.1145/38713.38749}
}

@article{mendelzon1995,
  author = {Alberto O. Mendelzon and Peter T. Wood},
  title = {Finding Regular Simple Paths in Graph Databases},
  journal = {SIAM Journal on Computing},
  volume = {24},
  number = {6},
  pages = {1235--1258},
  year = {1995},
  xdoi = {10.1137/S009753979122370X}
}

@article{barcelo2012,
  author = {Pablo Barcel{\'o} and Leonid Libkin and Anthony W. Lin and Peter T. Wood},
  title = {Expressive Languages for Path Queries over Graph-Structured Data},
  journal = {ACM Transactions on Database Systems},
  volume = {37},
  number = {4},
  pages = {31:1--31:46},
  year = {2012},
  xdoi = {10.1145/2389241.2389250}
}

@article{libkin2012,
  author = {Leonid Libkin and Wim Martens and Domagoj Vrgo{\v{c}}},
  title = {Querying Graphs with Data},
  journal = {Journal of the ACM},
  volume = {63},
  number = {2},
  pages = {14:1--14:53},
  year = {2016},
  xdoi = {10.1145/2850413}
}

@article{freydenschweik2013,
  author = {Dominik D. Freydenberger and Nicole Schweikardt},
  title = {Expressiveness and static analysis of extended conjunctive regular path queries},
  journal = {Journal of Computer and System Sciences},
  volume = {79},
  number = {6},
  pages = {892--909},
  year = {2013},
  xdoi = {10.1016/j.jcss.2013.01.008}
}

@article{wood2012,
  author = {Peter T. Wood},
  title = {Query languages for graph databases},
  journal = {ACM SIGMOD Record},
  volume = {41},
  number = {1},
  pages = {50--60},
  year = {2012},
  xdoi = {10.1145/2206869.2206879}
}

@inproceedings{salvati2024,
  author = {Salvati, Sylvain and Tison, Sophie},
  title = {{Containment of Regular Path Queries Under Path Constraints}},
  booktitle = {27th International Conference on Database Theory (ICDT 2024)},
  pages = {17:1--17:19},
  series = {Leibniz International Proceedings in Informatics},
  year = {2024},
  volume = {290},
  xdoi = {10.4230/LIPIcs.ICDT.2024.17},
  publisher = {Schloss Dagstuhl -- Leibniz-Zentrum f{\"u}r Informatik}
}

@article{baldwin1985second,
  title = {Second-order quantifiers and the complexity of theories},
  author = {Baldwin, John T. and Shelah, Saharon},
  journal = {Notre Dame Journal of Formal Logic},
  volume = {26},
  number = {3},
  pages = {229--303},
  year = {1985},
  publisher = {Duke University Press},
  xdoi = {10.1305/ndjfl/1093870870}
}

@article{ganian2019shrub,
  title = {Shrub-depth: Capturing height of dense graphs},
  author = {Ganian, Robert and Hlin{\v{e}}n{\'y}, Petr and Ne{\v{s}}et{\v{r}}il, Jaroslav and Obdr{\v{z}}{\'a}lek, Jan and {Ossona de Mendez}, Patrice},
  journal = {Logical Methods in Computer Science},
  volume = {15},
  number = {1},
  pages = {7:1--7:25},
  year = {2019},
  xdoi = {10.23638/LMCS-15(1:7)2019}
}

@article{ganian2012trees,
  title = {Shrub-depth: Capturing height of dense graphs},
  author = {Ganian, Robert and Hlin{\v{e}}n{\'y}, Petr and Ne{\v{s}}et{\v{r}}il, Jaroslav and Obdr{\v{z}}{\'a}lek, Jan and {Ossona de Mendez}, Patrice},
  journal = {Logical Methods in Computer Science},
  volume = {15},
  number = {1},
  pages = {7:1--7:25},
  year = {2019},
  xdoi = {10.23638/LMCS-15(1:7)2019}
}

@article{covers,
  author = {Braunfeld, Samuel and Ne{\v s}et{\v r}il, Jaroslav and {Ossona de Mendez}, Patrice and Siebertz, Sebastian},
  journal = {European Journal of Combinatorics},
  title = {Decomposition horizons and a characterization of stable hereditary classes of graphs},
  year = {2025},
  pages = {104130},
  xdoi = {10.1016/j.ejc.2025.104130},
  volume = {129}
}

@article{lmcs_perm,
  author = {Bonnet, {\'E}douard and Ne\v{s}et\v{r}il, Jaroslav and {Ossona de Mendez}, Patrice and Siebertz, Sebastian and Thomass{\'e}, St{\'e}phan},
  journal = {Logical Methods in Computer Science},
  title = {Twin-width and permutations},
  year = {2024},
  number = {3},
  volume = {20},
  pages = {4:1--4:25},
  xdoi = {10.46298/lmcs-20(3:4)2024}
}

@book{gross2018graph,
  title = {Graph theory and its applications},
  author = {Gross, Jonathan L. and Yellen, Jay and Anderson, Mark},
  year = {2018},
  publisher = {Chapman and Hall/CRC},
  edition = {3rd},
  xdoi = {10.1201/9780429425134}
}

@inproceedings{korhonen2024minor,
  title = {Minor containment and disjoint paths in almost-linear time},
  author = {Korhonen, Tuukka and Pilipczuk, Micha{\l} and Stamoulis, Giannos},
  booktitle = {Proc. of the 2024 IEEE 65th Annual Symposium on Foundations of Computer Science (FOCS)},
  pages = {53--61},
  year = {2024},
  publisher = {IEEE Computer Society},
  xdoi = {10.1109/FOCS61266.2024.00014}
}

@misc{dvorak2014list,
  title = {List-coloring Apex-minor-free Graphs},
  author = {Zden{\v e}k Dvo{\v r}{\'a}k and Robin Thomas},
  year = {2014},
  eprint = {1401.1399},
  archiveprefix = {arXiv},
  xdoi = {10.48550/arXiv.1401.1399},
  primaryclass = {math.CO}
}

@article{Th2,
  author = {Thomason, Andrew},
  journal = {Journal of Combinatorial Theory, Series B},
  title = {The extremal function for complete minors},
  year = {2001},
  number = {2},
  pages = {318--338},
  volume = {81},
  xdoi = {10.1006/jctb.2000.2013}
}

@article{braunfeld2022existential,
  author = {Braunfeld, Samuel and Laskowski, Michael C.},
  journal = {Communications of the American Mathematical Society},
  title = {Existential characterizations of monadic {NIP}},
  year = {2026},
  volume = {6},
  pages = {244--267},
  xdoi = {10.1090/cams/62}
}

@article{geniet2023first,
  author = {Geniet, Colin and Thomass{\'e}, St{\'e}phan},
  title = {First order logic and twin-width in tournaments and dense oriented graphs},
  journal = {European Journal of Combinatorics},
  volume = {132},
  pages = {104247},
  year = {2026},
  xdoi = {10.1016/j.ejc.2025.104247}
}

@misc{mergemodels,
  title = {On merge-models},
  author = {Buffi{\`e}re, Hector and Lin, Yuquan and Ne{\v{s}}et{\v{r}}il, Jaroslav and {Ossona de Mendez}, Patrice and Siebertz, Sebastian},
  year = {2026},
  eprint = {2603.26570},
  archiveprefix = {arXiv},
  xdoi = {10.48550/arXiv.2603.26570}
}

@inproceedings{dreier2025merge,
  title = {Merge-width and first-order model checking},
  author = {Dreier, Jan and Toru{\'n}czyk, Szymon},
  booktitle = {Proceedings of the 57th Annual ACM Symposium on Theory of Computing},
  pages = {1944--1955},
  year = {2025},
  publisher = {ACM},
  xdoi = {10.1145/3717823.3718259}
}

@misc{simon2021note,
  title = {A note on stability and {NIP} in one variable},
  author = {Simon, Pierre},
  year = {2021},
  eprint = {2103.15799},
  archiveprefix = {arXiv},
  xdoi = {10.48550/arXiv.2103.15799}
}

@article{nishimura18reconfiguration,
  author = {Nishimura, Naomi},
  title = {Introduction to Reconfiguration},
  journal = {Algorithms},
  volume = {11},
  year = {2018},
  number = {4},
  issn = {1999-4893},
  xdoi = {10.3390/a11040052},
  pages = {52}
}

@article{cualinescu2008reconfigurations,
  title = {Reconfigurations in graphs and grids},
  author = {C{\u{a}}linescu, Gruia and Dumitrescu, Adrian and Pach, J{\'a}nos},
  journal = {SIAM Journal on Discrete Mathematics},
  volume = {22},
  number = {1},
  pages = {124--138},
  year = {2008},
  publisher = {SIAM},
  xdoi = {10.1137/060652063}
}

@article{cooper2024parameterized,
  title = {Parameterized complexity of reconfiguration of atoms},
  author = {Cooper, Alexandre and Maaz, Stephanie and Mouawad, Amer E. and Nishimura, Naomi},
  journal = {Algorithmica},
  volume = {86},
  number = {10},
  pages = {3284--3308},
  year = {2024},
  publisher = {Springer},
  xdoi = {10.1007/s00453-024-01262-z}
}

@inproceedings{bonnet2022twin,
  title = {Twin-width {VI}: the lens of contraction sequences},
  author = {Bonnet, {\'E}douard and Kim, Eun Jung and Reinald, Amadeus and Thomass{\'e}, St{\'e}phan},
  booktitle = {Proceedings of the 2022 Annual ACM-SIAM Symposium on Discrete Algorithms (SODA)},
  pages = {1036--1056},
  year = {2022},
  xdoi = {10.1137/1.9781611977073.45},
  publisher = {SIAM}
}

@article{BonnetHickingbotham2025,
  author = {Bonnet, {\'E}douard and Hickingbotham, Robert},
  title = {Induced Minors and Region Intersection Graphs},
  journal = {Innovations in Graph Theory},
  volume = {2},
  pages = {313--327},
  year = {2025},
  xdoi = {10.5802/igt.15}
}

@article{bonnet2022twwII,
  author = {Bonnet, {\'E}douard and Geniet, Colin and Kim, Eun Jung and Thomass{\'e}, St{\'e}phan and Watrigant, R{\'e}mi},
  title = {Twin-width {II}: small classes},
  journal = {Combinatorial Theory},
  volume = {2},
  number = {2},
  pages = {10:1--10:42},
  year = {2022},
  xdoi = {10.5070/C62257876}
}

@article{RobertsonSGM5,
  title = {Graph minors. {V}. Excluding a planar graph},
  journal = {Journal of Combinatorial Theory, Series B},
  volume = {41},
  number = {1},
  pages = {92--114},
  year = {1986},
  issn = {0095-8956},
  xdoi = {10.1016/0095-8956(86)90030-4},
  author = {Neil Robertson and P. D. Seymour}
}

@article{ADLER2014322,
  title = {Interpreting nowhere dense graph classes as a classical notion of model theory},
  journal = {European Journal of Combinatorics},
  volume = {36},
  pages = {322--330},
  year = {2014},
  issn = {0195-6698},
  xdoi = {10.1016/j.ejc.2013.06.048},
  author = {Hans Adler and Isolde Adler}
}

@article{grohe2004learnability,
  title = {Learnability and definability in trees and similar structures},
  author = {Grohe, Martin and Tur{\'a}n, Gy.},
  journal = {Theory of Computing Systems},
  volume = {37},
  number = {1},
  pages = {193--220},
  year = {2004},
  publisher = {Springer},
  xdoi = {10.1007/s00224-003-1112-8}
}

@inproceedings{flip-flatness,
  author = {Dreier, Jan and M\"{a}hlmann, Nikolas and Siebertz, Sebastian and Toru\'{n}czyk, Szymon},
  title = {{Indiscernibles and Flatness in Monadically Stable and Monadically NIP Classes}},
  booktitle = {50th International Colloquium on Automata, Languages, and Programming (ICALP 2023)},
  isbn = {978-3-95977-278-5},
  issn = {1868-8969},
  year = {2023},
  volume = {261},
  urn = {urn:nbn:de:0030-drops-181779},
  xdoi = {10.4230/LIPIcs.ICALP.2023.125},
  pages = {125:1--125:18},
  series = {Leibniz International Proceedings in Informatics (LIPIcs)},
  publisher = {Schloss Dagstuhl -- Leibniz-Zentrum f{\"u}r Informatik}
}

\clearpage
\appendix

\section{Dependence of separator logic on monotone graph classes excluding a topological minor}
\label{app:separator-characterization}

\getkeytheorem{thm:separator-characterization}

For $(1)\Rightarrow(2)$, graphs in $\Cc$ admit tree decompositions that can be encoded by rooted trees equipped with their ancestor order and auxiliary structures on sibling sets. For each fixed $\FOconn$-formula, Lemma~8.3 and Remark~7.1 of the full version of~\cite{pillipczuk2021connectivity} provide such an encoding, with finitely many unary colors, on which the formula is equivalent to a plain first-order formula, preserving all assignments to its free variables. The auxiliary structures belong to a fixed bounded-expansion class and are therefore monadically dependent.

More generally, trees augmented in this way with structures from any fixed monadically dependent class are monadically dependent. 
The counting argument follows the same principle as the tree-automata argument of Paszke and Pilipczuk~\cite[Section~3]{paszke2020vc}, who bound the VC density of definable set systems on trees and, through tree encodings, on graphs of bounded clique-width.

We briefly sketch the argument using Theorem~7.1 of the full version of~\cite{pillipczuk2021connectivity}. 
Every first-order formula on these augmented trees can be translated into an equivalent automaton, independently of any structural assumptions on the auxiliary structures. 
The automaton depends only on the formula, its transitions are specified by first-order formulas evaluated on the auxiliary structures colored by the states already computed at their vertices.

By the standard one-object-variable criterion for NIP~\cite{simon2021note}, it suffices to consider formulas $\psi(x;\bar y)$ with $x$ a single variable and $\bar y$ an arbitrary tuple. 
Note that this criterion also applies to classes of finite structures. 
Fix such a formula, its automaton, and a nonempty finite parameter set~$A$. 
The $\psi$-type of a vertex over~$A$ records precisely which tuples from $A^{|\bar y|}$ make the automaton accept when that vertex is marked as~$x$ and the tuple is marked as~$\bar y$. 
Thus, counting types amounts to counting the different acceptance patterns obtained by moving the $x$-mark.

Consider the union of the paths from the root to~$A$. Retaining the root and the vertices of~$A$ and suppressing all other vertices with only one child leaves $\Oof(|A|)$ vertices and connecting paths. 
Color the auxiliary structures by the states of the automaton run with no free-variable markers. 
Along each connecting path, the effect of the $x$-mark is described by one of finitely many transformations of automaton states. 
At a retained vertex, a parameter tuple changes the states of at most $|\bar y|$ children. 
The interaction with the child containing the $x$-mark is therefore described by finitely many first-order formulas on the colored auxiliary structure. 
Monadic dependence bounds the number of resulting patterns over the children leading to~$A$ by $\Oof(|A|^d)$, where $d$ depends only on the formula and the class of auxiliary structures. 
These local patterns determine the state passed towards the root and hence acceptance for every parameter tuple. 
Including the $\Oof(|A|)$ possibilities where $x$ is itself a retained vertex and summing over all vertices and paths gives $\Oof(|A|^{d+1})$ possible $\psi$-types over $A^{|\bar y|}$.

The argument allows arbitrary unary colors and therefore establishes monadic dependence of the augmented trees. 
Applying this to the encodings above proves $(1)\Rightarrow(2)$. 
The implication $(2)\Rightarrow(3)$ is immediate.

For $\neg(1)\Rightarrow\neg(3)$, we recall the subdivision argument underlying Lemma~9.1 of the same paper. 
Let $H$ be any graph and attach three private leaves to each vertex, obtaining~$H^+$. By monotonicity, $\Cc$ contains a subdivision $S$ of~$H^+$. The original vertices of~$H$ are precisely the vertices of degree at least three in~$S$, subdivision vertices have degree two, and the added leaves have degree one. 
Two distinct original vertices $x$ and $y$ are adjacent in~$H$ if and only if they are adjacent in~$S$, or a component of $S-\{x,y\}$ consists entirely of vertices of degree two in~$S$ and contains a neighbor of each of $x$ and~$y$. 
This is expressible by a fixed $\FOconn$ formula. 
Applying it to subdivisions of incidence graphs realizing arbitrary finite incidence patterns proves independence.

\end{document}